\documentclass[sigconf]{acmart}
\AtBeginDocument{%
  \providecommand\BibTeX{{%
    \normalfont B\kern-0.5em{\scshape i\kern-0.25em b}\kern-0.8em\TeX}}}

\usepackage{amsmath,amsfonts}
\usepackage{algorithmic}
\usepackage{graphicx}
\usepackage{textcomp}
\usepackage{cleveref}
\usepackage{amsthm}
\usepackage{mdframed}
\usepackage{graphicx}
\usepackage[colorinlistoftodos]{todonotes}
\usepackage{verbatim} 
\usepackage{fetamont}
\usetikzlibrary{patterns}
\usepackage{caption, subcaption}
\usepackage{xspace}
\usepackage[nolist]{acronym}
\usepackage{ifthen}
\usepackage{braket}
\usepackage[nolist]{acronym}
\usepackage{tikz}
\usepackage{pgfplots}\pgfplotsset{compat=1.17}
\usepackage{nicefrac}

\usepackage{comment}

\newtheorem{thm}{Theorem}
\newtheorem{lem}{Lemma}
\newtheorem{cor}[thm]{Corollary}

\newtheorem{claim}{Claim}

\usetikzlibrary{shapes.geometric, arrows.meta, graphs, matrix}
\tikzset{%
	>={Latex[width=2mm,length=2mm]},
	base/.style = {rectangle, rounded corners, draw=black,
		minimum width=4cm, minimum height=1cm,
		text centered, font=\sffamily},
	activityStarts/.style = {base, fill=blue!30},
	startstop/.style = {base, fill=red!30},
	activityRuns/.style = {base, fill=subbi!30},
	process/.style = {base, minimum width=2.5cm, fill=orange!15,
		font=\ttfamily},
}
\tikzset{toprule/.style={%
		execute at end cell={%
			\draw [line cap=rect,#1] (\tikzmatrixname-\the\pgfmatrixcurrentrow-\the\pgfmatrixcurrentcolumn.north west) -- (\tikzmatrixname-\the\pgfmatrixcurrentrow-\the\pgfmatrixcurrentcolumn.north east);%
		}
	},
	bottomrule/.style={%
		execute at end cell={%
			\draw [line cap=rect,#1] (\tikzmatrixname-\the\pgfmatrixcurrentrow-\the\pgfmatrixcurrentcolumn.south west) -- (\tikzmatrixname-\the\pgfmatrixcurrentrow-\the\pgfmatrixcurrentcolumn.south east);%
		}
	}
}
\tikzstyle{startstop} = [rectangle, rounded corners, minimum width=3cm, minimum height=1cm,text centered, draw=black, fill=red!30]
\tikzstyle{io} = [trapezium, trapezium left angle=70, trapezium right angle=110, minimum width=3cm, minimum height=1cm, text centered, draw=black, fill=blue!30]
\tikzstyle{process} = [rectangle, minimum width=3cm, minimum height=1cm, text centered, draw=black, fill=orange!30]
\tikzstyle{decision} = [diamond, minimum width=3cm, minimum height=1cm, text centered, draw=black, fill=green!30]
\tikzstyle{arrow} = [thick,->,>=stealth]

\definecolor{subbi}{rgb}{0.53, 0.81, 1}
\definecolor{color2}{rgb}{0.95, 0.53, 1 }
\definecolor{colora3}{rgb}{1, 0.72, 0.53 }
\definecolor{stripColor}{rgb}{0.57, 1, 0.53}
\definecolor{hugeItemColor}{rgb}{1, 0.53, 0.57 }
\definecolor{background}{rgb}{0.6,0.6,0.6}
\newcommand{\eps}{\varepsilon}

\newcommand{\opt}{\mathrm{OPT}}

\newcommand{\height}{h}
\newcommand{\width}{w}

\newcommand{\spo}{\lambda}

\newcommand{\items}{\mathcal{I}}

\newcommand*{\ari}{\textcolor{red}}
\renewcommand{\ari}[1]{#1}

\newcommand{\Oh}{\mathcal{O}}

\newboolean{BlackAndWhite}
\setboolean{BlackAndWhite}{true}
\definecolor{li}{HTML}{00677C}
\definecolor{wi}{HTML}{8E217D}
\definecolor{ti}{HTML}{F29400}
\definecolor{si}{HTML}{E43117}
\definecolor{vi}{HTML}{39842E}
\definecolor{hi}{HTML}{9B0a7d}
\definecolor{tbBgOdd}{rgb}{0.82,0.82,0.82}

\newcommand{\drawVerticalItem}[5][$ $]{%
	\ifthenelse{\boolean{BlackAndWhite}}{%
		\draw[fill = white!85!black, fill opacity = 0.7] (#2,#3) rectangle node[midway, opacity = 1]{#1} (#4,#5)}{%
		\draw[fill = white!50!vi, fill opacity = 0.7] (#2,#3) rectangle node[midway, opacity = 1]{#1} (#4,#5)}}
\newcommand{\drawVerticalItemRotate}[5][$ $]{%
	\ifthenelse{\boolean{BlackAndWhite}}{%
		\draw[fill = white!85!black, fill opacity = 0.7] (#2,#3) rectangle node[midway, opacity = 1, rotate=90]{#1} (#4,#5)}{%
		\draw[fill = white!50!vi, fill opacity = 0.7] (#2,#3) rectangle node[midway, opacity = 1,rotate=90]{#1} (#4,#5)}}		
\newcommand{\drawTallItem}[5][$ $]{%
	\ifthenelse{\boolean{BlackAndWhite}}{%
		\draw[fill = white!65!black, fill opacity = 0.7] (#2,#3) rectangle node[midway, opacity = 1]{#1} (#4,#5)}{%
		\draw[fill = white!50!ti, fill opacity = 0.7] (#2,#3) rectangle node[midway, opacity = 1]{#1} (#4,#5)}}
\newcommand{\drawLargeItem}[5][$ $]{	\ifthenelse{\boolean{BlackAndWhite}}{%
		\draw[fill = white!55!black, fill opacity = 0.7] (#2,#3) rectangle node[midway, opacity = 1]{#1} (#4,#5)}{%
		\draw[fill = white!50!li, fill opacity = 0.7] (#2,#3) rectangle node[midway, opacity = 1]{#1} (#4,#5)}}
\newcommand{\drawSmallItem}[5][$ $]{	\ifthenelse{\boolean{BlackAndWhite}}{%
		\draw[fill = white!95!black, fill opacity = 0.7] (#2,#3) rectangle node[midway, opacity = 1]{#1} (#4,#5)}{%
		\draw[fill = white!50!si, fill opacity = 0.7] (#2,#3) rectangle node[midway, opacity = 1]{#1} (#4,#5)}}
\newcommand{\drawHorizontalItem}[5][$ $]{	\ifthenelse{\boolean{BlackAndWhite}}{%
		\draw[fill = white!75!black, fill opacity = 0.7] (#2,#3) rectangle node[midway, opacity = 1]{#1} (#4,#5)}{%
		\draw[fill = white!50!hi, fill opacity = 0.7] (#2,#3) rectangle node[midway, opacity = 1]{#1} (#4,#5)}}
\newcommand{\drawJobNoBorder}[5][$ $]{	\ifthenelse{\boolean{BlackAndWhite}}{%
    	\draw[white!90!black, fill, fill opacity = 0.7] (#2,#3) rectangle node[midway, opacity = 1]{#1} (#4,#5)}{%
		\draw[white!50!vi, fill, fill opacity = 0.7] (#2,#3) rectangle node[midway, opacity = 1]{#1} (#4,#5)}}

\newcommand{\drawItem}[5][$ $]{\draw[fill = tbBgOdd, fill opacity = 0.6] (#2,#3) rectangle node[midway, opacity = 1]{#1} (#4,#5)}

\newcommand{\jobs}{\mathcal{J}}

\newcommand{\natNumbers}{\mathbb{N}}
\newcommand{\stripWidth}{W}
\newcommand{\stripHeight}{H}
\newcommand{\instance}{I}
\newcommand{\instanceTransform}{I^*}
\newcommand{\SP}{\textsc{Strip Packing}}
\newcommand{\SSP}{\textsc{Demand Strip Packing}}
\newcommand{\PTS}{\textsc{Parallel Task Scheduling}}
\newcommand{\startPoint}{\lambda}

\newcommand{\itemsLarge}{\mathcal{L}}
\newcommand{\itemsTall}{\mathcal{T}}
\newcommand{\itemsVert}{\mathcal{V}}
\newcommand{\itemsMedVert}{\mathcal{M}_v}
\newcommand{\itemsHor}{\mathcal{H}}
\newcommand{\itemsSmall}{\mathcal{S}}
\newcommand{\itemsMed}{\mathcal{M}}
\newcommand{\bigO}{\mathcal{O}}
\newcommand{\boxesLarge}{\mathcal{B}_L}
\newcommand{\boxesHor}{\mathcal{B}_H}
\newcommand{\boxesTallVert}{\mathcal{B}_{T\cup V}}
\newcommand{\boxesTall}{\mathcal{B}_T}
\newcommand{\boxesVert}{\mathcal{B}_V}
\newcommand{\boxesPseudo}{\mathcal{B}_P}
\newcommand{\boxesSmallVert}{\mathcal{B}_V^S}
\newcommand{\boxesHorSmall}{\mathcal{B}_H^S}
\newcommand{\boxesSmall}{\mathcal{B}_S}
\newcommand{\boxHalf}{\Tilde{B}}
\newcommand{\boxFourth}{\Breve{B}}
\newcommand{\countTall}{S_T}
\newcommand{\countPseudo}{S_P}
\newcommand{\countPseudoTall}{S_{P\cup T}}
\newcommand{\heightVertical}{H_\mathcal{V}}
\newcommand{\countStrips}{N_S}
\newcommand{\boxesVertPlace}{\mathcal{B}_{\mu\stripWidth}}
\newcommand{\countExtraBoxes}{N_F}
\newcommand{\countVertLines}{N_L}
\newcommand{\widthStripsTall}{W_T}
\newcommand{\widthStripsBoxes}{W_H}
\newcommand{\widthStripsRem}{W_R}
\newcommand{\countBoxesTallVert}{N_B}
\newcommand{\widthsHor}{\mathcal{W}_H}
\newcommand{\tallItemsHalf}{T_{\nicefrac{1}{2}\stripHeight}}

\crefname{lemma}{Lemma}{Lemmas}
\crefname{lem}{Lemma}{Lemmas}
\crefname{thm}{Theorem}{Theorems}
\crefname{figure}{Figure}{Figures}
\crefname{claim}{Claim}{Claims}
\crefname{table}{Table}{Tables}

\def\BibTeX{{\rm B\kern-.05em{\sc i\kern-.025em b}\kern-.08em
    T\kern-.1667em\lower.7ex\hbox{E}\kern-.125emX}}
    
\begin{document}

\title{Hardness and Tight Approximations of Demand Strip Packing}

\author{Klaus Jansen}
\email{kj@informatik.uni-kiel.de}
\affiliation{%
  \institution{Kiel University}
  \city{Kiel}
  \country{Germany}
}

\author{Malin Rau}
\affiliation{%
  \institution{Hamburg University}
  \city{Hamburg}
  \country{Germany}}
\email{malin.rau@uni-hamburg.de}

\author{Malte Tutas}
\affiliation{%
  \institution{Kiel University}
  \city{Kiel}
  \country{Germany}
}
\email{mtu@informatik.uni-kiel.de}

\renewcommand{\shortauthors}{Anonymous}

\begin{abstract}
  We settle the pseudo-polynomial complexity of the \SSP{} (DSP) problem: 
Given a strip of fixed width and a set of items with widths and heights, the items must be placed inside the strip with the objective of minimizing the peak height.  
This problem has gained significant scientific interest due to its relevance in smart grids~[Deppert~et~al.\ APPROX'21, G\'alvez~et~al.\ APPROX'21]. 
Smart Grids are a modern form of electrical grid that provide opportunities for optimization. 
They are forecast to impact the future of energy provision significantly. 
Algorithms running in pseudo-polynomial time lend themselves to these applications as considered time intervals, such as days, are small. 
Moreover, such algorithms can provide superior approximation guarantees over those running in polynomial time. 
Consequently, they evoke scientific interest in related problems~[Jansen and Rau ESA'19].

We prove that \SSP{} is strongly NP-hard for approximation ratios below~$5/4$. 
Through this proof, we provide novel insights into the relation of packing and scheduling problems. 
Using these insights, we show a series of frameworks that solve both \SSP{} and \PTS{} optimally when increasing the strip's width or number of machines. 
Such alterations to problems are known as resource augmentation.
For both problems, the use of resource augmentation is a new idea. 
Applications are found when penalty costs are prohibitively large. 
Finally, we provide a pseudo-polynomial time approximation algorithm for DSP with an approximation ratio of $(5/4+\eps)$, which is nearly optimal assuming $P\neq NP$. 
The construction of this algorithm provides several insights into the structure of DSP solutions and uses novel techniques to restructure optimal solutions.
\end{abstract}




\maketitle

\section{Introduction}


Increased adoption of technologies such as electric cars as well as global economic growth, lead to a growing demand for electric power. 
To better handle this demand, new technologies are developed. 
An example of such a  technology is the so-called Smart Grid~\cite{KarbasiounSLK18,SIANO2014461,elecApp2}. 
Traditional power systems have to incorporate infrastructure to support the peak load on the entirety of the system~\cite{Torriti2016}. 
However, this peak demand is seldom reached, meaning that, at non-peak times, there is some capacity for additional power demands in the power grid. 
Smart~Grids use automated digital communication between power consumers and electricity providers to utilize the available capacity more efficiently by, for example, shifting a user's power-demanding tasks onto off-peak hours. 
This directly results in a decreased peak load on the entire system. 
The benefits of this are plentiful. 
For the consumer, modern devices can automatically shift their power usage to off-peak hours to avoid surcharges through congestion pricing~\cite{DILEEP20202589}. 
On the other hand, for electricity suppliers, the utilization of Smart~Grids and the accompanying reduction in peak load reduces the required infrastructure to supply the ever-increasing electricity demand across the globe. 
It has been estimated that implementation of these Smart~Grids across the United States alone leads to savings between $46$ and $117$~billion US-Dollars over 20 years~\cite{kannberg2004gridwisetm}. 

To effectively generate these benefits, efficient algorithms to balance the load on the electrical system are required. 
We can model the power consumption of individual applications through a rectangle. 
The width of this rectangle then represents the duration for which the application consumes power, and its height represents the amount of power it consumes. 
Modeling individual demands in this way lets us inspect larger time intervals over which tasks may recur. 
We can model these time intervals as a single strip with a set width, where the width represents the length of the interval. The time intervals in which tasks repeat are usually rather short, i.e.\ days. Furthermore, when thinking of appliances such as washing machines, their operating time can be quantified in minutes.
Finding a packing of all rectangles into this strip with a minimum height then represents a minimization of peak load on the electrical grid. 
However, the load on the system is the sum of all demands at a certain point in time, which we can not model through these simple rectangles yet. To compensate we introduce the concept of slicing, which allows us to place vertical cuts inside any generated item and place these cut parts at any height as long as they are placed contiguously in their width. 
This modeling motivates the problem studied in this paper.
\begin{figure*}[t]
\centering
	\begin{subfigure}{0.35\textwidth}
		\centering
		\resizebox{0.98\textwidth}{!}{
	\begin{tikzpicture}
\pgfmathsetmacro{\w}{0.5}
\pgfmathsetmacro{\h}{0.5}
\pgfmathsetmacro{\boxh}{6*\h}
\draw (0,0) -- (9*\w,0);
\draw[-stealth] (0,0) -- (0,\boxh);
\draw[-stealth] (9*\w,0) -- (9*\w,\boxh);
\draw (-0.05*9*\w,4*\h) node [anchor=east] {$4$}-- (0,4*\h);
\draw (-0.05*9*\w,5*\h) node [anchor=east] {$5$}-- (0,5*\h);
\draw [dotted] (0,5*\h) -- (9*\w,5*\h);
\drawLargeItem[$a$]{0}{0}{3*\w}{3*\h};
\drawHorizontalItem[$b$]{3*\w}{0}{8*\w}{1*\h};
\drawHorizontalItem[$c$]{0}{3*\h}{7*\w}{4*\h};
\drawLargeItem[$d$]{3*\w}{1*\h}{6*\w}{3*\h};
\drawSmallItem[$e$]{6*\w}{2*\h}{7*\w}{3*\h};
\drawSmallItem[$f$]{6*\w}{1*\h}{9*\w}{2*\h};
\drawVerticalItem[$g$]{7*\w}{2*\h}{8*\w}{4*\h};
\drawVerticalItem[$h$]{8*\w}{2*\h}{9*\w}{5*\h};
\end{tikzpicture}
		}	
	\end{subfigure}
		\begin{subfigure}{0.35\textwidth}
		\centering
		\resizebox{0.98\textwidth}{!}{
	\begin{tikzpicture}
\pgfmathsetmacro{\w}{0.5}
\pgfmathsetmacro{\h}{0.5}
\pgfmathsetmacro{\boxh}{6*\h}
\draw (0,0) -- (9*\w,0);
\draw[-stealth] (0,0) -- (0,\boxh);
\draw[-stealth] (9*\w,0) -- (9*\w,\boxh);
\draw (-0.05*9*\w,4*\h) node [anchor=east] {$4$}-- (0,4*\h);
\draw (-0.05*9*\w,5*\h) node [anchor=east] {$5$}-- (0,5*\h);
\draw [dotted] (0,5*\h) -- (9*\w,5*\h);
\drawLargeItem[$a$]{0}{0}{3*\w}{3*\h};
\drawHorizontalItem[$b$]{3*\w}{0}{8*\w}{1*\h};
\drawHorizontalItem[$c$]{0}{3*\h}{7*\w}{4*\h};
\drawLargeItem[$d$]{3*\w}{1*\h}{6*\w}{3*\h};
\drawSmallItem[$e$]{6*\w}{2*\h}{7*\w}{3*\h};
\drawSmallItem[$f$]{6*\w}{1*\h}{8*\w}{2*\h};
\drawSmallItem[$f$]{8*\w}{0}{9*\w}{1*\h};
\drawVerticalItem[$g$]{7*\w}{2*\h}{8*\w}{4*\h};
\drawVerticalItem[$h$]{8*\w}{1*\h}{9*\w}{4*\h};
\end{tikzpicture}
		}	
	\end{subfigure}
\caption{An illustration of the difference slicing can make to the packing height. 
On the left, you can see the optimal solution for the classical strip packing problem, whereas the right displays the optimal packing in the demand variant. 
As you can see, the height of the packing decreased through slicing. 
The sliced item contains its name to easier identify it. 
We can see the difference in optimal heights of $\nicefrac{5}{4}$ here.}
\label{fig:DifferenceSP/SSP}
\end{figure*}
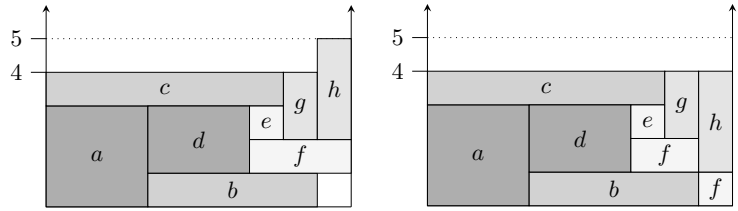
\paragraph*{Demand Strip Packing} 
We are given a set~$\items$ of $n$~items. 
Each item $i\in \items$ has a width $\width(i)$ and a height $\height(i)$. 
Further, we are given a strip~$S$ of width $\stripWidth$. 
The goal is to pack all items into $S$ such that no items overlap each other. 
Such a packing $\Lambda$ is defined by a placement function $\lambda: \items \rightarrow \natNumbers$ that assigns a starting point $\lambda(i)$ to each item. 
A starting point $\lambda(i)$ is only feasible if $0\leq \lambda(i)\leq \lambda(i)+\width(i)\leq \stripWidth$ holds, i.e., the items do not overlap the edges of the strip. 
All items must be placed contiguously, i.e., no part of the item is allowed to be placed after $\lambda(i)+\width(i)$. 
The objective is to minimize the peak height \ari{$\stripHeight:= \max_{k\in \{0,\ldots, \stripWidth\}}(\sum_{i\in \items: \lambda(i)+\width(i)=k} \height(i)),$}
i.e., the largest sum of heights for all items placed at a given point $k$.
Slicing an item refers to placing a vertical cut at any point inside the item. Either part of the item is then allowed to be placed at any height, see \cref{fig:DifferenceSP/SSP} for an illustration. 
This property allows for lower peak heights in the packing when compared to integral packings.

We study the pseudo-polynomial time setting. Pseudo-polynomiality refers to the fact that the algorithm is polynomial in the unary representation of the input. Applying this to DSP means that we are allowed to iterate over the width of the strip $\stripWidth$, which is why this setting lends itself especially well to the discussed and widespread application of smaller time intervals.


\paragraph*{Related Work}
\textsc{Demand Strip Packing} is an adaptation of the classical \SP{}(SP) problem.
SP was first studied in 1980, by Baker et al. They provide a 3-approximation~\cite{BakerCR80}. 
Study into this problem continued with Coffman et al.~\cite{CoffmanGJT80} providing two shelf-based algorithms. These algorithms achieve a ratio of $3$ and $2.7$, respectively. This was improved by Sleator~\cite{Sleator80}, who gives a 2.5 approximation. Research stalled until 1994, where Schiermeyer~\cite{Schiermeyer94} presented a 2 approximation. 
This ratio later got matched by Steinberg~\cite{Steinberg97}, whose algorithm is of particular interest because it only relies on the size of items to ensure its ratio. This makes the algorithm applicable to DSP as well. It took until twelve years later for further improvement to be found, when Harren and van\ Stee~\cite{HarrenS09} presented a $1.936$-approximation. The current best known result is a $(5/3+\eps)$-approximation, given by Harren et al.\ in 2014~\cite{HarrenJPS14}. 
Bladek et al.\ show that there exists a gap of $\nicefrac{5}{4}$ between the heights of optimal SP and DSP solutions~\cite{BladekDGS15}. This is commonly referred to as an integrality gap. 
As such, algorithms for the classical \SP{} problem produce solutions that are worse by a factor of up to $\nicefrac{5}{4}$ for the demand case.
Through a simple reduction from the \textsc{Partition}-problem, a hardness of $3/2$ can be shown for SP.

However, in pseudo-polynomial time, improved approximation algorithms and lower hardness results are achievable. 
Study into this setting began in 2010 by Jansen and Th\"ole~\cite{JansenT10}. 
They provide a $(3/2+\eps)$-approximation. This was later improved by Nadiradze and Wiese~\cite{NadiradzeW16}, who present a $(7/5+\eps)$-approximation. Research into this setting remained steady and the next improvement was made by G\'alvez et al.\ in the following year~\cite{GalvezGIK16}. 
They present a $(4/3+\eps)$-approximation. 
This ratio was matched by Jansen and Rau in the following year, who provide an improvement in running time for their algorithm~\cite{JansenR17}. 
The best approximation ratio known is given by Jansen and Rau~\cite{StripPacking54} as well, who provide a $(5/4+\eps)$-approximation. 
These improved approximations are accompanied by a lower theoretical hardness of $5/4$, shown by Henning et al.\ in 2018~\cite{HenningJRS18}. 

Variations of strip packing also became fields of thorough research. \SSP{} was first studied in 2013 by Tang et al., who present a $7$-approximation~\cite{TangS13}. In the following year, Yaw et al.\ provide an improved approximation of $4$ for cases where all items have the same width~\cite{yaw2014peak}. They also show that DSP is $NP$-hard to approximate for ratios below $3/2$~\cite{yaw2014peak}. In that same year, Ranjan et al.\ present a $3$-approximation~\cite{ranjan2014offline}, which they improve to a $2.7$-approximation in the following year~\cite{ranjan2015offline}. In 2021 two groups independently achieved the best-known approximation ratio of $(5/3+\eps)$ using different techniques~\cite{DeppertJ0RT21,GalvezGJK2021}. 
As a result, both hardness and the best approximation algorithm known match results from the classical \SP{} problem.

\begin{figure*}[t]
    	\centering
    	\begin{subfigure}[t]{0.45\textwidth}
     \centering

      	\begin{tikzpicture}
                \pgfmathsetmacro{\w}{0.65}
                \pgfmathsetmacro{\h}{0.4}
    			\draw (0,0) -- (8*\w,0);
    			\draw [dashed] (4*\w,4.5*\h) -- (4*\w,0*\h);
                \draw (4*\w,0*\h) -- (4*\w,-0.5*\h);
    			\draw [dashed] (3*\w,4.5*\h) -- (3*\w,0*\h);
                \draw (3*\w,0*\h) -- (3*\w,-0.5*\h);
    			\draw [dashed] (5*\w,4.5*\h) -- (5*\w,0*\h);
                \draw (5*\w,0*\h) -- (5*\w,-0.5*\h);
    			\node[below] at (4*\w,-0.5*\h) {\small $t$};
    			\node[below] at (3*\w,-0.5*\h) {\small $t-\eps$};
    			\node[below] at (5*\w,-0.5*\h) {\small $t+\eps$};
    			\drawLargeItem[\small A]{1*\w}{0}{4*\w}{2*\h};
    			\drawVerticalItem[\small B]{0}{2*\h}{7*\w}{3*\h};
    			\drawHorizontalItem[\small C]{1*\w}{3*\h}{4*\w}{4*\h};
    			\drawSmallItem[\small D]{4*\w}{3*\h}{6*\w}{4*\h};
    			\drawSmallItem[\small D]{4*\w}{1*\h}{6*\w}{2*\h};
    			\drawHorizontalItem[\small E]{4*\w}{0}{8*\w}{1*\h};
    			\end{tikzpicture}
    	\end{subfigure}
    	\begin{subfigure}[t]{0.45\textwidth}
     \centering
    			\begin{tikzpicture}
                \pgfmathsetmacro{\w}{0.65}
                \pgfmathsetmacro{\h}{0.4}
    			\draw (0,0) -- (8*\w,0);
    			\draw [dashed] (4*\w,4.5*\h) -- (4*\w,0*\h);
                \draw (4*\w,0*\h) -- (4*\w,-0.5*\h);
    			\draw [dashed] (5*\w,4.5*\h) -- (5*\w,0*\h);
                \draw  (5*\w,0*\h) -- (5*\w,-0.5*\h);
    			\node[below] at (4*\w,-0.5*\h) {\small $t$};
    			\node[below] at (5*\w,-0.5*\h) {\small $t+\eps$};
    			\drawLargeItem[\small A]{1*\w}{0}{4*\w}{2*\h};
    			\drawVerticalItem[\small B]{0}{2*\h}{4*\w}{3*\h};
    			\drawHorizontalItem[\small C]{1*\w}{3*\h}{4*\w}{4*\h};
    			\drawSmallItem[\small D]{4*\w}{2*\h}{6*\w}{4*\h};
    			\drawVerticalItem[\small B]{4*\w}{1*\h}{7*\w}{2*\h};
    			\drawHorizontalItem[\small E]{4*\w}{0*\h}{8*\w}{1*\h};
    			\end{tikzpicture}

    	\end{subfigure}
    	\caption{An illustration of the transformation from a $PTS$ instance (left) to a $DSP$ instance (right). 
     The line $t$ represents the current state of the transformation algorithm. 
     No item is sliced by a vertical line in the $PTS$ instance. 
     Similarly, no item is sliced by a horizontal line in the $DSP$ instance. 
     Note that the height did not change.}
    	\label{fig:fiveTransform}
    \end{figure*}
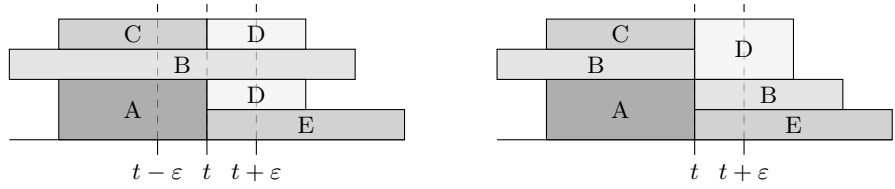
\paragraph*{Our Contribution} 
We prove the inapproximability of DSP, even in pseudo-polynomial time, below an approximation ratio of $(5/4)$ through a novel reduction. 
Through this reduction, we discover a previously unknown connection between DSP and \textsc{Parallel Task Scheduling} (PTS). 
PTS is a scheduling problem where jobs require a certain number of machines, and the aim is to minimize the makespan of the generated schedule. 
See \cref{sec:Inapp} for a complete, formal introduction of the problem.
Using this connection, we are able to incorporate some existing algorithms to generate optimal solutions for both DSP and PTS when admitting some resource augmentation. 
Finally, we present an approximation algorithm almost matching the inapproximability result we show.
We provide techniques to
Finally, we add techniques to the best-known pseudo-polynomial approximation algorithm for the classical SP problem to make it applicable to DSP. 
Techniques used in that algorithm do not translate to DSP directly. 
Structurally, sliced solutions are more difficult to analyze. We provide techniques to partition optimal packings such that they can be structured. 
We provide a subroutine to place certain items integrally, enabling the structuring. 
Through it, we can assign items to be packed integrally, even in this demand setting. 
These new techniques enable us to adapt the existing SP algorithm. 
This generates an algorithm with an approximation ratio almost matching the inapproximability result we show. 
Below is a brief summary of the presented results.
\begin{itemize}
    \item Proof that DSP is strongly $NP$-hard to approximate below a ratio of $(5/4)$. 
    \item A novel transformation algorithm between DSP and PTS.
    \item A framework for an algorithm producing an optimal solution for DSP when augmenting the width of the strip by a factor of $(3/2+\eps)$ in polynomial time.
    \item Frameworks for algorithms producing optimal solutions for PTS when augmenting the number of used machines by either $(5/3+\eps)$ or $(5/4+\eps)$ in polynomial and pseudo-polynomial time respectively.
    \item A  $(5/4+\eps)$-approximation algorithm for the DSP problem.
\end{itemize}

\paragraph*{Structure of the document}
This paper is divided into two parts. 
We begin by approaching the DSP problem in pseudo-polynomial time by providing a lower bound, proving the inapproximability below ratios of $(5/4)$ even in pseudo-polynomial time. 
The  techniques we develop yield further interesting results for both DSP and PTS when admitting some resource augmentation. 
These results are also presented in this part.

Afterward, we provide an algorithmic upper bound. 
We achieve an approximation ratio of~$(5/4+\eps)$. 
In this section, we provide an overview of the algorithm's steps and the required analysis.

\section{Inapproximability of DSP for ratios below \texorpdfstring{$5/4$}{5/4}}
\label{sec:Inapp}
We show that the pseudo-polynomial strip packing problem is $NP$-hard to approximate for ratios better than $\nicefrac{5}{4}\cdot \opt$. 
We achieve this result by providing a reduction onto the \textsc{Parallel Task Scheduling} problem, where we are given a set of $m$ machines and a set $\jobs$ of $n$ jobs. 
Each job $j \in \jobs$ has a processing time $p(j) \in \natNumbers$ and a number of required machines $q(j)\in \natNumbers$ assigned. 
We aim to place all jobs inside a schedule with a minimum makespan. 
A schedule $\Sigma$ is a combination of two functions $\sigma, \rho$. 
The first function $\sigma : \jobs\longrightarrow \natNumbers$ assigns each job to a starting point in the schedule. 
The second function $\rho : \jobs \longrightarrow \{M\vert M \subseteq \{1,\ldots, m\}\}$ maps each job to a set of machines it is processed on. 
We can express the makespan $T:= \max_{i\in \jobs} \sigma(i)+p(i)$ as the latest finishing point of any job in~$\jobs$. 

It is known that \textsc{Parallel Task Scheduling} is strongly $NP$-hard to solve for all values $m\geq 4$ \cite{DuL89a,HenningJRS18} while there exist pseudo-polynomial time algorithms to solve it for all values $m\leq 3$, see \cite{DuL89a}. 
As a consequence, providing a reduction from DSP onto PTS shows that there is no pseudo-polynomial time algorithm to approximate DSP with a ratio better than $\nicefrac{5}{4}$. This yields the following theorem.

\begin{thm}
\textsc{Demand Strip Packing} is strongly $NP$-hard to approximate with a ratio lower than $\nicefrac{5}{4}$. 
\label{thm:Hardness}
\end{thm}

\begin{proof}
The transformation of an instance $\instance$ of PTS to an instance $\instanceTransform$ of DSP is defined as follows.
We begin by creating an item $i_j$ for every job $j$ in $I$. 
The width $\width(i_j)$ is equal to the processing time $p(j)$ of job $j$, i.e.\ $\width(i_j)=p(j)$. \ari{Similarly, we have $\height(i_j)=q(j)$.}
As we create exactly one item per job, the value $n$ does not change. 
Finally, we map the processing time $T$ onto the width of the strip $\stripWidth$ and the number of machines $m$ onto the desired height of the strip $\stripHeight$.
\ari{Using this transformation, we show a feasible packing for $\instanceTransform$ with height $\stripHeight$ and width $\stripWidth$ exists if and only if a feasible schedule for $\instance$ with makespan $T$ and $m$ machines exists.}

    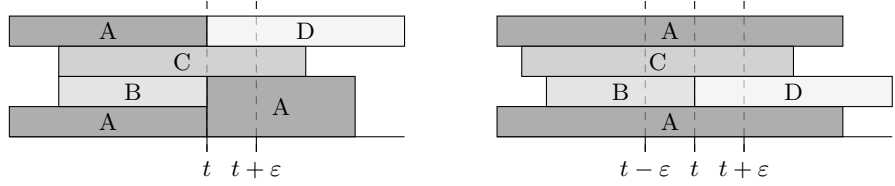
\begin{figure*}[t]
	 \centering
    \begin{subfigure}[t]{0.45\textwidth}
    \centering
			\begin{tikzpicture}
                \pgfmathsetmacro{\w}{0.65}
                \pgfmathsetmacro{\h}{0.4}
				\draw (0,0) -- (8*\w,0);
				\draw [dashed] (4*\w,4.5*\h) -- (4*\w,-0.5*\h);
                 \draw  (4*\w,0*\h) -- (4*\w,-0.5*\h);
				\draw [dashed] (5*\w,4.5*\h) -- (5*\w,-0.5*\h);
                 \draw  (5*\w,0*\h) -- (5*\w,-0.5*\h);
				\node[below] at (4*\w,-0.5*\h) {\small $t$};
				\node[below] at (5*\w,-0.5*\h) {\small $t+\eps$};
				\drawLargeItem[\small A]{0}{0}{4*\w}{1*\h};
				\drawLargeItem[\small A]{0}{3*\h}{4*\w}{4*\h};
				\drawLargeItem[\small A]{4*\w}{0}{7*\w}{2*\h};
				\drawVerticalItem[\small B]{1*\w}{1*\h}{4*\w}{2*\h};
				\drawHorizontalItem[\small C]{1*\w}{2*\h}{6*\w}{3*\h};
				\drawSmallItem[\small D]{4*\w}{3*\h}{8*\w}{4*\h};
			\end{tikzpicture}

	   \end{subfigure}
        \begin{subfigure}[t]{0.45\textwidth}
        \centering
			\begin{tikzpicture}
                \pgfmathsetmacro{\w}{0.65}
                \pgfmathsetmacro{\h}{0.4}
				\draw (0,0) -- (8*\w,0);
				\draw [dashed] (4*\w,4.5*\h) -- (4*\w,-0.5*\h);
                \draw  (4*\w,0*\h) -- (4*\w,-0.5*\h);
				\draw [dashed] (3*\w,4.5*\h) -- (3*\w,-0.5*\h);
                 \draw (3*\w,0*\h) -- (3*\w,-0.5*\h);
				\draw [dashed] (5*\w,4.5*\h) -- (5*\w,-0.5*\h);
                 \draw  (5*\w,0*\h) -- (5*\w,-0.5*\h);
				\node[below] at (4*\w,-0.5*\h) {\small $t$};
				\node[below] at (3*\w,-0.5*\h) {\small $t-\eps$};
				\node[below] at (5*\w,-0.5*\h) {\small $t+\eps$};
				\drawLargeItem[\small A]{0}{0}{7*\w}{1*\h};
				\drawLargeItem[\small A]{0}{3*\h}{7*\w}{4*\h};
				\drawVerticalItem[\small B]{1*\w}{1*\h}{4*\w}{2*\h};
				\drawHorizontalItem[\small C]{0.5*\w}{2*\h}{6*\w}{3*\h};
				\drawSmallItem[\small D]{4*\w}{1*\h}{8*\w}{2*\h};
			\end{tikzpicture}
	\end{subfigure}
	\caption{An illustration of the transformation from a $DSP$ instance on the left to a $PTS$ instance. 
    The line at $t$ represents the current state of the transformation algorithm. 
    The job $A$ is sliced after this point in the packing because $D$ is placed such that it occupies the top machine. 
    After the swap, all jobs are wholly scheduled on the machines they start on. 
    The height does not change.}
	\label{fig:SSPtoPTS}
\end{figure*}
Let $\instance$ be a yes-instance of PTS with $m$ machines and makespan $T$, and $\Sigma$ be the generated feasible schedule. 
Create a strip of width $T$. 
Place all transformed items $i_j$ according to their assignment $\rho(j)$, i.e.\ an item representing a job that has $\rho(j)=m_0,m_1,m_4$ gets placed at height $0,1$ and $4$. 
The items start at $\lambda(i_j)=\sigma(j)$, the starting point of the corresponding job in the schedule. 
As the height $\height(i_j)$ and width $\width(i_j)$ of all items is identical to the dimensions $q(j),(p(j)$ of their respective jobs, this packing must fit into the strip and have height at most $\stripHeight=m$. 
However, note that this packing is not a feasible solution for the DSP problem yet, as items may contain horizontal gaps. 
To remove these gaps and  create a feasible packing we use the following procedure. 

Traverse the packing from left to right until you find the first point $t$ in the packing that is infeasible, i.e.\ where an item $i_j$ contains a horizontal gap. 
Inspect the point at $t+\eps$. 
Draw a vertical line and sort all the items that are intersected line in ascending order of their heights. 
Place these items according to this order, beginning at the bottom. 
After this swap $i_j$ no longer contains a horizontal gap. 
We can iteratively repeat this procedure, traversing the schedule until the end. 
After this, the solution $\Sigma$ is transformed into a feasible solution $\Lambda$ for DSP. 
The height of the packing did not change at any point in the procedure, so we have $\stripHeight=m$. 
Similarly, the total width of the strip was not exceeded, so we have $\stripWidth=T$. 
For an illustration, see \cref{fig:fiveTransform}.

    Next, we show that a yes-instance $\instanceTransform$ of DSP can be transformed into a feasible schedule $\Sigma'$. 
    Let $\Lambda'$ be the generated packing for $\instanceTransform$ with height $\stripHeight$ and width $\stripWidth$. 
    The start point of any item $i$ in $\Lambda$ is given by $\lambda(i)$. 
    Generate a schedule $\Sigma'$ containing $\stripHeight$ machines. Then, place all jobs $j_i$ given by transformed items $i$ at their positions in $\Lambda'$, according to $\lambda(i)$. 
    Since the number of required machines $q(j_i)=\height(i)$, we still require at most $m$ machines for the packing. 
    Similarly, the makespan of the schedule is at most $\stripWidth$, as we have $p(j_i)=\width(i)$ for the processing times.
    However, note that this packing is not a feasible solution for PTS yet, as there may be vertical slices present in some jobs. 
    We utilize the following procedure to remove these slices and create a feasible schedule $\Sigma'$.

    Traverse the schedule from the beginning, i.e.\ left, to the end. At the beginning of the schedule, sort all jobs according to their required machines $q(j_i)$, then assign them in ascending order to machines. 
    This does not increase the number of required machines but removes any slices for all these jobs, as they are wholly scheduled on their assigned machines. 
    Traverse the schedule further until we reach a job $k$ that is infeasibly scheduled, i.e.\ placed on an occupied machine or sliced in some way. 
    This can only occur at the start of some job~$j'$. Call this point in the schedule $t$. 
    Draw a vertical line at $t-\eps$ and one at $t+\eps$. 
    Note the order of jobs at $t-\eps$. 
    These are scheduled feasibly. 
    Since we had a packing of height $\stripHeight=m$, we know that the total number of processors required at both lines is at most $m$. 
    Therefore, there must be gaps inside the schedule at the vertical line $t+\eps$. 
    The total number of empty processors is at least $q(k)$, as all items placed at this line in the packing sum up to a height of at most $m$. 
    Thus, we can feasibly assign $k$ to some $q(k)$ processors. 
    For an illustration, see \cref{fig:SSPtoPTS}.
    We repeat this for every job at $t$ and traverse the schedule further. 
    \ari{As all jobs to the left of $t$ are scheduled feasibly, we must generate a feasible schedule $\Sigma'$.}

    Both procedures we describe here only act once for every beginning item/job. 
    As such, their running time is polynomial.
    We have shown that any solvable instance $\instance$ of PTS can be transformed into a feasible packing $\Lambda$ and vice versa. 
    There exists a transformation from the strongly $NP$-complete problem \textsc{3-Partition} onto PTS, given in~\cite{HenningJRS18}. 
    They show that a schedule using $4$~machines with makespan~$\stripWidth$ can not exist unless you can solve the underlying \textsc{3-Partition} instance~$\instance_{3P}$. 
    Applying our transformation result, therefore, yields the following: 
    If there was an algorithm that solves DSP in pseudo-polynomial time with an approximation ratio less than $\nicefrac{5}{4}$, we can apply the transformation to the instance $\instance_{3P}$ to generate a feasible demand strip packing instance. 
    Then, the algorithm produces a packing of height $<\nicefrac{5}{4}\cdot \opt =5$. 
    After transforming the instance back into the PTS instance, we would generate a schedule that utilizes $4$ machines and has makespan $\stripWidth$. 
    This contradicts the results in~\cite{HenningJRS18}, proving \cref{thm:Hardness}.
    \end{proof}
   \begin{lem}
    The transformation algorithms from PTS to SSP and vice versa have a running time of $\mathcal{O}(n\cdot n\log(n))$ and $\Oh(n^2)$ respectively.
    \label{lem:runTimeTrans}
\end{lem}
\begin{proof}
  The procedure to transform a schedule $\Sigma$ into a feasible packing $\Lambda$ is computable in $\mathcal{O}(n\cdot n\log(n))$. This is because a horizontal gap can only exist once for every job, at its starting point, because jobs are wholly scheduled on the same machines. Therefore, there can only be at most $n$ points $t$ at which we have to employ the procedure. Due to the sorting, we utilize, a single reordering takes $\Oh(n\log(n))$ time. This yields a total transformation time of $\mathcal{O}(n\cdot n\log(n))$.
    
    Similarly, the procedure to transform a packing $\Lambda'$ into a feasible schedule $\Sigma'$ can be computed in $\Oh(n^2)$. Any application of the described procedure only occurs once at the beginning of the schedule and then at most once more for every starting point $\startPoint(j)$ of a job $j$. Thus, we have $\Oh(n)$ many applications. While the initial application requires some sorting of jobs, the remaining applications simply need to check all other jobs $j'$ present at that point in time for their assigned machines. The sorting is feasible in $\Oh(n\log(n))$, checking the values $\rho(j')$ is feasible in linear time $\Oh(n)$. As the sorting is only required once, the total running time of this procedure is $\Oh(n^2)$.  
\end{proof}

  \subsection{Applying our transformation using resource augmentation}  
    An important consequence of this computable transformation is that these problems can now be treated as the dual of each other. 
    This directly yields results for both the DSP and PTS problems when allowing some resource augmentation. 
    In the general case, a $(\nicefrac{3}{2}+\eps)$ approximation algorithm is known for PTS~\cite{JansenT10}. 
    \ari{Implementing the transformations given above and embedding them in a dual approximation framework yields the following. }
    \begin{cor}
    There is a polynomial time algorithm that yields a packing with optimal height for 
    \SSP{} 
    when we are allowed to augment the width of the strip by a factor of $(\nicefrac{3}{2}+\eps)$, i.e.\ consider the strip with width $(\nicefrac{3}{2}+\eps)\stripWidth$
    \label{cor:SSPAlg}
    \end{cor}
    \begin{proof}
        Let $\instance$ be an instance of DSP containing $n$ items and a strip of width $\stripWidth$. 
        Transform $\items$ into a set of jobs $\jobs$ as described above. 
        Guess the number of machines $m$ required via a binary search. 
        Generate the upper and lower bound via the Steinberg algorithm for $\instance$, which guarantees a $2$-approximation.
        
        As the lower bound, use $\sum(\width(i)\cdot \height(i))/\stripWidth$, i.e.\ the total area of items distributed evenly across the strip. As the upper bound use the height of the packing given by the Steinberg algorithm for $\instance$. This is at most a $2$-approximation.

        For every guess $\stripHeight$ of the binary search, use the algorithm given in~\cite{JansenT10} 
        to compute a solution for the transformed PTS instance. Check whether the generated makespan $T$ is at most $T\leq (\nicefrac{3}{2}+\eps)\cdot \stripWidth$. 
        If this holds, decrease the guess $\stripHeight$ and repeat the procedure. 
        If we have $T>(\nicefrac{3}{2}+\eps)\cdot \stripWidth$, we know that there is no feasible packing with height $\stripHeight$ and width $\stripWidth$. 
        We know that there is an optimal packing of height $\stripHeight$ and width $\stripWidth$ if and only if there is a schedule of makespan $T=\stripWidth$ that uses $m=\stripHeight$ machines. 
        Further, we know that the algorithm in~\cite{JansenT10} yields a schedule of makespan $(\nicefrac{3}{2}+\eps)\cdot \opt=(\nicefrac{3}{2}+\eps)\cdot \stripWidth$. 
        Thus, if the algorithm does not find a solution with makespan less than $(\nicefrac{3}{2}+\eps)\cdot \stripWidth$, we know that there can not be a packing with width $\stripWidth$ and height $\stripHeight$.
        \ari{In this case, increase the guess $\stripHeight$.}

        After completing the binary search, we save the solution with the lowest number of required machines $m=\stripHeight$. 
        Transforming this solution yields a feasible packing with height $\stripHeight$ and width at most $(\nicefrac{3}{2}+\eps)\cdot \stripWidth$.
        As all algorithm steps run in polynomial time, the whole algorithm also runs in polynomial time.
    \end{proof}
    Similarly, we can generate schedules with optimal makespan $T$ for the PTS problem when we are allowed to augment the number of machines used. 
    If the set of jobs $\jobs$ is arbitrary, we obtain an optimal makespan by augmenting $m$ by a factor of $(\nicefrac{5}{3}+\eps)$. 
    To achieve this, we utilize the results given in~\cite{DeppertJ0RT21,GalvezGJK2021}, as they both provide a $(\nicefrac{5}{3}+\eps)$ for DSP. 
    If the jobs are further parameterized in their processing time, we can apply the pseudo-polynomial time algorithm presented in this paper to reduce the required augmentation to $(\nicefrac{5}{4}+\eps)$. This yields the following.
    
    \begin{cor}
        There is a polynomial time algorithm that yields a packing with optimum makespan for \PTS{} when we are allowed to augment the number of used machines $m$ by a factor of $(\nicefrac{5}{3}+\eps)$.
    \end{cor}
    \begin{proof}
        The proof closely resembles that of \cref{cor:SSPAlg}. 
        The core difference is that we implement the dual approximation framework by guessing the makespan $T=\stripWidth$. We iterate the binary search depending on whether we find a packing with height at most $(\nicefrac{5}{3}+\eps)\cdot m$. Again, all subroutines used run in polynomial time, so the entire algorithm does as well.
    \end{proof}
    \begin{cor}
        There is a polynomial time algorithm that yields a packing with optimum makespan for \PTS{} where the processing time is parameterized, when we are allowed to augment the number of used machines $m$ by a factor of $(\nicefrac{5}{4}+\eps)$.
        \label{cor:AugSSP}
    \end{cor}
    \begin{proof}
        The proof resembles that of \cref{cor:SSPAlg}. The dual approximation framework is implemented by guessing the makespan $T=\stripWidth$. We iterate the binary search depending on whether we find a packing with height at most $(\nicefrac{5}{3}+\eps)\cdot m$. All subroutines used run in polynomial time, so the entire algorithm does as well.
    \end{proof}
    
  
    With the hardness and transformation algorithms shown, let us now focus our attention on constructing the algorithm that achieves the approximation ratio of $(\nicefrac{5}{4}+\eps)$.

       \section{A \texorpdfstring{$(5/4 +\eps)$}{(5/4 + epsilon)} approximation for DSP in pseudo-polynomial time}
    We provide an algorithm that proves the following theorem.
    \begin{thm}
    There is an algorithm to solve the DSP problem with approximation ratio  $(\nicefrac{5}{4}+\eps)\opt$ in 
    $\Oh(n\log(n))\cdot \stripWidth^{\Oh_\eps(1)}$ time, where $\Oh_\eps$ neglects all factors depending on $1/\eps$.
    \label{thm:BigTheorem}
    \end{thm}
    To prove this theorem, we present an algorithm that fulfills the desired properties. Before we can begin summarizing this algorithm however, we must briefly introduce some concepts essential to the functionality of the algorithm. 
    
        \emph{Boxes.} When talking about boxes, we refer to some rectangular, axis-aligned objects with a width and a height. All boxes we describe are also able to be sliced as long as they are placed contiguously, i.e.\ their end point is at the sum of their starting point and their width. When illustrating proofs, we mostly ignore this sliceability of boxes in an effort to increase clarity. When discussing boxes, we \emph{do not} imply that they are not sliced.

    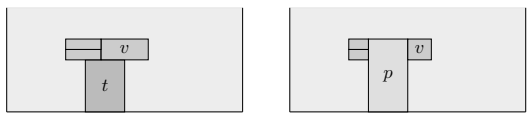
\begin{figure}
    \centering
    \resizebox{0.5\textwidth}{!}{
        \begin{tikzpicture}
            \pgfmathsetmacro{\w}{4.5}
            \pgfmathsetmacro{\h}{2}
            \drawJobNoBorder{0}{0}{\w}{\h};
            \draw (0,0)--(\w,0);
            \draw (0,0)--(0,\h);
            \draw (\w,0)--(\w,\h);
            \drawTallItem[$t$]{\w/3}{0}{\w/2}{\h/2};
            \drawHorizontalItem{\w/4}{\h/2}{0.4*\w}{0.6*\h};
            \drawHorizontalItem[$v$]{0.4*\w}{\h/2}{0.6*\w}{0.7*\h};
            \drawHorizontalItem{\w/4}{0.6*\h}{0.4*\w}{0.7*\h};
            \begin{scope}[xshift=1.2*\w cm]
            \drawJobNoBorder{0}{0}{\w}{\h};
            \draw (0,0)--(\w,0);
            \draw (0,0)--(0,\h);
            \draw (\w,0)--(\w,\h);
            \drawVerticalItem[$p$]{\w/3}{0}{\w/2}{0.7*\h};
            \drawHorizontalItem{\w/4}{\h/2}{\w/3}{0.6*\h};
            \drawHorizontalItem[$v$]{0.5*\w}{\h/2}{0.6*\w}{0.7*\h};
            \drawHorizontalItem{\w/4}{0.6*\h}{\w/3}{0.7*\h};
            \end{scope}
        \end{tikzpicture}
        }
    \caption{A basic illustration of the pseudo item generation. To the left we see the original items, in particular the defining tall item $t$. To the right we see the newly generated pseudo item $p$, with increased height but the same width as the original item constructing it. The remaining vertical items might get separated horizontally by later routines.}
    \label{fig:pseudoItemGen}
\end{figure}

    \emph{Pseudo Items} Pseudo items are a set of combined items that are placed adjacent to each other inside some packing. We combine these items and call this combination a new pseudo item. These items are especially relevant when we discuss the separation of vertical and tall boxes in the following. For a simple illustration of this process, see \cref{fig:pseudoItemGen}. Note that the construction of pseudo items may introduce horizontal separation in some \emph{vertical} items. Such items are later placed integrally into gaps left open by our generated packing.

    First, we give an overview of its functionality and steps. 
    Then, we continue by proving the correctness of each step. 
    Due to space concerns, we omit some technical analyses and instead focus on providing a sketch of the desired algorithm.
    
    We begin by rounding and categorizing the instance. 
    Then, we show that any optimal solution can be transformed into a structured solution, producing only a small loss in \ari{packing height}
    Finally, we show that we can fill this structure with items to generate a feasible packing. 
    The algorithm can be described as such:
    \begin{enumerate}
        \item[1] Define $\eps'=\Oh(\eps)$, formulate lower and upper bounds on the optimal packing height $\stripHeight, 2\stripHeight$. Scale the item heights in the instance such that they are in $\{1,\ldots,n/\eps'\}$
        \item[2] Guess values $\stripHeight'\in [\stripHeight,\ldots, 2\stripHeight]$ via binary search and test each value for a feasible packing in the following manner:
        \begin{enumerate}
            \item[3] Round and scale all items in the instance according to $\stripHeight'$ and $\eps'$. 
            Categorize items into item-types depending on their dimensions. 
            \item[4] Discard small and medium items. Guess the partition of the optimal packing into boxes. For each guessed partition, attempt the following:
            \begin{enumerate}
                \item[5] Place items in item-types inside their respective boxes. 
                If this placement is impossible, we know the guessed partition must be wrong and discard it. 
                If no partition permits a feasible placement, the guessed optimal value $\stripHeight'$ must be wrong.
            \end{enumerate}
            \item[6] Place the discarded medium items atop the generated packing.
        \end{enumerate}
        \item[7] Return the packing generated for the smallest value $\stripHeight'$.
    \end{enumerate}
    \paragraph*{Step 1}
    Set $\eps'=\min\{\nicefrac{1}{4},(1/\lceil10/\eps\rceil)\}$ for computational reasons. 
    The values $\stripHeight,2\stripHeight$ can be generated using Steinbergs algorithm~\cite{Steinberg97}, which computes a solution with height at most $2\cdot\opt$. Thus, we apply Steinbergs algorithm to the instance and take the resulting packing height as $2\stripHeight$, and halve it to obtain $\stripHeight.$
    \paragraph*{Step 2} We use a framework pioneered in~\cite{HochbaumS87}. Instead of computing the optimal height directly, we guess the optimal height inside the bounds given in step 1 and attempt to compute a packing of this height for every guess. As long as this is feasible, we decrease the guessed height, increasing it otherwise.
    \begin{figure}[t]
     \centering
  	\resizebox{0.3\textwidth}{!}{
  	\begin{tikzpicture}
  	\pgfmathsetmacro{\w}{5}
  	\pgfmathsetmacro{\h}{\w}
  	\pgfmathsetmacro{\e}{4*\w/10}
  	\pgfmathsetmacro{\d}{3*\w/10}
  	\pgfmathsetmacro{\m}{1.5*\w/10}
  	\pgfmathsetmacro{\f}{6*\w/10}
  	\draw [dashed] (\d,0) node [anchor=north]{$\delta\stripWidth$}-- (\d,\h);
  	\draw [dashed] (\m,0) node [anchor=north]{$\mu\stripWidth$}-- (\m,\h);
  	\draw [dashed] (0,\d) node [anchor=east]{$\delta \stripHeight'$}-- (\w,\d);
  	\draw [dashed] (0,\e) node [anchor=east]{$\eps \stripHeight'$}-- (\w,\e);
  	\draw [dashed] (0,\m) node [anchor=east]{$\mu \stripHeight'$}-- (\w,\m);
  	\draw [dashed]  (0,\f) node [anchor=east]{$(\nicefrac{1}{4}+\eps)\stripHeight'$} -- (\w,\f);
  	\drawSmallItem[$\itemsSmall$]{0}{0}{\m}{\m};
  	\drawVerticalItem[$\itemsVert$]{0}{\d}{\m}{\f};
  	\drawTallItem[$\itemsTall$]{0}{\f}{\d}{\h};
  	\drawLargeItem[$\itemsLarge$]{\d}{\d}{\w}{\h};
  	\drawHorizontalItem[$\itemsHor$]{\d}{0}{\w}{\m};
  	\draw (0,\d) -- (0,0) -- (\w,0) -- (\w,\d);
  	\draw (\m,\e) -- (\d,\e);
  	\node at (\m*0.5+\d*0.5,\e*0.5+\f*0.5) {$\itemsMedVert$};
  	\node at (\m*0.5+\d*0.5,\m*0.5+\d*0.5) {$\itemsMed$};
  	\end{tikzpicture}
  }
\caption{An illustration of the partitioning of items. Items are partitioned by their width and height according to this figure.}
\label{fig:structPart}
\end{figure}
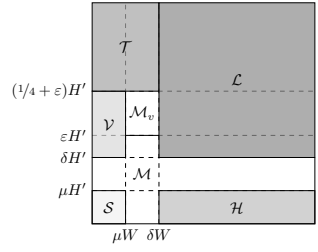
    \paragraph*{Step 3}
    From this point on, we assume that the guessed value $\stripHeight'$ given by step 2 is correct, i.e.\ $\stripHeight'=\opt$.
    Showing that all steps successfully compute proves the correctness of the algorithm.
    We round the items according to their sizes. 
    We select some values $\mu, \delta$ such that medium items have a small total area and discard them later. 
    Small items are discarded as well. 
    We categorize items as either large~$(\itemsLarge)$, tall~$(\itemsTall)$, vertical~$(\itemsVert)$, medium-vertical~$(\itemsMedVert)$, medium~$(\itemsMed)$, horizontal~$(\itemsHor)$ or small~$(\itemsSmall)$ according to some values $\mu, \delta$. See \cref{fig:structPart} for an intuitive illustration. 
    We partition items $i$ depending on their width $\width(i)$ and height $\height(i)$ as follows:
      \begin{itemize}
        \item $\itemsLarge:=\{i\in \instance\vert \height(i)> \delta \stripHeight', \width(i)\geq \delta W\}$ as the set of large items,
        \item $\itemsTall:= \{i\in \instance\vert \height(i)\geq (\nicefrac{1}{4}+\eps) \stripHeight', \width(i)< \delta W\}$ as the set of tall items,
        \item $\itemsVert:= \{i\in \instance\vert \delta \stripHeight' \leq \height(i)<(\nicefrac{1}{4}+\eps)\stripHeight', \width(i)\leq \mu W\}$ as the set of vertical items,
        \item $\itemsMedVert:= \{i\in \instance\vert \eps \stripHeight' \leq \height(i)< (\nicefrac{1}{4}+\eps) \stripHeight', \width(i)\leq \mu W\}$ as the set of vertical medium items,
        \item $\itemsHor:= \{i\in \instance\vert \height(i)\leq \mu \stripHeight', \width(i)\geq \delta W\}$ as the horizontal items,
        \item $\itemsSmall:= \{i\in \instance\vert \height(i)\leq \mu \stripHeight', \width(i)\leq \mu W\}$ as the set of small items and
        \item $\itemsMed:= \{i\in \instance\vert \height(i)< \eps \stripHeight', \mu W < \width(i)\leq \delta W\}\cup \{i\in \instance\vert \mu \stripHeight' <\height(i)\leq \delta \stripHeight'\}$ as the set of medium sized items.
    \end{itemize}
  
    We use techniques developed for the classical SP problem~\cite{StripPacking54} to generate values for $\mu, \delta$ accordingly to group the items. As a result, the total area of medium and medium-vertical items is small. 
     Next, we introduce the values for $\mu,\eps$ used in our rounding scheme and categorization of items.
    \begin{lem} 
        Let $f: \mathbb{R}\rightarrow \mathbb{R}$ be any function such that $1/f(\eps)$ is integral. Consider the sequence $\sigma_0=f(\eps), \sigma_{i+1}= \sigma_i^2f(\eps)$. There is a value $i\in \{0,\ldots,(2/f\eps))-1\}$ such that the total area of the items in $\itemsMed\cup \itemsMedVert$ is at most $f(\eps) \cdot W\cdot \opt$, if we set $\delta:=\sigma_i$ and $\mu:= \sigma_{i+1}.$
        \label{lem:RoundingValues}
    \end{lem}
    \begin{proof}
        This is a standard rounding technique that can be proven using the pigeonhole principle. The sequence $\sigma$ and the corresponding choices of $\delta$ and $\mu$ builds a sequence of $2/f(\eps)$ sets $\itemsMed_{\sigma_i} \cup \itemsMedVert{_{\sigma_i}}$.
        Each item $i \in \instance$ can occur in at most two of these sets, either because of its width or its height. Since the total area of all items is at most $W\cdot \opt $ one of these sets must have an area which is at most $f(\eps)\cdot W\cdot \opt.$        
    \end{proof}
    For this application, it is sufficient to choose $f(\eps)=\eps^{13}/k$ for a constant $k\in \mathbb{N}$ which has to fulfill certain properties we address later on. Since $1/\eps\in \mathbb{N}$, we have that $1/f(\eps)\in \mathbb{N}.$ Let $\delta$ and $\mu$ be values defined as in \cref{lem:RoundingValues}.
    Note that $\sigma_i=f(\eps)^{2^{i+1}-1}$ and, therefore $\delta\geq \sigma_{f(\eps)-1}\geq (\eps^{13}/k)^\ell\in \eps^{\mathcal{O}(\ell)}$, where $\ell:=2^{2k/\eps^{13}}$,i.e.\ $\delta \geq \eps^{2^{\mathcal{O}(1/\eps^{13}}}$. 
    In the following steps, we need $\delta$ to be of the form $\eps^x$ for some $x\in \mathbb{N}$. Therefore, define $\delta':=\eps^x$, such that $x\in \mathbb{N}$ and $\delta'\leq \delta \leq \delta'/\eps$. Note that $\mu:= \delta^2\eps^{13}/k\leq(\delta'/\eps)^2\eps^{13}/k= \delta'^2\eps^{11}/k$ and $\mu=\delta^2\eps^{13}/k \geq \delta^2\eps^{13}/k\geq \delta'^2\eps^{13}/k.$ 
    In the following, we rename $\delta:=\delta'$ for all the steps.
    By this choice, it still holds that the set of medium sized items has a total area of at most $(\eps^{13}/k)\cdot W \cdot \opt$ because, by reducing $\delta$ and not changing $\mu$, we only remove items from this set.\newline\indent
    \textit{Observation 2.} Since each item in $\itemsMedVert$ has a height of at least $\eps\opt$ and width of at least $\mu W\geq \delta^2\eps^{13}/k W$, i.e.\ an area of at least $(\delta^2\eps^{14}/k)\cdot W \cdot \opt,$ it holds that $$ \vert \itemsMedVert\vert \leq (\eps^{13}/k)\cdot W\cdot \opt /((\delta^2\eps^{14}/k)\cdot W \cdot \opt) = 1/\delta^2\eps. $$
    After we have found the corresponding values for $\delta$ and $\mu$ and partitioning the set of items in the instance accordingly, we round the height of all items with height at least $\delta\opt$.
   
    We then round the instance such that all items with significant height have some canonical height. 
    \begin{lem} 
    Let $\delta=\eps^k$ for some value $k\in \mathbb{N}.$ At loss of a factor $(1+2\eps)$ in the approximation ratio, we can ensure that each item $i$ with height $\eps^{\ell-1}\opt\geq \height(i)\geq \eps^{\ell}\opt$ for some $\ell \in \mathbb{N}\leq k$ has height $k_i\eps^{\ell+1}\opt$ for a value $k_i\in \{1/\eps,\ldots, 1/\eps^2-1\}.$ Furthermore, the items upper and lower borders can be placed at multiples of $\eps^{\ell+1}\opt$.  \label{lem:Rounding} \end{lem}
    \begin{proof}
        Let there be a given packing inside the strip with height $\opt$. We stretch it vertically by a factor of $1+2\eps$. As a result, any point $(x,y)$ in the original packing now corresponds to a point $(x,(1+2\eps)y)$ in this stretched packing. Let $i\in \itemsLarge\cup\itemsTall\cup\itemsVert$ be an item with $\eps^{\ell-1}\cdot \opt \geq \height_i\geq \eps^\ell\cdot\opt$ and let $y_T$ and $y_B$ be the $y$-coordinates of its top and bottom edges in the original strip. As $i$ can be sliced, both $y_T$ and $y_B$ are sets that can contain up to $\width(i)$ many values. As items cannot be sliced horizontally, there must be the same amount of values~$k$ in the sets $y_T$ and $y_B$. Define the stretched $y$-coordinates as $\Bar{y}_{T_k}=(1+2\eps)y_{T_k}$ for every entry $k\in y_T$. Define $\Bar{y}_B$ analogously. As a consequence, we have $$\Bar{y}_{T_k}-\Bar{y}_{B_k}=(1+2\eps)(y_{T_k}-y_{B_k}=(1+2\eps)\height_i.$$ Now, we change the $y$-coordinates of $i$ in the new strip to $$y_{B_k}':=\Bar{y}_{B_k}+\eps \height_i$$ and $$y_{T_k}':=\Bar{y}_{T_k}-\eps \height_i.$$ 
        For every value of $k$ we get $$y_{T_k}'-y_{B_k}'=\Bar{y}_{T_k}-\Bar{y}_{B_k}-2\eps\height_i=(1+2\eps)\height_i-2\eps\height_i=\height_i.$$ We have $\height_i\geq \eps^\ell\opt$, which implies $$y_{T_k}'-y_{B_k}'=\Bar{y}_{T_k}-\Bar{y}_{B_k}=\eps\height_i\geq \eps^{\ell+1}\cdot\opt.$$ This ensures that there is an integer $j_{T_k}$ such that $j_{T_k}\cdot\eps^{\ell+1}\cdot\opt\in [y_{T_k}',\Bar{y}_{T_k}]$ for the interval $[y_{T_k}',\Bar{y}_{T_k}]$. Analogously, there exists a set of integers $j_{B_k}$ so that $j_{T_k}\cdot\eps^{\ell+1}\cdot\opt\in [\Bar{y}_{B_k},y_{B_k}']$. We change the top $y-$coordinates of item $i$ to $y_{T_k}''=j_{T_k}\cdot\eps^{\ell+1}\cdot\opt$ and similarly the bottom coordinates to $y_{B_k}''=j_{B_k}\cdot\eps^{\ell+1}\cdot\opt$. 
        It may occur that $y_{T_k}''-y_{B_k}''>\height_i':=\lceil \height_i/(\eps^{\ell+1}\opt)\rceil\cdot \eps^{\ell+1}\cdot\opt$. In this case we increase $y_{B_k}''$ by $(y_{T_k}''-y_{B_k}'')-\height_i'$, such that $\height_i'= y_{T_k}''-y_{B_k}''.$ Any part of the item $i$ cannot intersect another item part because it is placed inside the stretched version of that item part. Thus, when we change the height $\height_i$ of each item $i\in \itemsLarge\cup\itemsTall\cup\itemsVert$ to $\height_i'=\lceil \height_i/(\eps^{\ell+1}\opt)\rceil\cdot \eps^{\ell+1}\cdot\opt,$ where $\ell\in \mathbb{N}$ is chosen such that $\eps^{\ell-1}\opt\geq\height_i\geq\eps^\ell\opt$. Note that  $\lceil \height_i/(\eps^{\ell+1}\opt)\rceil\in \{1/\eps,\ldots,1/\eps^2\}$ since $\eps^{\ell-1}\opt\geq\height_i\geq\eps^\ell\opt$. Since the value $h'$ does not exceed the stretched item height, we increased the optimal solution by a factor of $1+2\eps$ at most.
    \end{proof}   
    
    This eases the placement of these items. 
    Further, all of these items can be placed on some imaginary grid-lines inside an optimal packing of height $\stripHeight'$. 
    This is significant when we argue the structure of an optimal packing. 

    \paragraph*{Step 4} 
    We show that we can partition an optimal packing into a set of boxes in this step. 
    Each box only contains a certain subset of item types. 
    For this to be feasible, we first must reduce the different starting points of horizontal jobs to $\Oh(1/\eps\delta)$ with only a loss of $\bigO(\eps \stripHeight')$ in the packing height. 
    This is done using the technique given in~\cite{DeppertJ0RT21} for DSP. 
        \begin{lem}\cite{DeppertJ0RT21}
        At a loss of at most $\bigO(\eps H)$, we can reduce the number of used starting points of horizontal items to $\bigO(1/\eps\delta).$
        \label{lem:HorStartPoint}
    \end{lem}
    
    Next, we guess these points to partition the packing.
    We have to take special care when discussing the optimal structure of an DSP-packing, as we can not argue over the vertical placement of items in the optimal solution. 
    Therefore, the following proof requires some novel insights into the utilization of slicing.
   
        \begin{lem}
        We can partition the rounded instance into at most $\bigO_{\eps}(1)$ boxes. The set of boxes can be partitioned into sets $\boxesLarge,\boxesHor$ and $\boxesTallVert$ such that 
        \begin{itemize}
            \item boxes in $\boxesLarge$ contain only one item that is either large or medium vertical, i.e.\ it is in $\itemsLarge$ or $\itemsMedVert$. The total number of such boxes is $|\itemsLarge|+|\itemsMedVert|\leq \Oh_\eps(1),$
            \item $\boxesHor$ consists of $\Oh_\eps(1)$ boxes
            of height $\eps \delta \opt$, each of them containing at least some item in $\itemsHor$ but only items in $\itemsHor\cup \itemsSmall$,
            \item $\boxesTallVert$ consists of $\Oh_\eps(1),$ boxes, each of them containing items in $\itemsTall\cup \itemsVert\cup \itemsSmall$.
        \end{itemize} 
        \label{lem:boxPartition}
    \end{lem}
    \begin{figure}[t]
    \begin{center}
     \centering
     		
		\begin{subfigure}{0.5\textwidth}
			\begin{tikzpicture}
			\pgfmathsetmacro{\w}{6}
			\pgfmathsetmacro{\h}{3}
			\colorlet{myfillcolor}{white!65!black}
			\colorlet{myBGcolor}{white!90!black}
			\colorlet{largeProfile}{white!55!black}
			\colorlet{horProfile}{white!75!black}
			\colorlet{vertProfile}{white!85!black}
			\coordinate (A) at (0,\h);
			\coordinate (B) at (\w, \h);
			\coordinate (C) at (0, 0.6*\h);
			\coordinate (D) at (0.1*\w, 0.6*\h);
			\coordinate (E) at (0.1*\w, 0.45*\h);
			\coordinate (F) at (0.3*\w, 0.45*\h);
			\coordinate (G) at (0.3*\w, 0.55*\h);
			\coordinate (H) at (0.45*\w, 0.55*\h);
			\coordinate (J) at (0.45*\w, 0.15*\h);
			\coordinate (K) at (0.7*\w, 0.15*\h);
			\coordinate (L) at (0.7*\w, 0.65*\h);
			\coordinate (M) at (0.9*\w, 0.65*\h);
			\coordinate (N) at (0.9*\w, 0.4*\h);
			\coordinate (O) at (\w, 0.4*\h);
		
			\coordinate (P) at (0,0);
			\coordinate (Q) at (\w,0);
			\coordinate (zero) at (0,0);
			\coordinate (width) at (\w,0);
			\coordinate (LA) at (0,0.2*\h);
			\coordinate (LB) at (0.15*\w,0.2*\h);
			\coordinate (LC) at (0.15*\w,0.15*\h);
			\coordinate (LD) at (0.3*\w,0.15*\h);
			\coordinate (LE) at (0.3*\w,0.3*\h);
			\coordinate (LF) at (0.45*\w,0.3*\h);
			\coordinate (LG) at (0.45*\w,0*\h);
			\coordinate (LH) at (0.7*\w,0*\h);
			\coordinate (LI) at (0.7*\w,0.42*\h);
			\coordinate (LJ) at (0.9*\w,0.42*\h);
			\coordinate (LK) at (0.9*\w,0.18*\h);
			\coordinate (LL) at (1*\w,0.18*\h);
			
			\draw (0,0) -- (\w,0) --(\w,\h) -- (0,\h) -- (0,0);
		
			
		
		      \draw [opacity=0.7,fill=largeProfile] (zero) -- (LA) -- (LB) -- (LC) -- (LD) -- (LE) -- (LF) -- (LG) --
		    (LH) -- (LI) -- (LJ) -- (LK) -- (LL) -- (width) --cycle;
		
		      \draw [opacity=0.7,fill=horProfile]	(C) -- (D) -- (E) -- 
		      (F) -- (G) -- (H) -- (J) -- (K) -- (L) -- (M) --(N) -- (O) -- 
		      (LL) -- (LK) -- (LJ) -- (LI) -- (LH) -- (LG) -- (LF) --
		      (LE) -- (LD) -- (LC) -- (LB) -- (LA) -- cycle;
            \draw [opacity=0.7,pattern=north west lines]	(C) -- (D) -- (E) -- 
		      (F) -- (G) -- (H) -- (J) -- (K) -- (L) -- (M) --(N) -- (O) -- 
		      (LL) -- (LK) -- (LJ) -- (LI) -- (LH) -- (LG) -- (LF) --
		      (LE) -- (LD) -- (LC) -- (LB) -- (LA) -- cycle;
		      \draw [opacity=0.7, fill=myfillcolor] (A)  -- (C) -- (D) -- (E) -- 
		      (F) -- (G) -- (H) -- (J) -- (K) -- (L) -- (M) --(N) -- (O) -- (B) -- cycle;
		
			\draw(0,\h) rectangle node [midway, opacity=1] {$B_1$} (0.1*\w,0.6*\h);
			\draw(0.1*\w,\h) rectangle node [midway, opacity=1] {$B_2$} (0.3*\w,0.45*\h);
			\draw(0.3*\w,\h) rectangle node [midway, opacity=1] {$B_3$} (0.45*\w,0.55*\h);
			\draw(0.45*\w,\h) rectangle node [midway, opacity=1] {$B_4$} (0.7*\w,0.15*\h);
			\draw(0.7*\w,\h) rectangle node [midway, opacity=1] {$B_5$} (0.9*\w,0.65*\h);
			\draw(0.9*\w,\h) rectangle node [midway, opacity=1] {$B_6$} (1*\w,0.4*\h);

            \draw (0.15*\w,0.2*\h) rectangle node [midway] {$L_1$} (0,0);
            \draw (0.15*\w,0.15*\h) rectangle node [midway] {$L_2$} (0.3*\w,0);
            \draw (0.3*\w,0.3*\h) rectangle node [midway] {$L_3$} (0.45*\w,0);
            \draw (0.7*\w,0.42*\h) rectangle node [midway] {$L_4$} (0.9*\w,0.2*\h);
            \draw (0.7*\w,0.2*\h) rectangle node [midway] {$L_5$} (0.9*\w,0);
            \draw (0.9*\w,0.18*\h) rectangle node [midway] {$L_6$} (\w,0);
			\end{tikzpicture}
		\end{subfigure}
  \end{center}
	\caption{The partitioned optimal packing generated by \cref{lem:boxPartition}. Named boxes $B_i$ contain only tall and vertical items. The hatched area contains only boxes for horizontal items. Finally, the lowest areas contain only boxes for large or medium vertical items. These boxes are named according to the item placed inside. Boxes for horizontal items may be sliced.}
	\label{fig:preProBoxes}
	\end{figure}
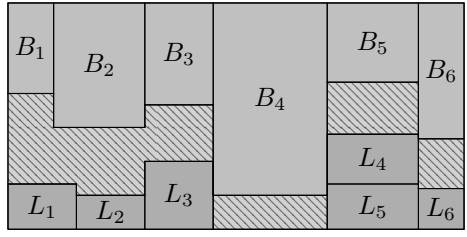
    
    \begin{proof}
        Consider the rounded optimal packing. In it, each item has a defined starting point on the horizontal axis. As such, we can consider any singular item to be packed fully integrally and simply shift the slicing of other items around such that they retain their original starting positions. We utilize this technique to define the boxes items fit in fully integrally, while still allowing them to be sliced again in order to be identical to the optimal packing again. We begin this procedure with the items in $\itemsLarge\cup \itemsMedVert$. Each of these items $i$ gets a personal box with width $w_i$, and height $h_i$, placed at exactly its starting position $\spo_i$. Since the starting point remains the same, the packing remains optimal.
        \newline\indent
        Next, we define the boxes $\boxesHor$. We begin at the left side of the packing. Traverse the packing towards the right until you find the starting point of a horizontal item $i$. Then, check all items in $\itemsHor$ with the same starting point $\spo(i)$ and select the item $j$ with the largest width. The box $B$ we now define starts at $\spo(i)=\spo(j)$ and ends at $\spo(i)+\width(j)$. Its height is $\eps \delta$. We iteratively add items to this box in the following manner: first, find all items $\ell$ with $\spo(i)\leq\spo(\ell) \leq \spo(\ell)+\width(\ell)\leq \spo(i)+\width(j)$, i.e.\ all items that can be fully contained inside the newly formed box. Next, order these items in descending order regarding their width. Then, take the item with the largest width and add it into the box. Repeatedly add such items until the next item we added would break the height of the box, i.e.\ there is a coordinate $x$ with $\sum_{i\in B, \spo(i)\leq x \leq \spo(i)+\width(i)}\height(i)>\eps\delta.$ Continue checking whether we can add items to this box until there are no possible items left. As the last step, take the widest possible item $\ell$ to add that exceeds the top border of the box. Take further items that do not overlap with $\ell$ or with each other, until no such items exist anymore, and add them to the box as well. Finally, we simply repeat this procedure until we have reached the right border of the packing. Clearly, these boxes contain only items in $\itemsHor$. At every point in the box, at most one item overlaps the top border. The items in $\itemsHor$ have a width of at least $\delta W$ and the boxes have a height of $\eps \delta\opt$. Thus, we can stack at most $(1+2\eps)/(\eps\delta)$ such boxes atop one another, and can have at most $1/\delta -1$ boxes beside each other. Therefore, the total amount of boxes in $\boxesHor$ is at most $(1+2\eps)/(\eps\delta)\cdot (1\delta -1)\leq (1+2\eps)/(\eps\delta^2)-2$. Because we reduced the amount of starting points of horizontal items inside the packing, we generate at most this amount of boxes.\newline\indent
        Finally, consider the boxes $\boxesTallVert$. For these, imagine the boxes $\boxesHor$ and $\boxesLarge$ forming a profile. That is, they are all sliced towards the top at their original positions. Since both of these sets of boxes have a height that is a multiple of $\eps\delta,$ this profile must be at multiples of $\eps \delta$ as well. Traverse this profile from left to right, until we reach the first point $x$ at which the height of the profile changes, i.e.\ $\sum_{i, \spo(i)\leq x \leq \spo(i)+w(i)}h(i)\neq \sum_{j, \spo(j)\leq x-1 \leq \spo(j)+w(j)}h(j)$. Draw a vertical line at this coordinate $x$. This defines the right border of the first box, and the left border of the second box. Continue this procedure, iteratively drawing vertical lines that define box borders, until we arrive at the right border of the strip. These borders define the boxes in $\boxesTallVert$. Clearly, they contain only items in $\itemsTall\cup \itemsVert\cup \itemsSmall$. We have at most $(1+2\eps)/(\eps\delta^2)-2$ boxes in $\boxesLarge \cup \boxesHor$. Since we drew at most two lines for each box in $\boxesHor$ and $\boxesLarge$ and each such line touches at most two boxes in $\boxesTallVert$, we have at most $2(1+2\eps)/(\eps\delta^2)$ boxes in $\boxesTallVert$. For an illustration of this procedure, see \cref{fig:preProBoxes}.\newline\indent
        All types of boxes we generated in this step have their upper and lower border at a multiple of $\eps\delta\opt$. This is because, for large, tall, vertical and medium vertical items, our rounding in \cref{lem:Rounding} ensures that the items contained in those boxes start at multiples of such a value, as $\delta=\eps^x$ for a value $x\in \mathbb{N}$. Additionally, while the horizontal items might not have their borders at such values, recall how we generated their boxes. We have shifted the packing such that they were on top of the remaining two sets of boxes, $\boxesTallVert$ and $\boxesLarge$. Since these have their upper and lower borders at multiples of $\eps\delta\opt$, the lowest set of horizontal boxes in $\boxesHor$ must have the same property. As each of these boxes has a height of $\eps\delta\opt$, all boxes further above must also have their upper and lower borders at those values.
    \end{proof}
    After this partition, we discard the small items. 
    We show later that we can place them inside some gaps without increasing the packing height. 
    The medium items get discarded and placed atop the packing using Steinbergs' algorithm, as they have a small total area due to \cref{lem:Rounding}. 
    
    The boxes $\boxesLarge$ and $\boxesHor$ only contain either a single item or items of one type. 
    As such, we can feasibly fill them in \emph{step 5.} 
    However, the boxes $\boxesTallVert$ do not share this property yet. To make them fillable, we must therefore partition them further.

    We partition each box $B\in \boxesTallVert$ separately. 
    We utilize different routines depending on the height of $B$. 
    The core complexity here is derived from the number of tall items that can be stacked atop one another. 
    As such, the three types of boxes we inspect are ones with height $\height(B)\leq \nicefrac{1}{2}\stripHeight'$, those with height $\height(B)\in (\nicefrac{1}{2}\stripHeight',\nicefrac{3}{4}\stripHeight']$, and boxes with height $\height(B)\in (\nicefrac{3}{4}\stripHeight',\stripHeight']$. 
    The boxes are categorized as such because there can only be one tall item at any given point in the smallest boxes, two tall items atop another in the second type, and finally, up to three tall items atop another in the largest type of box.

    An important new challenge derives for these boxes compared to the classical SP problem. 
    The core idea of the restructuring arguments for each box is the \emph{assignment} of tall items to the bottom, top, or middle of the box. 
    However, finding such an assignment is much harder when admitting sliced items. 
    Several items $i_1,i_2,i_3$ that are placed atop another can be sliced in a way where each of these items is the bottom-most at one point in the optimal packing.
   \ari{The most complex assignment is required for boxes with height $\height(B)\in (\nicefrac{3}{4}\stripHeight',\stripHeight']$.}
    
  Using this assignment subroutine, it becomes feasible to structure the three types of boxes. 
  In this structure, each type of box contains only at most $\Oh_\eps(1)$ sub-boxes. 
          \begin{figure*}[t]
        	\begin{subfigure}[t]{0.3\textwidth}
        			\centering
        				\resizebox{0.98\textwidth}{!}{
        		\begin{tikzpicture}
        		\pgfmathsetmacro{\w}{7.5}
        		\pgfmathsetmacro{\h}{5}
        		
        		\draw  (0,0) --(\w,0) -- (\w,\h) -- (0,\h) -- cycle;
        		
        		\foreach \x/\y/\xx/\yy/\z in {
        			0.1 / 0.2 /0.25 / 0.9/, 
        			0.3 / 0.1 /0.4 / 0.73/,
        			0.42 / 0 /0.54 / 0.65/,	
        			0.6 / 0.25 /0.72 / 1/, 
        			0.8 / 0.15 /0.96 / 0.93/,	 
        			0.98 / 0.3 /1.05 / 1/$o$				
        		}
        		{
        			\drawTallItem[\z]{\x*\w}{\y*\h}{\xx*\w}{\yy*\h};
        		}
        		\foreach \x/\y/\xx/\yy/\z in {
        			0.0 / 0.0 /0.1 / 1/, 
        			0.1 / 0.0 /0.25 / 0.2/,
        			0.25 / 0 /0.3 / 1/,	
        			0.3 / 0.0 /0.4 / 0.1/, 
        			0.4 / 0.0 /0.42 / 1/,
        			0.54 / 0.0 /0.6 / 1/,
        			0.6 / 0.0 /0.72 / 0.25/,
        			0.72 / 0.0 /0.8 / 1/,	
        			0.8 / 0.0 /0.96 / 0.15/,
        			0.96 / 0.0 /0.98 / 1/,
        			0.98 / 0.0 /1 / 0.3/			
        		}
        		{
        			\drawVerticalItem[\z]{\x*\w}{\y*\h}{\xx*\w}{\yy*\h};
        		}
        		\foreach \x/\y/\xx/\yy/\z in {
        			0.1 / 0.9 /0.25 / 1/, 
        			0.3 / 0.73 /0.4 / 1/,
        			0.42 / 0.65 /0.54 / 1/,	
        			0.8 / 0.93 /0.96 / 1/	 
        		}
        		{
        			\drawVerticalItem[\z]{\x*\w}{\y*\h}{\xx*\w}{\yy*\h};
        		}
        		
        		\end{tikzpicture}
        	}
        	\end{subfigure}
        	\begin{subfigure} {0.3\textwidth}
        		\centering
        			\resizebox{0.98\textwidth}{!}{
        				\begin{tikzpicture}
        				\pgfmathsetmacro{\w}{7.5}
        				\pgfmathsetmacro{\h}{5}
        				
        				\draw  (0,0) --(\w,0) -- (\w,\h) -- (0,\h) -- cycle;
        				
        				\foreach \x/\y/\xx/\yy/\z in {
        					0.1 / 0.0 /0.25 / 0.7/, 
        					0.3 / 0.0 /0.4 / 0.63/,
        					0.42 / 0 /0.54 / 0.65/,	
        					0.6 / 0.0 /0.72 / 0.75/, 
        					0.8 / 0.0 /0.96 / 0.78/,	 
        					0.98 / 0.3 /1.05 / 1/	$o$			
        				}
        				{
        					\drawTallItem[\z]{\x*\w}{\y*\h}{\xx*\w}{\yy*\h};
        				}
        				\foreach \x/\y/\xx/\yy/\z in {
        					0.0 / 0.0 /0.1 / 1/, 
        					0.25 / 0 /0.3 / 1/,	
        					0.4 / 0.0 /0.42 / 1/,
        					0.54 / 0.0 /0.6 / 1/,
        					0.72 / 0.0 /0.8 / 1/,	
        					0.96 / 0.0 /0.98 / 1/,
							0.98 / 0.0 /1 / 0.3/		
        				}
        				{
        					\drawVerticalItem[\z]{\x*\w}{\y*\h}{\xx*\w}{\yy*\h};
        				}
        				\foreach \x/\y/\xx/\yy/\z in {
        					0.1 / 0.7 /0.25 / 1/, 
        					0.3 / 0.63 /0.4 / 1/,
        					0.42 / 0.65 /0.54 / 1/,	
        					0.6 / 0.75 /0.72 / 1/, 
        					0.8 / 0.78 /0.96 / 1/	 
        				}
        				{
        					\drawVerticalItem[\z]{\x*\w}{\y*\h}{\xx*\w}{\yy*\h};
        				}
        				
        				\end{tikzpicture}
        			}
        	\end{subfigure}
   		     \begin{subfigure} {0.3\textwidth}
        	\resizebox{0.98\textwidth}{!}{
        		\begin{tikzpicture}
        		\pgfmathsetmacro{\w}{7.5}
        		\pgfmathsetmacro{\h}{5}
        		
        		\draw  (0,0) --(\w,0) -- (\w,\h) -- (0,\h) -- cycle;
        		
        		\foreach \x/\y/\xx/\yy/\z in {
        			0.28 / 0.0 /0.43 / 0.7/, 
        			0.55 / 0.0 /0.65 / 0.63/,
        			0.43 / 0 /0.55 / 0.65/,	
        			0.16 / 0.0 /0.28 / 0.75/, 
        			0.0 / 0.0 /0.16 / 0.78/,	 
        			0.98 / 0.3 /1.05 / 1/	$o$		
        		}
        		{
        			\drawTallItem[\z]{\x*\w}{\y*\h}{\xx*\w}{\yy*\h};
        		}
        		\foreach \x/\y/\xx/\yy/\z in {
        			0.65 / 0.0 /0.98 / 1/, 
        			0.98 / 0.0 /1 / 0.3/			
        		}
        		{
        			\drawVerticalItem[\z]{\x*\w}{\y*\h}{\xx*\w}{\yy*\h};
        		}
        		\foreach \x/\y/\xx/\yy/\z in {
        			0.28 / 0.7 /0.43 / 1/, 
        			0.55 / 0.63 /0.65 / 1/,
        			0.43 / 0.65 /0.55 / 1/,	
        			0.16 / 0.75 /0.28 / 1/, 
        			0.0 / 0.78 /0.16 / 1/	 
        		}
        		{
        			\drawVerticalItem[\z]{\x*\w}{\y*\h}{\xx*\w}{\yy*\h};
        		}
        		
        		\end{tikzpicture}
        	}
             \end{subfigure}
         \caption{Illustration of a packing of tall items for a box $\boxFourth$ with height less than $\nicefrac{1}{2}\stripHeight'$. Tall items are colored orange and vertical items green. To the left we see the box in the optimal packing, with the slices for pseudo items already generated. The second image represents the box with all tall items sliced to the bottom and their pseudo items fused on top of them. The final image represents the resulting box, sorted by height of tall items in movable slices. Note that the overlapping item $o$ or its slice is never moved, and thus requires one additional box.}
         \label{fig:SolveBoxFourthProof}
        \end{figure*}
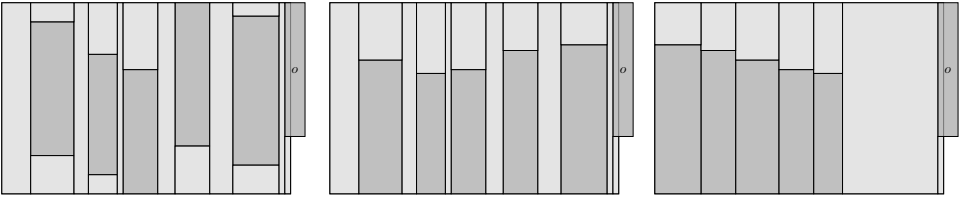
\paragraph*{Boxes $\boxFourth$ with $\height(\boxFourth)\in (1/4\cdot H', 1/2\cdot H']$}
          We solve these boxes by assigning each tall item to the bottom of the box. 
          Then, we slice the box at the borders of tall items. 
          We sort these slices in descending order of the height of tall items inside them and place them as such. 
          \begin{lem}
        Let $\boxFourth\in \boxesTallVert$ be a box with height $\height(\boxFourth)\in (\nicefrac{1}{4}\stripHeight',\nicefrac{1}{2}\stripHeight']$. There is a rearrangement of items in $\boxFourth$ such that there are at most $\mathcal{O}(1/\eps)$ different boxes for tall items, and at most $\mathcal{O}(1/\eps)$ different boxes for vertical items.
        \label{lem:ReorderTallVertBox14}
    \end{lem}
    \begin{proof}
    	This proof is an adaptation of a similar procedure from \cite{JansenR17}.
        The height of $\boxFourth$ clearly implies that there can never be two tall items placed above one another, since they do not fit into the box otherwise. Therefore, through slicing, we can assume that all tall items touch the bottom of the box with their bottom border. If this were not the case, we simply slice these tall items downward and the vertical items touching the ground upwards. Similarly, we can stack all vertical items into pseudo items that touch the ground as well.\newline\indent
        Before we reorder the inside of the box, we define any items that intersect the left or right box border as immovable. 
        Next, we construct slices around the tall items. For that, we draw a vertical line at the left and right border of a tall item. Each area between two such lines is considered a movable slice. We sort these movable slices in descending order of the tall items inside them. For an illustration of this procedure, see \cref{fig:SolveBoxFourthProof}. Through this reordering, we get at most $1/\eps$ boxes for movable tall items, since the heights of tall items were rounded to multiples of $\eps\opt$. Since one tall item can overlap each box border, we might need two further boxes for these items. Therefore, the total number of boxes for tall items is in $\bigO(1/\eps)$.\newline\indent
        On top of each box for tall items, we introduce at most one box for vertical items. Finally, there might be movable slices that do not contain any tall item. These are placed adjacent to each other at the right side of $\boxFourth$ after our reordering. For these, we construct one box with height $\height(\boxFourth)$. In total, that leaves us with $1/\eps$ boxes above the movable tall items, two boxes above the tall items that may intersect the borders of $\boxFourth$ and one final box of height $\height(\boxFourth)$. Thus, the total number of boxes for vertical items is in $\bigO(1/\eps)$.

    \end{proof}

\paragraph*{Boxes $\boxHalf$ with $\height(\boxHalf)\in (1/2\cdot H', 3/4\cdot H']$}
            \begin{lem}
        Let $\boxHalf\in \boxesTallVert$ be a box with height $\height(\boxHalf)\in (\nicefrac{1}{2}\stripHeight',\nicefrac{3}{4}\stripHeight']$. We find a rearrangement of items in $\boxHalf$ such that we generate at most $\mathcal{O}_\eps(1)$ sub-boxes containing only either tall or vertical items, where each sub-box $\boxHalf_T$ for tall items only contains items of the same height and the vertical items can be packed fractionally into boxes for vertical items.
        \label{lem:ReorderTallVertBox12}
    \end{lem}
    \begin{figure*}[h]
    	\centering
    	\begin{subfigure}[t]{0.45\textwidth}
    		\centering
    		\resizebox{0.9\textwidth}{!}{
    			\begin{tikzpicture}
    			\pgfmathsetmacro{\w}{15}
    			\pgfmathsetmacro{\h}{8}
    			
    			\draw  (0,0) --(\w,0) -- (\w,\h) -- (0,\h) -- cycle;
    			\draw [dashed, thick] (0.3*\w,0) -- (0.3*\w,\h);
    			\draw [dashed, thick] (0.7*\w,0) -- (0.7*\w,\h);
    			
    			\foreach \x/\y/\xx/\yy/\z in {
    				0.3 / 0 /0.37 / 0.75/{\large $i_\ell$},
    				0.37 / 0 /0.4 / 0.65/,
    				0.43 / 0 /0.48 / 0.7/,
    				0.52 / 0 /0.55 / 0.53/,
    				0.55 / 0 /0.58 / 0.6/	
    			}
    			{
    				\drawTallItem[\z]{\x*\w}{\y*\h}{\xx*\w}{\yy*\h};
    			}
    			\foreach \x/\y/\xx/\yy/\z in {
    				0.4 / 0 /0.43 / 0.25/,
    				0.48 / 0 /0.5 / 0.3/,
    				0.5 / 0 /0.52 / 0.4/,
    				0.58 / 0 /0.60 / 0.1/,
    				0.60 / 0 /0.75 / 0.2/$t'$	
    			}
    			{
    				\drawVerticalItem[\z]{\x*\w}{\y*\h}{\xx*\w}{\yy*\h};
    			}
    			\foreach \x/\y/\xx/\yy/\z in {
    				0.4 / 0.4 /0.43 / 1/,
    				0.48 / 0.45 /0.52 / 1/,
    				0.58 / 0.3 /0.63 / 1/,
    				0.63 / 0.22 /0.7 /1/$i_r$		
    			}
    			{
    				\drawTallItem[\z]{\x*\w}{\y*\h}{\xx*\w}{\yy*\h};
    			}
    			\foreach \x/\y/\xx/\yy/\z in {
    				0.24 / 0.8 /0.4 / 1/$b'$,
    				0.43 / 0.9/0.45 / 1/,
    				0.45 / 0.75 /0.48 / 1/,
    				0.52 / 0.6 /0.54 / 1/,
    				0.54 / 0.7 /0.58 /1/
    			}
    			{
    				\drawVerticalItem[\z]{\x*\w}{\y*\h}{\xx*\w}{\yy*\h};
    			}

    			\end{tikzpicture}
    		}
    	\end{subfigure}
    	\begin{subfigure} [t]{0.45\textwidth}
    		\centering
    		\resizebox{0.9\textwidth}{!}{
    			\begin{tikzpicture}
    			\pgfmathsetmacro{\w}{15}
    			\pgfmathsetmacro{\h}{8}
    			
    			\draw  (0,0) --(\w,0) -- (\w,\h) -- (0,\h) -- cycle;
    			\draw [dashed, thick] (0.3*\w,0) -- (0.3*\w,\h);
    			\draw [dashed, thick] (0.7*\w,0) -- (0.7*\w,\h);
    			
    			\foreach \x/\y/\xx/\yy/\z in {
    				0.3 / 0 /0.37 / 0.75/{\large $i_\ell$},
    				0.37 / 0 /0.42 / 0.7/,
    				0.42 / 0 /0.45 / 0.65/,
    				0.45 / 0 /0.48 / 0.6/,
    				0.48 / 0 /0.51 / 0.53/			
    			}
    			{
    				\drawTallItem[\z]{\x*\w}{\y*\h}{\xx*\w}{\yy*\h};
    			}
    			\foreach \x/\y/\xx/\yy/\z in {
    				0.51 / 0 /0.53 / 0.4/,
    				0.53 / 0 /0.55 / 0.3/,
    				0.55 / 0 /0.58 / 0.25/,	
    				0.58 / 0 /0.60 / 0.1/,
    				0.60 / 0 /0.75 / 0.2/$t'$	
    			}
    			{
    				\drawVerticalItem[\z]{\x*\w}{\y*\h}{\xx*\w}{\yy*\h};
    			}
    			\foreach \x/\y/\xx/\yy/\z in {
    				0.51 / 0.45 /0.55 / 1/,
    				0.55 / 0.4 /0.58 / 1/,		
    				0.58 / 0.3 /0.63 / 1/,
    				0.63 / 0.22 /0.7 /1/$i_r$		
    			}
    			{
    				\drawTallItem[\z]{\x*\w}{\y*\h}{\xx*\w}{\yy*\h};
    			}
    			\foreach \x/\y/\xx/\yy/\z in {
    				0.24 / 0.8 /0.4 / 1/$b'$,
    				0.49 / 0.6 /0.51 / 1/,
    				0.45 / 0.7 /0.49 /1/,
    				0.42 / 0.75 /0.45 / 1/,
    				0.4 / 0.9/0.42 / 1/				
    			}
    			{
    				\drawVerticalItem[\z]{\x*\w}{\y*\h}{\xx*\w}{\yy*\h};
    			}		
    			\end{tikzpicture}
    		}
    	\end{subfigure}
    \caption{An illustration of the repacking for a box $\boxHalf$ with height in~$(\nicefrac{1}{2}\stripHeight',\nicefrac{3}{4}\stripHeight']$. Illustrated here is the first step where we find the tallest items $i_\ell,i_r$ and generate a sub-box between these two. To the left we see the original sub-box and to the right we see the sorted result. All items either touch the top or bottom of the packing after our assignment, as is required.}
    \label{fig:SolveBoxHalfBoxpt1}
    \end{figure*}
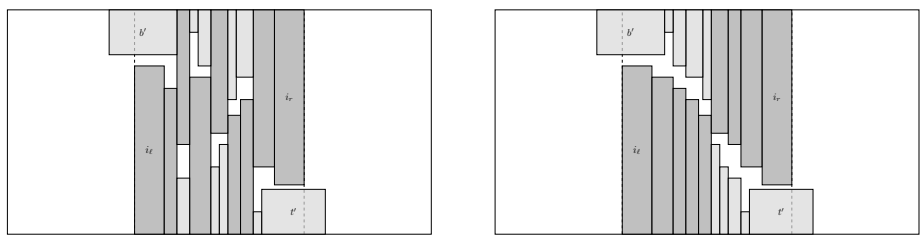   

    \begin{proof}
    Before we begin with the proof itself we have to handle the assignment of tall items to either the top or bottom of the packing. Thankfully, as there can be at most two tall items placed on the same vertical line, this is straightforward. We draw two horizontal lines, one at $\nicefrac{1}{4}\stripHeight'$ and the other at $\height(\boxHalf)-\nicefrac{1}{4}\stripHeight'$. If a tall item $t$ overlaps both of these lines, we know that no other tall item can be placed at the same x-coordinate as $t$. Thus, we shift all such items to the bottom of the box, shifting vertical items up accordingly. If any tall item only intersects the bottom line, we shift it down to the bottom of the box. Finally, if any tall item only intersects the top line, we shift it upwards such that it touches the top of the box. Clearly, no overlaps can be generated from this procedure and the packing stays feasible as we did not shift the x-coordinates of any item.
    
    We proceed by generating our pseudo-items for this box. We draw a vertical line at the beginning and end of each tall item that is placed at the bottom of the box. We generate pseudo-items by fusing the items between these lines. If any tall items are placed at the top of the box we select the one with the lowest bottom between each of these lines and let this item mark the top of the pseudo item, if enough items are present. 
    Thereby we have at most one pseudo item at each vertical line, which is placed either at the top of the box if there is a tall item below this pseudo item intersecting the line at $\nicefrac{1}{4}\stripHeight'$ or at the bottom of the box otherwise.
    
    We separate this proof into two cases. In the trivial case, no tall item overlaps the left or right border of $\boxHalf$. In this case, we know from \cite{NadiradzeW16} that we can simply sort the items touching the top in descending order of heights and items touching the bottom of the box in ascending order of height. In the case where we have tall and pseudo items of the same height, we place them such that the tall items are left of the pseudo items. Through this reordering, no items can overlap and we generate at most one box for any tall or pseudo item at the top and the same amount for items at the bottom. Because of the minimum width of tall items, we can have at most $\Oh_\eps(1)$ many boxes. It can happen that we separate some vertical items horizontally in this procedure while reordering the pseudo items. We show how to pack these items integrally again after handling all box sizes of boxes in $\boxesTallVert$.\newline\indent
    
    Let us now consider the case in which at least one tall item intersects the box border. Let $\height_b$ be the height of the tallest item touching the bottom of the box and $b_\ell$ and $b_r$ be the left- and rightmost items with height $\height_b$. Define $\height_t, t_\ell, t_r$ in the same manner. Further, let $i_\ell$ be the leftmost item in $\{b_\ell,t_\ell\}$ and $i_r$ be the rightmost in $\{b_r,t_r\}$. Let us assume, w.l.o.g. that $i_\ell=b_\ell$ and $i_r=t_r$.\newline\indent
     
    We draw a vertical line on the left border of $b_\ell$. The item cut with this line is defined as a new immovable item $b'$. We do the same on the right side of $t_r$ and name the cut item $t'$. Now, we sort all movable items between the two lines. Items touching the top are sorted in ascending order regarding their heights, while items at the bottom are sorted in descending order.\newline\indent
    This reordering ensures that no two items overlap. There is no item that overlaps $b'$ as the tallest item that was reordered at the bottom has height $h_b$. Because $b'$ was placed above $b_\ell$, this cannot lead to any overlap. The same holds for $t'$.\newline\indent
    The other points at which overlaps can occur is between these two items $b'$ and $t'$. Assume that there is a point $p=(x,y)$ at which two items $i_b$ and $i_t$ overlap. Since all items at the top to the left of $i_t$ and all items at the bottom to the right of $i_b$ have a height of at least $\height(i_t)$ and $\height(i_b)$ respectively, the total width of the strip between $b'$ and $t'$ is covered by items that must overlap if they are placed above each other in any configuration. Therefore, the strip must have had an overlap before this reordering, leading to a contradiction. Finally, this means there cannot be a new overlap created by this reordering. \newline\indent
    
    As the next step, we inspect the items placed at the top with height~$\height(b')$. We remove these items, shift the items with height smaller than $\height(b')$ to the right as far as possible, and place the removed items in this area. Clearly, the width of the packing is the same and does not lead to overlap. Furthermore, there cannot be any vertical overlap, as the item $i_r$ below $b'$ is the tallest and fits below the items of height~$\height(b')$, so all other items must do so as well. Similarly, the items with height smaller than~$\height(b')$ can obviously not lead to overlap if we shift them above smaller items. We do the same thing for items of height~$\height(t')$ on the other side. We do this to only require one box of height~$\height(b')$ and $\height(t')$. For an illustration of this step, see \cref{fig:SolveBoxHalfBoxpt1}.\newline\indent
    We have again generated $\Oh_\eps(1)$ boxes up to this point. This is because we generate at most one box for either size on the top and the bottom of the strip each. Importantly, the total number of different heights touching the top or bottom left of $i_\ell$ is reduced by one. This is due to the strip containing all items of height+$\height(i_\ell)$. The same holds for the right side of $i_r$. 
    
    The procedure described above can be repeated iteratively to fully reorder the box $\boxHalf$. To do that, we simply repeat the procedures both to the left of $i_\ell$ and the right of $i_r$ with the box borders of $\boxHalf$ as the furthest outer part. Again, by finding the outermost tallest items inside these areas, we reduce the number of heights in the remaining areas by at least one. For an illustration of this continuation, see \cref{fig:SolveBoxHalfBoxpt2}. We repeat this procedure until one of the following occurs:
    \begin{enumerate}
        \item The height of the tallest item touching the top and the height of the tallest item touching the bottom summed up is at most $\height(\boxHalf)$.
        \item The item $i_r$ becomes the immovable item overlapping the left border. To the left of $i_\ell$, the procedure ends when $i_\ell$ becomes the item overlapping the right border.
    \end{enumerate}
                   \begin{figure*}[h]
    \centering
        \begin{subfigure} [t]{0.45\textwidth}
    		\centering
    		\resizebox{0.9\textwidth}{!}{
    			\begin{tikzpicture}
    			\pgfmathsetmacro{\w}{15}
    			\pgfmathsetmacro{\h}{8}
    			
    			\draw  (0,0) --(\w,0) -- (\w,\h) -- (0,\h) -- cycle;
    			\draw [dashed, thick] (0.2*\w,0) -- (0.2*\w,\h);
    			\draw [dashed, thick] (0.5*\w,0) -- (0.5*\w,\h);
    			
    			\foreach \x/\y/\xx/\yy/\z in {
    				0.5 / 0 /0.57 / 0.75/{\large $i_r$},
    				0.46 / 0 /0.5 / 0.65/,
    				0.42 / 0 /0.46 / 0.73/,
    				0.35 / 0 /0.37 / 0.6/,
    				0.32 / 0 /0.35 / 0.63/			
    			}
    			{
    				\drawTallItem[\z]{\x*\w}{\y*\h}{\xx*\w}{\yy*\h};
    			}
    			\foreach \x/\y/\xx/\yy/\z in {
    				0.40 / 0 /0.42 / 0.35/,
    				0.37 / 0 /0.40 / 0.25/,
    				0.27 / 0 /0.32 / 0.18/,	
    				0.18 / 0 /0.27 / 0.23/$t'$	
    			}
    			{
    				\drawVerticalItem[\z]{\x*\w}{\y*\h}{\xx*\w}{\yy*\h};
    			}
    			\foreach \x/\y/\xx/\yy/\z in {
    				0.2 / 0.25 /0.3 / 1/$i_\ell$,
    				0.38 / 0.4 /0.41 / 1/				
    			}
    			{
    				\drawTallItem[\z]{\x*\w}{\y*\h}{\xx*\w}{\yy*\h};
    			}
    			\foreach \x/\y/\xx/\yy/\z in {
    				0.44 / 0.8 /0.6 / 1/$b'$,
    				0.41 / 0.75 /0.44 / 1/,
    				0.33 / 0.65 /0.38 /1/,
    				0.3 / 0.7 /0.33 / 1/				
    			}
    			{
    				\drawVerticalItem[\z]{\x*\w}{\y*\h}{\xx*\w}{\yy*\h};
    			}		
    			\end{tikzpicture}
    		}
    	\end{subfigure}
        	\begin{subfigure} [t]{0.45\textwidth}
    	\centering
    	\resizebox{0.9\textwidth}{!}{
    		\begin{tikzpicture}
    		\pgfmathsetmacro{\w}{15}
    		\pgfmathsetmacro{\h}{8}
    		
    		\draw  (0,0) --(\w,0) -- (\w,\h) -- (0,\h) -- cycle;
    		\draw [dashed, thick] (0.2*\w,0) -- (0.2*\w,\h);
    		\draw [dashed, thick] (0.5*\w,0) -- (0.5*\w,\h);
    		
    		\foreach \x/\y/\xx/\yy/\z in {
    			0.5 / 0 /0.57 / 0.75/{\large $i_r$},
    			0.46 / 0 /0.5 / 0.73/,
    			0.42 / 0 /0.46 / 0.65/,	
    			0.39 / 0 /0.42 / 0.63/,
    			0.37 / 0 /0.39 / 0.6/					
    		}
    		{
    			\drawTallItem[\z]{\x*\w}{\y*\h}{\xx*\w}{\yy*\h};
    		}
    		\foreach \x/\y/\xx/\yy/\z in {
    			0.35 / 0 /0.37 / 0.35/,
    			0.32 / 0 /0.35 / 0.25/,
    			0.27 / 0 /0.32 / 0.18/,	
    			0.18 / 0 /0.27 / 0.23/$t'$	
    		}
    		{
    			\drawVerticalItem[\z]{\x*\w}{\y*\h}{\xx*\w}{\yy*\h};
    		}
    		\foreach \x/\y/\xx/\yy/\z in {
    			0.2 / 0.25 /0.3 / 1/$i_\ell$,
    			0.3 / 0.4 /0.33 / 1/				
    		}
    		{
    			\drawTallItem[\z]{\x*\w}{\y*\h}{\xx*\w}{\yy*\h};
    		}
    		\foreach \x/\y/\xx/\yy/\z in {
    			0.44 / 0.8 /0.6 / 1/$b'$,
    			0.33 / 0.65 /0.38 /1/,
    			0.38 / 0.7 /0.41 / 1/,
    			0.41 / 0.75 /0.44 / 1/						
    		}
    		{
    			\drawVerticalItem[\z]{\x*\w}{\y*\h}{\xx*\w}{\yy*\h};
    		}		
    		\end{tikzpicture}
    	}
    \end{subfigure}
\caption{The continuation of the algorithm to solve a box $\boxHalf$. Here, $i_r$ is the same item as $i_\ell$ in the first step of the algorithm, i.e.\ we are inspecting the sub-box to the left of our starting point. Again, we find another tall item to generate this sub-box, and solve it in the same manner we solved the sub-box in the first step. This procedure is repeatable until the entire box $\boxHalf$ is packed in a structured manner.}
\label{fig:SolveBoxHalfBoxpt2}
    \end{figure*}
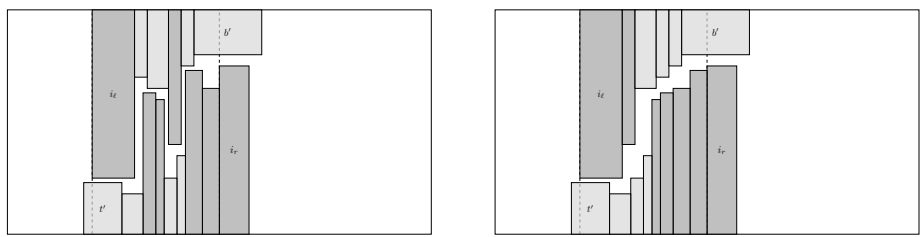
   
    If the first condition occurs, we can simply sort the items on the top and bottom in ascending order of heights, never generating any overlap. This fully solves the box $\boxHalf$, again generating at most $\Oh_\eps(1)$ boxes for tall and pseudo items.\newline\indent
    If the second condition occurs, we simply arrived at the last iteration of the reordering procedure. Afterwards, the vertical line is placed exactly at the box border, thereby finishing this procedure. Again, we generate at most  $\Oh_\eps(1)$ boxes in this final iteration. \newline\indent
    Due to the first break condition, we repeat the procedure at most until the tallest item in the packing has a height of $\height(\boxHalf/2)$. Since we reduce the number of different heights inside the box by at least one in each iteration of this procedure, we can have at most $\Oh_\eps(1)$ steps until both the tallest item at the top and bottom have at most this height. Finally, since we double this procedure on the left and right side of the initial sub-box, we create at most $\Oh_\eps(1)$ boxes for tall items, and $\Oh_\eps(1)$ boxes for pseudo items in total.
    \end{proof}
    We later fill the boxes for pseudo items with vertical items that were separated horizontally in any procedure.

\paragraph*{Boxes $B$ with $\height(B)\in (3/4\stripHeight',\stripHeight']$}
This is the most complex type of box. To solve it, we first utilize \cref{lem:Assignment} to infer some structure. 

    \begin{lem}
        Given a box $B\in \boxesTallVert$ with height $\height(B)>\nicefrac{3}{4}H$ only containing tall and vertical items, we can find an assignment for tall items such that they are fully assigned to either the bottom, top or middle of the box.
        \label{lem:Assignment}
    \end{lem}
    \begin{proof}
        Begin by partitioning the box into slices of width $1$. 
        For each such slice, we sort the tall items inside according to their height and order them tallest at the bottom ascending. 
        The vertical items inside these slices get placed at the top of the last tall item. 
        Afterward, we draw three horizontal lines throughout the box. 
        One at $\nicefrac{1}{4}\stripHeight'$, the second at $\nicefrac{1}{2}\height(B)$ and the final one at $\height(B)-\nicefrac{1}{4}\stripHeight'.$ 
        After sorting the tall items, this selection ensures that each tall item is intersected by at least one of these lines. 

        We now transform the box into a scheduling problem on three machines. 
        We discard vertical items for this transformation. Each machine represents one of the horizontal lines. 
        Item-parts get assigned to a machine. 
        One part gets assigned for every vertical slice of the box. 
        They are assigned to the machine representing the horizontal line that intersects the item in a given vertical slice. 
        The height of such a transformed item is the number of machines it is initially scheduled on. 
        We show that we can transform the schedule generated by this transformation into one where \ari{each item is assigned to the correct number of machines, items assigned to at least two machines are always assigned to the middle machine and items are fully assigned to the machines they start on.}
        
        
        After this transformation, we obtain a schedule. Let us first discuss the trivial cases. 
        If there are three items, each occupying one machine at any point in time, we can reorder them in any manner we want and fulfill the criteria. Similarly, an item requiring all three machines always fulfills the criteria above. 
        Thus, the only point at which we must carefully reorder the items is when an item occupies two machines. 
        As these items must always occupy the middle machine, we can place their remaining part either at the top or bottom. 
        Furthermore, there can only be another item with a height of $1$ scheduled in parallel to this item.\newline\indent
        Traverse the schedule from left to right. 
        Let $d$ be the first item of height $2$. 
        Schedule $d$ entirely on the machines it starts on, w.l.o.g.\ the bottom two machines. 
        All items of height $1$ scheduled in parallel get fully assigned to the top machine and are called immovable. 
        Note that, through this procedure, we can only ever have two immovable items present in any vertical slice. 
        This is because we only call items immovable when they intersect an item with a height of $2$. 
        Because the maximum height is $3$, only one such item can exist to the right and left of any vertical slice respectively. 

        We repeat the procedure described above. 
        If there is ever a point where we assign items such that two immovable items would intersect, we swap the order of the item that begins later. 
        Let $e$ be an item of height $2$ that is marked immovable in an earlier step of the procedure. 
        It intersects an item $d$ of height $2$. 
        Let $d$ be placed at the bottom and $e$ at the top. 
        Further, let $g$ be an item of height $1$ that overlaps a different item $h$ with height $2$. 
        We want to mark $g$ as immovable and assign it to the top machine, but $e$ and $g$ overlap at some point in the schedule. 
        Because we can only have two immovable items at any point in the schedule, we must be able to swap the order of $g$ and $h$ such that $g$ is assigned to the bottom machine and $h$ to the top two machines. 
        Now, neither item can intersect an immovable item.

        We repeat this procedure until we reach the end of the schedule. 
        At this point, the schedule has the three desired properties. 
        We then assign each item according to the machines it is scheduled upon. 
        See \cref{fig:AssingmentPt1} for an illustration.

        If an item is scheduled on the bottom machine, it is assigned to the bottom of the box. 
        If an item is assigned to the top machine but not the bottom, it is assigned to the top of the box.
        Finally, the remaining items are assigned to the middle of the box.

        After extending the height of the box by $\nicefrac{1}{4}\stripHeight'$, we can place all items according to their assignment without needing to slice them.
    \end{proof}
        After assigning items to their respective height-lines, we are ready to repack the entire box. To complete this repacking, we utilize an involved routine that separates the box $B$ into up to eighteen different subareas. We solve each of these areas with a unique procedure and obtain a fully structured box as the result. 
           Finally, the proof for the most complex boxes. Note that this proof is very elaborate, readers discretion is advised.
            	\begin{figure*}[t]
		\begin{subfigure}{0.3\textwidth}
			\centering
			\resizebox{0.98\textwidth}{!}{
				\begin{tikzpicture}
				\pgfmathsetmacro{\w}{6}
				\pgfmathsetmacro{\h}{4}
				\draw (0,0) -- (\w,0) --(\w,\h) --(0,\h) -- cycle;
				\drawItem{0}{0}{\w}{\h};
				\draw[dashed] (-0.05*\w,0.25*\h) node[anchor=east] {$\nicefrac{1}{4}\stripHeight'$}-- (\w,0.25*\h);
				\draw[dashed] (-0.05*\w,0.5*\h) node[anchor=east] {$\nicefrac{1}{2}\height(B)$}-- (\w,0.5*\h);
				\draw [dashed](-0.05*\w,0.75*\h) node[anchor=east] {$\height(B)-\nicefrac{1}{4}\stripHeight'$}-- (\w,0.75*\h);
				\drawTallItem[$x$]{0.0*\w}{0.0*\h}{0.1*\w}{0.47*\h};
				\drawTallItem[$c$]{0.0*\w}{0.47*\h}{0.1*\w}{0.73*\h};
				\drawTallItem[$c$]{0.1*\w}{0.75*\h}{0.2*\w}{1*\h};	
				\drawTallItem[$a$]{0.0*\w}{0.73*\h}{0.1*\w}{0.99*\h};
				
				\drawTallItem[$b$]{0.1*\w}{0.49*\h}{0.2*\w}{0.75*\h};
				\drawTallItem[$b$]{0.2*\w}{0.6*\h}{0.3*\w}{0.86*\h};				
				\drawTallItem[$y$]{0.1*\w}{0.0*\h}{0.2*\w}{0.49*\h};
				\drawTallItem[$d$]{0.2*\w}{0*\h}{0.3*\w}{0.6*\h};
				\drawTallItem[$d$]{0.3*\w}{0*\h}{0.4*\w}{0.6*\h};
				\drawTallItem[$e$]{0.3*\w}{0.6*\h}{0.4*\w}{0.87*\h};
				\drawTallItem[$e$]{0.4*\w}{0.61*\h}{0.5*\w}{0.88*\h};
				\drawTallItem[$e$]{0.5*\w}{0.61*\h}{0.6*\w}{0.88*\h};
				\drawTallItem[$f$]{0.4*\w}{0.33*\h}{0.5*\w}{0.61*\h};
				\drawTallItem[$f$]{0.5*\w}{0.33*\h}{0.6*\w}{0.61*\h};
				\drawTallItem[$g$]{0.4*\w}{0.0*\h}{0.5*\w}{0.33*\h};
				\drawTallItem[$g$]{0.5*\w}{0.0*\h}{0.6*\w}{0.33*\h};
				\drawTallItem[$g$]{0.6*\w}{0.53*\h}{0.7*\w}{0.86*\h};
				\drawTallItem[$g$]{0.7*\w}{0.53*\h}{0.8*\w}{0.86*\h};
				\drawTallItem[$h$]{0.6*\w}{0.00*\h}{0.7*\w}{0.53*\h};
				\drawTallItem[$h$]{0.7*\w}{0.00*\h}{0.8*\w}{0.53*\h};
				\drawTallItem[$z$]{0.8*\w}{0.0*\h}{0.9*\w}{0.85*\h};
				\node at (0.95*\w,0.3*\h) {$\ldots$};
				\node at (0.95*\w,0.6*\h) {$\ldots$};
				\end{tikzpicture}
			}
		\end{subfigure}
		\begin{subfigure}{0.3\textwidth}
			\centering
			\resizebox{0.98\textwidth}{!}{
				\begin{tikzpicture}
				\pgfmathsetmacro{\w}{6}
				\pgfmathsetmacro{\h}{4}
				\draw (0,0) -- (\w,0) --(\w,\h) --(0,\h) -- cycle;
				\draw (-0.05*\w,2*\h/3) node [anchor=south] {$M_3$} -- (\w,2*\h/3);
				\draw (-0.05*\w,\h/3) node [anchor=south] {$M_2$} -- (\w,\h/3);
				\draw (-0.05*\w,0) node [anchor=south] {$M_1$} -- (\w,0);
				
				\drawTallItem[$x$]{0.01*\w}{0.02*\h}{0.09*\w}{0.31*\h};
				\drawTallItem[$y$]{0.11*\w}{0.02*\h}{0.19*\w}{0.31*\h};
				\drawTallItem[$d$]{0.21*\w}{0.02*\h}{0.29*\w}{0.31*\h};
				\drawTallItem[$d$]{0.31*\w}{0.02*\h}{0.39*\w}{0.31*\h};
				\drawTallItem[$g$]{0.41*\w}{0.02*\h}{0.49*\w}{0.31*\h};
				\drawTallItem[$g$]{0.51*\w}{0.02*\h}{0.59*\w}{0.31*\h};
				\drawTallItem[$h$]{0.61*\w}{0.02*\h}{0.69*\w}{0.31*\h};
				\drawTallItem[$h$]{0.71*\w}{0.02*\h}{0.79*\w}{0.31*\h};
				\drawTallItem[$z$]{0.81*\w}{0.02*\h}{0.89*\w}{0.31*\h};
				\drawTallItem[$c$]{0.01*\w}{0.35*\h}{0.09*\w}{0.64*\h};
				\drawTallItem[$b$]{0.11*\w}{0.35*\h}{0.19*\w}{0.64*\h};
				\drawTallItem[$d$]{0.21*\w}{0.35*\h}{0.29*\w}{0.64*\h};
				\drawTallItem[$d$]{0.31*\w}{0.35*\h}{0.39*\w}{0.64*\h};
				\drawTallItem[$f$]{0.41*\w}{0.35*\h}{0.49*\w}{0.64*\h};
				\drawTallItem[$f$]{0.51*\w}{0.35*\h}{0.59*\w}{0.64*\h};
				\drawTallItem[$h$]{0.61*\w}{0.35*\h}{0.69*\w}{0.64*\h};
				\drawTallItem[$h$]{0.71*\w}{0.35*\h}{0.79*\w}{0.64*\h};
				\drawTallItem[$z$]{0.81*\w}{0.35*\h}{0.89*\w}{0.64*\h};
				\drawTallItem[$a$]{0.01*\w}{0.686*\h}{0.09*\w}{0.976*\h};
				\drawTallItem[$c$]{0.11*\w}{0.686*\h}{0.19*\w}{0.976*\h};
				\drawTallItem[$b$]{0.21*\w}{0.686*\h}{0.29*\w}{0.976*\h};
				\drawTallItem[$e$]{0.31*\w}{0.686*\h}{0.39*\w}{0.976*\h};
				\drawTallItem[$e$]{0.41*\w}{0.686*\h}{0.49*\w}{0.976*\h};
				\drawTallItem[$e$]{0.51*\w}{0.686*\h}{0.59*\w}{0.976*\h};
				\drawTallItem[$g$]{0.61*\w}{0.686*\h}{0.69*\w}{0.976*\h};
				\drawTallItem[$g$]{0.71*\w}{0.686*\h}{0.79*\w}{0.976*\h};
				\drawTallItem[$z$]{0.81*\w}{0.686*\h}{0.89*\w}{0.976*\h};
				\node at (0.95*\w,0.8*\h) {$\ldots$};
				\node at (0.95*\w,0.5*\h) {$\ldots$};
				\node at (0.95*\w,0.2*\h) {$\ldots$};
				\end{tikzpicture}
			}
		\end{subfigure}
  \begin{subfigure}{0.3\textwidth}
        		\centering
        		\resizebox{0.98\textwidth}{!}{
        			\begin{tikzpicture}
        			\pgfmathsetmacro{\w}{6}
        			\pgfmathsetmacro{\h}{4}
        			\draw (0,0) -- (\w,0) --(\w,\h) --(0,\h) -- cycle;
        			\draw (-0.05*\w,2*\h/3) node [anchor=south] {$M_3$} -- (\w,2*\h/3);
        			\draw (-0.05*\w,\h/3) node [anchor=south] {$M_2$} -- (\w,\h/3);
        			\draw (-0.05*\w,0) node [anchor=south] {$M_1$} -- (\w,0);
        			
        			\drawTallItem[$x$]{0.01*\w}{0.02*\h}{0.09*\w}{0.31*\h};
        			\drawTallItem[$y$]{0.11*\w}{0.02*\h}{0.19*\w}{0.31*\h};
        			\drawTallItem[$d$]{0.21*\w}{0.02*\h}{0.29*\w}{0.31*\h};
        			\drawTallItem[$d$]{0.31*\w}{0.02*\h}{0.39*\w}{0.31*\h};
        			\drawTallItem[$\mathbf{G}$]{0.41*\w}{0.02*\h}{0.49*\w}{0.31*\h};
        			\drawTallItem[$\mathbf{G}$]{0.51*\w}{0.02*\h}{0.59*\w}{0.31*\h};
        			\drawTallItem[$\mathbf{G}$]{0.61*\w}{0.02*\h}{0.69*\w}{0.31*\h};
        			\drawTallItem[$\mathbf{G}$]{0.71*\w}{0.02*\h}{0.79*\w}{0.31*\h};
        			\drawTallItem[$z$]{0.81*\w}{0.02*\h}{0.89*\w}{0.31*\h};
        			\drawTallItem[$c$]{0.01*\w}{0.35*\h}{0.09*\w}{0.64*\h};
        			\drawTallItem[$c$]{0.11*\w}{0.35*\h}{0.19*\w}{0.64*\h};
        			\drawTallItem[$d$]{0.21*\w}{0.35*\h}{0.29*\w}{0.64*\h};
        			\drawTallItem[$d$]{0.31*\w}{0.35*\h}{0.39*\w}{0.64*\h};
        			\drawTallItem[$f$]{0.41*\w}{0.35*\h}{0.49*\w}{0.64*\h};
        			\drawTallItem[$f$]{0.51*\w}{0.35*\h}{0.59*\w}{0.64*\h};
        			\drawTallItem[$h$]{0.61*\w}{0.35*\h}{0.69*\w}{0.64*\h};
        			\drawTallItem[$h$]{0.71*\w}{0.35*\h}{0.79*\w}{0.64*\h};
        			\drawTallItem[$z$]{0.81*\w}{0.35*\h}{0.89*\w}{0.64*\h};
        			\drawTallItem[$a$]{0.01*\w}{0.686*\h}{0.09*\w}{0.976*\h};
        			\drawTallItem[$\mathbf{B}$]{0.11*\w}{0.686*\h}{0.19*\w}{0.976*\h};
        			\drawTallItem[$\mathbf{B}$]{0.21*\w}{0.686*\h}{0.29*\w}{0.976*\h};
        			\drawTallItem[$\mathbf{E}$]{0.31*\w}{0.686*\h}{0.39*\w}{0.976*\h};
        			\drawTallItem[$\mathbf{E}$]{0.41*\w}{0.686*\h}{0.49*\w}{0.976*\h};
        			\drawTallItem[$\mathbf{E}$]{0.51*\w}{0.686*\h}{0.59*\w}{0.976*\h};
        			\drawTallItem[$h$]{0.61*\w}{0.686*\h}{0.69*\w}{0.976*\h};
        			\drawTallItem[$h$]{0.71*\w}{0.686*\h}{0.79*\w}{0.976*\h};
        			\drawTallItem[$z$]{0.81*\w}{0.686*\h}{0.89*\w}{0.976*\h};
        			\node at (0.95*\w,0.8*\h) {$\ldots$};
        			\node at (0.95*\w,0.5*\h) {$\ldots$};
        			\node at (0.95*\w,0.2*\h) {$\ldots$};
        			\end{tikzpicture}
        		}
        	\end{subfigure}
	 \caption{An illustration of the transformation from a packed box of tall items to the desired scheduling problem from \cref{lem:Assignment}. 
  We see the sorted original packing on the left, the initially generated schedule in the middle, and the final schedule to the right. 
  Items with bold and capitalized names are those we mark as immovable during the transformation of the schedule.}
	 \label{fig:AssingmentPt1}
	\end{figure*}
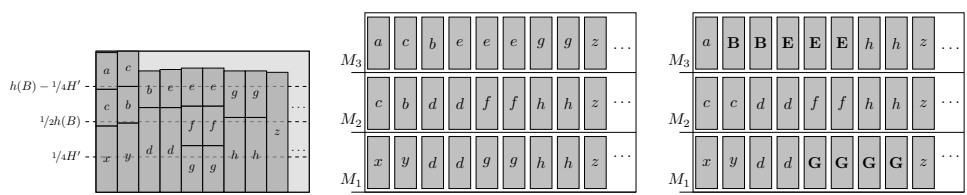  
    \begin{lem}
        Let $B\in \boxesTallVert$ be a box with height $h(B)>\nicefrac{3}{4}H$ such that at most two tall items overlap the left or right box border, while at most three such items overlap any box border. By adding $\nicefrac{1}{4}H$ to $B$'s height, we can rearrange the items in $B$ such that we generate at most $\mathcal{O}(N^2)$ boxes for tall and  $\mathcal{O}(N^2)$ boxes for vertical items without moving the immovable items. The vertical items are sliced and may be separated horizontally while tall items are not separated.
        \label{lem:ReorderTallVertBox34}
    \end{lem}
    \begin{proof}
    Disclaimer: All figures displayed in this proof are adapted from~\cite{StripPacking54} with explicit permission of the authors.
    	This procedure is an evolution of similar ideas in \cite{StripPacking54}.
        We present a reordering strategy for items in these boxes. Let $h(B)$ be the height of $B$. To keep notation simple, we assume that each box has its lower border at $0$, which can be achieved through slicing.\newline\indent
        \subparagraph*{Step 1: Shifting the items} Let us first consider the tall items overlapping the box border. We slice these items down to the bottom of the box, such that they are placed integrally. Let the one that protrudes further inwards, i.e.\ the leftmost on the right or the rightmost on the left, be placed below the other item if there are two on a side. Next, we inspect the three horizontal lines inside the box: one at~$\nicefrac{1}{4}H$, the second at~$\nicefrac{1}{2}\height(B)$, and the final one at~$\height(B)-\nicefrac{1}{4}H$. Since each tall item has a height of at least $\nicefrac{1}{4}H$, any optimal packing must have these items intersect at least one of these lines. We assign each tall item to at least one of these lines using \cref{lem:Assignment}.\newline\indent
        After this assignment, we shift each tall item assigned to the line at $\nicefrac{1}{4}H$ downward to the bottom of the box, as well as parts of all vertical items intersected by this line. Vertical items where only parts are sliced downwards here might be separated horizontally later. Afterwards, we repeat the same procedure for tall items assigned to the horizontal line at $h(B)-\nicefrac{1}{4}H$ and vertical items intersected by that line, only placing them such that their upper border touches the top of the box. For each tall item $i$ that is assigned to both of these horizontal lines, i.e.\ items with height greater than $h(B)-\nicefrac{1}{4}H$, we create a pseudo item with height $\height(B)$ and width $\width(i)$.\newline\indent
        Next, we extend the box upwards by $\nicefrac{1}{4}H$, and shift all items intersected by the line $h(B)-\nicefrac{1}{4}H$ up by $\nicefrac{1}{4}H$. All pseudo items between the lines $h(B)-\nicefrac{1}{4}H$ and $1/2h(B)$ are shifted up so their lower border touches $h(B)-\nicefrac{1}{4}H$. Afterwards, shift all tall items assigned only to the line at $1/2h(B)$ so that their top touches $\height(B)-\nicefrac{1}{4}H$. Do the same for all pseudo items that were not shifted yet and that are intersected by that line. Through both of these shifts, we simply moved slices inside the optimal packing upwards or downwards. As such, no overlaps can exist at this point.\newline\indent
         \subparagraph*{Step 1.5: Fusing pseudo items} As an intermediary step, we need to fuse some pseudo items together. As a result of this, we want to gain the property that all tall and pseudo items have at least one border that touches the horizontal lines at $0$, $\height(B)-\nicefrac{1}{4}H$ or $\height(B)+\nicefrac{1}{4}$. For now, there can still be pseudo items located between $\nicefrac{1}{4}H$ and $\nicefrac{1}{2}\height(B)$. However, these items can only exist if there is a tall item touching the top of the packing, and a different tall item touching the bottom of the packing at this width, as they would have been sliced downwards otherwise. Consider two consecutive vertical lines we had drawn to generate the pseudo items. If a  tall item overlaps the vertical strip between these lines, its right and left borders lie either on the strips borders or outside of the strip, due to the definition of pseudo items. \newline\indent
         \textit{Case 1:} In the first case, there are three tall items $t_1,t_2$ and $t_3$ from bottom to top that overlap the strip. Since the first two shifting steps have already occurred, $t_1$ must have its lower border at $0$, $t_2$ its upper border at $\height(B)-\nicefrac{1}{4}H$ and $t_3$ its upper border at $\height(B)+\nicefrac{1}{4}H$. Therefore, there can be at most two pseudo items. One between $t_1$ and $t_2$, and one between $t_2$ and $t_3$. We stack these pseudo items together such that their lower border begins at $\height(B)-\nicefrac{1}{4}H$. This is possible without overlapping $t_3$ due to the total height of the pseudo items. Taken together, both pseudo items have a height of at most $\height(B)-\height(t_1)-\height(t_2)-\height(t_3)$. The total area not occupied by tall items is $\height(B)-\height(t_1)-\height(t_2)-\height(t_3)+\nicefrac{1}{4}H$ as we have extended the box. The distance between $t_1$ and $t_2$ is at most $\nicefrac{1}{4}H$ since they are both tall items and $t_1$ has its lower border at $0$, while $t_2$ has its upper border at $\height(B)-\nicefrac{1}{4}H$. Thus, the distance between $t_2$ and $t_3$ is at least $\height(B)-\height(t_1)-\height(t_2)-\height(t_3)$. This is because we added $\nicefrac{1}{4}H$ to the height of the box, but the area between $t_1$ and $t_2$ has a lesser height. Therefore, we can fuse and place the pseudo items as desired.\newline\indent
         \textit{Case 2:} Next we consider the case where there is only one tall item~$t_1$ touching the bottom of the box, and one tall item with height at least $\nicefrac{1}{2}\height(B)$ touching the top of the box. Since there are two tall items placed above each other, $t_2$ must have a height of at most $\nicefrac{3}{4}\height(B)$. Furthermore, there is at most one pseudo item, and it is placed between $\nicefrac{1}{4}H$ and $\nicefrac{1}{2}\height(B)$. We shift this item up until its bottom border touches $\nicefrac{1}{2}\height(B)$. This does not construct any overlap because the distance between the top of $t_1$ and $\nicefrac{1}{2}\height(B)$ is less than $\nicefrac{1}{4}H$, which is the distance we extended the box by.\newline\indent
         After this, we further fuse each tall item $t$ with height greater than $\nicefrac{1}{2}\height(B)$ to all pseudo items placed below the item such that their bottom borders touch $\nicefrac{1}{2}\height(B)$. We do this by generating a new pseudo item with height $\nicefrac{1}{2}\height(B)+\nicefrac{1}{4}H$ and width $\width(t)$. Thus, all such pseudo items have their top border at $\height(B)\nicefrac{1}{4}H$ and their lower border at $\nicefrac{1}{2}\height(B)$.\newline\indent
          \textit{Case 3:} In the final case, we have two tall items~$t_1$ and $t_2$ and two pseudo items. One of the tall items~$t_1$ or $t_2$ touches the top or bottom of the box, while the other ends at $\height(B)-\nicefrac{1}{4}H$. Therefore, the distance between the tall items has to be smaller than $\nicefrac{1}{4}H$. Additionally, one of the pseudo items touches the top or bottom of the box, while the other is positioned between $t_1$ and $t_2$. Since the distance between $t_1$ and $t_2$ is less than $\nicefrac{1}{4}H$ one of the distances between the box border and the lower border of $t_1$ or the top border of $t_2$ is at least $\height(B)-\height(t_1)-\height(t_2)$. Therefore, we can fuse both pseudo items together. We place this fused pseudo item depending on where the original pseudo item touched the box border. If it was at the top, the fused pseudo item gets placed such that its top touches the top of the box. Otherwise, place the fused pseudo item such that its bottom border touches the bottom of the box.\newline\indent
         With the fusing of pseudo items out of the way, our box now has the desired property that all tall or pseudo items touch one of the horizontal lines at $0$, $\height(B)-\nicefrac{1}{4}H$ or $\height(B)+\nicefrac{1}{4}$. Furthermore, we can assume that each item $t$ touching the bottom of the box with $\height(t)>\nicefrac{1}{2}\height(B)$ has height $\height(B)-\nicefrac{1}{4}H$: There can be at most two items above $t$, one tall and one pseudo item. The pseudo item now has its lower border at $\height(B)-\nicefrac{1}{4}H$. Therefore, we can extend $t$ to the horizontal line $\height(B)-\nicefrac{1}{4}H$.\newline\indent
         		\begin{figure*}[t]
         	\centering
         	\begin{subfigure}[t]{0.4\textwidth}
         		\centering
         		\begin{tikzpicture}
         		\pgfmathsetmacro{\w}{2.7}
         		\pgfmathsetmacro{\h}{4.8}
         		\pgfmathsetmacro{\hprime}{4.5}
         		
         		\drawVerticalItem{0*\w}{0}{1*\w}{\hprime};
         		
         		\foreach \x/\y/\xx/\yy in {
         			-0.03 /0.02/0.025/0.33,
         			-0.025/0.35/0.05/0.625,
         			1.075 /0.05/0.95/0.34,
         			1.1   /0.36/0.975/0.62,
         			0.175 /0.02 /0.325/0.33,
         			0.35  /0.0  /0.45 /0.29,
         			0.45  /0.14 /0.5  /0.67,
         			0.5   /0.03 /0.575/0.32,
         			0.6   /0.01 /0.7  /0.36,
         			0.7   /0.19 /0.75 /0.71,
         			0.75  /0.0  /0.85 /0.3,
         			0.875 /0.01 /0.925/0.3,
         			0.3   /0.36 /0.4  /0.65,
         			0.9   /0.44 /0.975/0.78,
         			0.05  /0.32 /0.125/0.65,
         			0.075 /0.29 /0.15 /0.01,
         			0.15  /0.33 /0.175/0.63,
         			0.225 /0.34 /0.3  /0.62,
         			0.55  /0.34 /0.6  /0.69,
         			0.0   /0.7  /0.025/0.99,
         			0.025 /0.65 /0.075/0.96,
         			0.075 /0.66 /0.125/0.95,
         			0.125 /0.69 /0.2  /1,
         			0.2   /0.345/0.225/0.975,
         			0.225 /0.65 /0.25 /0.99,
         			0.25  /0.63 /0.275/0.97,
         			0.275 /0.67 /0.325/1,
         			0.325 /0.66 /0.35 /1,
         			0.35  /0.68 /0.375/1,
         			0.375 /0.67 /0.425/1,
         			0.475 /0.69 /0.575/1,
         			0.425 /0.33 /0.45 /0.95,
         			0.6   /0.39 /0.625/0.99,
         			0.825 /0.31 /0.9  /0.94
         		}
         		{
         			\drawTallItem{\x*\w}{\y*\hprime}{\xx*\w}{\yy*\hprime};
         		}
         		
         		\foreach \x/\y/\xx/\yy in {
         			0    /0   /0.025/0.35,
         			0.950/0   /1    /0.34,
         			0.975/0.34/1    /0.62
         		}{
         			\draw[pattern = north west lines] (\x*\w,\y*\hprime) rectangle (\xx*\w,\yy*\hprime);
         		}
         		
         		\foreach \y/\z in {
         			0.5*\hprime 		/ $\nicefrac{1}{2}\height(B)$,
         			0.25*\h     		/ $\nicefrac{1}{4}H$,
         			\hprime - 0.25*\h	/ $\height(B) - \nicefrac{1}{4}H$,
         			\hprime				/ $\height(B)$
         		}{
         			\draw[dotted] (-0.1*\w,\y)node[left]{\z} --(1.1*\w,\y);
         		}
         		
         		\end{tikzpicture}
         		\caption{The introduction of unmovable pseudo items (hatched area)}
         		\label{fig:sub:IntroducingPseudoItems}
         	\end{subfigure}
         	\hfill
         	\centering
         	\begin{subfigure}[t]{0.24\textwidth}
         		\centering
         		\begin{tikzpicture}
         		
         		\pgfmathsetmacro{\w}{2.7}
         		\pgfmathsetmacro{\h}{4.8}
         		\pgfmathsetmacro{\hprime}{4.5}
         		\draw (0*\w,0) rectangle (1*\w,\hprime);

         		\foreach \x/\y/\xx/\yy in {
         			-0.03 /0.02/0.025/0.33,
         			-0.025/0.35/0.05/0.625,
         			1.075 /0.05/0.95/0.34,
         			1.1   /0.36/0.975/0.62,
         			0.175 /0.0 /0.325/0.31,
         			0.35  /0.0 /0.45 /0.29,
         			0.45  /0.0 /0.5  /0.53,
         			0.5   /0.0 /0.575/0.29,
         			0.6   /0.0 /0.7  /0.35,
         			0.7   /0.0 /0.75 /0.52,
         			0.75  /0.0 /0.85 /0.3,
         			0.875 /0.0 /0.925/0.29,
         			0.075 /0.0 /0.15 /0.28,
         			0.05  /0.32 /0.125/0.65,
         			0.15  /0.33 /0.175/0.63,
         			0.225 /0.34 /0.3  /0.62,
         			0.3   /0.36 /0.4  /0.65,
         			0.55  /0.34 /0.6  /0.69,
         			0.0   /0.71 /0.025/1,
         			0.025 /0.69 /0.075/1,
         			0.075 /0.71 /0.125/1,
         			0.125 /0.69 /0.2  /1,
         			0.2   /0.37 /0.225/1,
         			0.225 /0.66 /0.25 /1,
         			0.25  /0.66 /0.275/1,
         			0.275 /0.67 /0.325/1,
         			0.325 /0.66 /0.35 /1,
         			0.35  /0.68 /0.375/1,
         			0.375 /0.67 /0.425/1,
         			0.475 /0.69 /0.575/1,
         			0.425 /0.38 /0.45 /1,
         			0.6   /0.40 /0.625/1,
         			0.825 /0.37 /0.9  /1,
         			0.9   /0.66 /0.975/1
         		}
         		{
         			\drawTallItem{\x*\w}{\y*\hprime}{\xx*\w}{\yy*\hprime};
         		}
         		
         		\foreach \x/\y/\xx/\yy in {
         			0    /0   /0.025/0.35,
         			0.950/0   /1    /0.34,
         			0.975/0.34/1    /0.62
         		}{
         			\draw[pattern = north west lines] (\x*\w,\y*\hprime) rectangle (\xx*\w,\yy*\hprime);
         		}

         		\foreach \x/\y/\xx/\yy in {
         			0.0  /0.625/0.025/0.71,
         			0.025/0.625/0.05 /0.69,
         			0.025/0.0  /0.05 /0.35,
         			0.05 /0.65 /0.075/0.69,
         			0.05 /0.0  /0.075/0.32,
         			0.075/0.65 /0.125/0.71,
         			0.075/0.28 /0.125/0.32,
         			0.125/0.28 /0.15 /0.69,
         			0.15 /0.63 /0.175/0.69,
         			0.15 /0.0  /0.175/0.33,
         			0.175/0.31 /0.2  /0.69,
         			0.2  /0.31 /0.225/0.37,
         			0.225/0.31 /0.25 /0.34,
         			0.225/0.62 /0.25 /0.66,
         			0.25 /0.31 /0.275/0.34,
         			0.25 /0.62 /0.275/0.66,
         			0.275/0.31 /0.3  /0.34,
         			0.275/0.62 /0.3  /0.67,
         			0.3  /0.31 /0.325/0.36,
         			0.3  /0.65 /0.325/0.67,
         			0.325/0.0  /0.35 /0.36,
         			0.325/0.65 /0.35 /0.66,
         			0.35 /0.29 /0.375/0.36,
         			0.35 /0.65 /0.375/0.68,
         			0.375/0.29 /0.4  /0.36,
         			0.375/0.65 /0.4  /0.67,
         			0.4  /0.29 /0.425/0.67,
         			0.425/0.29 /0.45 /0.38,
         			0.45 /0.53 /0.475/1,
         			0.475/0.53 /0.5  /0.69,
         			0.5  /0.29 /0.55 /0.69,
         			0.55 /0.29 /0.575/0.34,
         			0.575/0.0  /0.6  /0.34,
         			0.575/0.69 /0.6  /1,
         			0.6  /0.35 /0.625/0.4,
         			0.625/0.35 /0.7  /1,
         			0.7  /0.52 /0.75 /1,
         			0.75 /0.30 /0.825/1,
         			0.825/0.30 /0.85 /0.37,
         			0.85 /0.0  /0.875/0.37,
         			0.875/0.29 /0.9  /0.37,
         			0.9  /0.29 /0.925/0.66,
         			0.925/0.0  /0.95 /0.66,
         			0.95 /0.34 /0.975/0.66,
         			0.975/0.62 /1    /1
         		}
         		{
         			\drawVerticalItem{\x*\w}{\y*\hprime}{\xx*\w}{\yy*\hprime};
         		}

         		\foreach \y/\z in {
         			0.5*\hprime 		/ $\nicefrac{1}{2}\height(B)$,
         			0.25*\h     		/ $\nicefrac{1}{4}H$,
         			\hprime - 0.25*\h	/ $\nicefrac{1}{2}\height(B)$,
         			\hprime				/ $\height(B)$
         		}{
         			\draw[dotted] (-0.1*\w,\y) -- (1.1*\w,\y);
         		}
         		
         		
         		\end{tikzpicture}
         		\caption{The first shift and the introduction of pseudo items}
         		\label{fig:sub:TheFirstShift}
         	\end{subfigure}
         	\hfill
         	\begin{subfigure}[t]{0.24\textwidth}
         		\centering
         		\begin{tikzpicture}
         		
         		\pgfmathsetmacro{\w}{2.7}
         		\pgfmathsetmacro{\h}{4.8}
         		\pgfmathsetmacro{\hprime}{4.5}
         		
         		\draw (0*\w,0) rectangle (1*\w,\hprime +\h/4);
         		
         		\foreach \x/\y/\xx/\yy in {
         			-0.03 /0.00/0.025/0.31,
         			-0.025/0.35/0.05/0.625,
         			1.075 /0.00/0.95/0.29,
         			1.1   /0.29/0.975/0.55
         		}
         		{
         			\drawTallItem{\x*\w}{\y*\hprime}{\xx*\w}{\yy*\hprime};
         		}
         		
         		\foreach \x/\y/\xx/\yy in {
         			0    /0   /0.025/0.35,
         			0.950/0   /1    /0.29,
         			0.975/0.29/1    /0.55
         		}{
         			\draw[pattern = north west lines] (\x*\w,\y*\hprime) rectangle (\xx*\w,\yy*\hprime);
         		}
         		
         		\foreach \x/\y/\xx/\yy in {
         			0.175 /0.0 /0.325/0.31,
         			0.35  /0.0 /0.45 /0.29,
         			0.45  /0.0 /0.5  /0.53,
         			0.5   /0.0 /0.575/0.29,
         			0.6   /0.0 /0.7  /0.35,
         			0.7   /0.0 /0.75 /0.52,
         			0.75  /0.0 /0.85 /0.3,
         			0.875 /0.0 /0.925/0.29,
         			0.075 /0.0 /0.15 /0.28
         		}
         		{
         			\drawTallItem{\x*\w}{\y*\hprime}{\xx*\w}{\yy*\hprime};
         		}
         		
         		\foreach \x/\y/\xx/\yy in {
         			0.05 /0.67/0.125/1,
         			0.15 /0.70/0.175/1,
         			0.225/0.72/0.3  /1,
         			0.3  /0.71/0.4  /1,
         			0.55 /0.65/0.6  /1
         		}
         		{
         			\drawTallItem{\x*\w}{\y*\hprime-\h/4}{\xx*\w}{\yy*\hprime-\h/4};
         		}
         		
         		\foreach \x/\y/\xx/\yy in {
         			0.0   /0.71 /0.025/1,
         			0.025 /0.69 /0.075/1,
         			0.075 /0.71 /0.125/1,
         			0.125 /0.69 /0.2  /1,
         			0.2   /0.37 /0.225/1,
         			0.225 /0.66 /0.25 /1,
         			0.25  /0.66 /0.275/1,
         			0.275 /0.67 /0.325/1,
         			0.325 /0.66 /0.35 /1,
         			0.35  /0.68 /0.375/1,
         			0.375 /0.67 /0.425/1,
         			0.475 /0.69 /0.575/1,
         			0.425 /0.38 /0.45 /1,
         			0.6   /0.40 /0.625/1,
         			0.825 /0.37 /0.9  /1,
         			0.9   /0.66 /0.975/1
         		}
         		{
         			\drawTallItem{\x*\w}{\y*\hprime+\h/4}{\xx*\w}{\yy*\hprime+\h/4};
         		}

         		\foreach \x/\y/\xx/\yy in {
         			0.45 /0.53/0.475/1,
         			0.575/0.69/0.6  /1,
         			0.625/0.35/0.7  /1,
         			0.7  /0.52 /0.75 /1,
         			0.75 /0.30/0.825/1,
         			0.975/0.62/1    /1
         		}
         		{
         			\drawVerticalItem{\x*\w}{\y*\hprime+\h/4}{\xx*\w}{\yy*\hprime+\h/4};
         		}
         		
         		\foreach \x/\y/\xx/\yy in {
         			0.0  /1/0.025/1.085,
         			0.025/1/0.05 /1.065,
         			0.05 /1/0.075/1.04 ,
         			0.075/1/0.125/1.06 ,
         			0.15 /1/0.175/1.06 ,
         			0.225/1/0.25 /1.04 ,
         			0.25 /1/0.275/1.04 ,
         			0.275/1/0.3  /1.05 ,
         			0.3  /1/0.325/1.02 ,
         			0.325/1/0.35 /1.01 ,
         			0.35 /1/0.375/1.03 ,
         			0.375/1/0.4  /1.02 ,
         			0.475/1/0.5  /1.16 
         		}
         		{
         			\drawVerticalItem{\x*\w}{\y*\hprime-\h/4}{\xx*\w}{\yy*\hprime-\h/4};
         		}
         		
         		\foreach \x/\y/\xx/\yy in {
         			0.125/0.59/0.15 /1,
         			0.175/0.62/0.2  /1,
         			0.4  /0.62/0.425/1,
         			0.5  /0.60/0.55 /1,
         			0.9  /0.63/0.925/1,
         			0.95 /0.68/0.975/1
         		}
         		{
         			\drawVerticalItem{\x*\w}{\y*\hprime-\h/4}{\xx*\w}{\yy*\hprime-\h/4};
         		}
         		
         		\foreach \x/\y/\xx/\yy in {
         			0.075/0.28/0.125/0.32,
         			0.2  /0.31/0.225/0.37,
         			0.225/0.31/0.25 /0.34,
         			0.25 /0.31/0.275/0.34,
         			0.275/0.31/0.3  /0.34,
         			0.3  /0.31/0.325/0.36,
         			0.35 /0.29/0.375/0.36,
         			0.375/0.29/0.4  /0.36,
         			0.425/0.29/0.45 /0.38,
         			0.55 /0.29/0.575/0.34,
         			0.6  /0.35/0.625/0.4 ,
         			0.825/0.30/0.85 /0.37,
         			0.875/0.29/0.9  /0.37
         		}
         		{
         			\drawVerticalItem{\x*\w}{\y*\hprime}{\xx*\w}{\yy*\hprime};
         		}
         		
         		\foreach \x/\y/\xx/\yy in {
         			0.025/0.0/0.05 /0.35,
         			0.05 /0.0/0.075/0.32,
         			0.15 /0.0/0.175/0.33,
         			0.325/0.0/0.35 /0.36,
         			0.575/0.0/0.6  /0.34,
         			0.85 /0.0/0.875/0.37,
         			0.925/0.0/0.95 /0.66
         		}
         		{
         			\drawVerticalItem{\x*\w}{\y*\hprime}{\xx*\w}{\yy*\hprime};
         		}
         		
         		\foreach \y/\z in {
         			0.5*\hprime 		/ $\nicefrac{1}{2}\height(B)$,
         			0.25*\h     		/ $\nicefrac{1}{4}H$,
         			\hprime - 0.25*\h	/ $\nicefrac{1}{2}\height(B)$,
         			\hprime				/ $\height(B)$
         		}{
         			\draw[dotted] (-0.1*\w,\y) --(1.1*\w,\y);
         		}
         		
         		
         		\end{tikzpicture}
         		\caption{The second shift}
         		\label{fig:sub:TheSecondShift}
         	\end{subfigure}
         	\caption{An overview of the shifting steps. First we define immovable pseudo items; second we shift all the tall items assigned to the  horizontal line $\nicefrac{1}{4}H$ and $\height(B)- \nicefrac{1}{4}H$ to the bottom and top and introduce pseudo items; last we extend the box by $\nicefrac{1}{4}H$ such that all tall items touch one of three horizontal lines, and shift the immovable items to the bottom of the packing. These illustrations are taken, and minorly adjusted, from \cite{StripPacking54}.}
         	\label{fig:schiftGeneral1}
         \end{figure*}
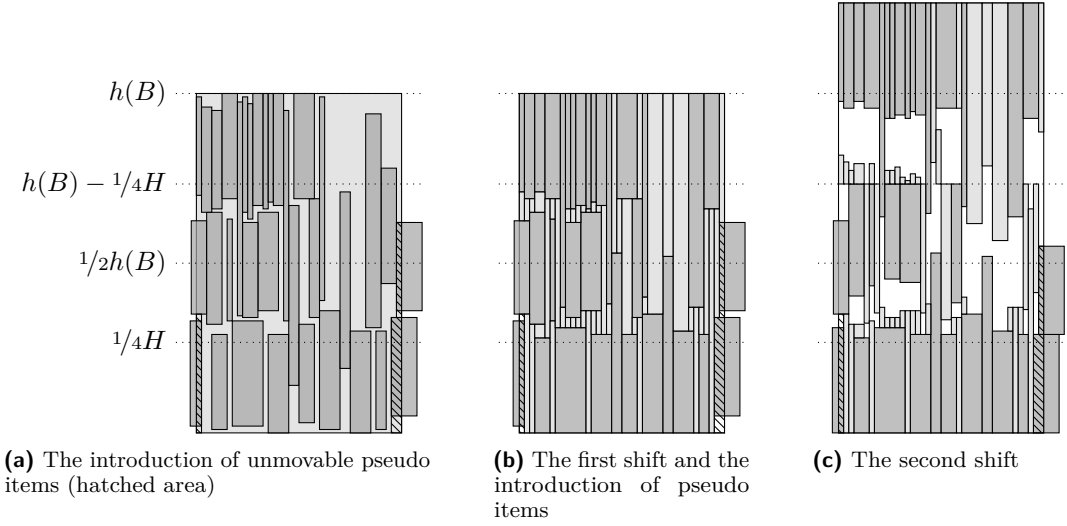
         After this shift, each movable item has one border at one of the horizontal lines $0$, $\height(B)-\nicefrac{1}{4}H$ or $\height(B)+\nicefrac{1}{4}$. Additionally, only items with height at least $\nicefrac{1}{2}\height(B)+\nicefrac{1}{4}H$ are crossing the line at $\height(B)-\nicefrac{1}{4}H$.\newline\indent
         \subparagraph*{Step 2: Reordering.} For now, assume there is no item with height $\height(B)$ in $B$. We handle the case where there is at least one of those later. In the following we reorder items step by step, by considering a constant number of smaller subareas of the box. We number these subareas from one to nine. We generate these subareas symmetrically on either side of the box, and call the $i$-th subarea $B_{r,i}$ or $B_{\ell,i}$ on the right or left respectively. We only describe the steps for the left side in the following. The right side gets handled analogously.\newline\indent
          \textit{Area $B_{\ell,1}$:} Consider the leftmost (pseudo) item $i_\ell$ with height $\nicefrac{1}{2}\height(B)$ in $B$ and let $i_r$ be the rightmost of these items. Left of $i_\ell$ inside $B$, there is no item intersecting the horizontal line $\height(B)-\nicefrac{1}{4}H$ since only (pseudo) items with height $\nicefrac{1}{2}\height(B)+\nicefrac{1}{4}H$ touching $\height(B)+\nicefrac{1}{4}H$ overlap this horizontal line. See \cref{fig:sub:AreaB1} for an illustration of this area. Therefore, each item left of $i_\ell$ above $\height(B)-\nicefrac{1}{4}H$ either touches $\height(B)-\nicefrac{1}{4}H$ with its lower border or $\height(B)+\nicefrac{1}{4}H$ with its upper border. Since there is no item intersecting the left box border at this height, we can sort the items left of $i_\ell$ touching $\height(B)+\nicefrac{1}{4}H$ in descending order and the items touching $\height(B)-\nicefrac{1}{4}H$ in ascending order of heights, without constructing any overlap. We may, however, have separated some vertical items such that they are no longer placed contiguously. These items get handled later on. We do the same for the area right of $i_r$ and call these areas~$B_{\ell,1}$ and $B_{r,1}$.\newline\indent
          \textit{Area $B_{\ell,2}$:} We draw a vertical line at the left border of $i_{\ell}$ to the bottom of the box. If this line cuts a tall item $\ell_b$ at the bottom, it defines a new unmovable item. Let us consider the area between the line and the left box border below $\height(B)-\nicefrac{1}{4}H$. We call this area $B_{\ell,2}$. In $B_{\ell,2}$, each (partial) item either touches the horizontal line at $0$ or $\height(B)-\nicefrac{1}{4}H$ and on each side, there are at most two tall unmovable items. We extend the unmovable item intersecting $\height(B)/2$ to the top, such that it touches the horizontal line at $\height(B)-\nicefrac{1}{4}H$ and reorder this box according to \cref{lem:ReorderTallVertBox12}. We do the same for the right side of $i_r$. The area is illustrated in \cref{fig:sub:AreaB2}. \newline\indent
         \textit{Cases for $i_\ell$ and $i_r$:} If $i_\ell$ and $i_r$ do not exist, there are no items overlapping the horizontal line $\height(B)-\nicefrac{1}{4}H$ and we can partition the box in two areas $B_1$ and $B_2$. We reorder $B_1$ as described for $B_{\ell,1}$ and $B_2$ as described for $B_{\ell,2}$. In the case that $i_\ell$ is the same item as $i_r$, we introduce the areas $B_{\ell,1},$ $B_{\ell,2}$, $B_{r,1}$ and  $B_{r,2}$ as described and order the items completely below $i_\ell$ such that items with the same height are positioned next to each other.\newline\indent
         \begin{figure*}[t]
         	\centering
         	\begin{subfigure}[t]{0.37\textwidth}
         		\centering
         		\begin{tikzpicture}
         		\pgfmathsetmacro{\w}{2.7}
         		\pgfmathsetmacro{\h}{4.8}
         		\pgfmathsetmacro{\hprime}{4.5}

         		\draw (0*\w,0) rectangle (1*\w,\hprime +\h/4);
         		
         		\foreach \x/\y/\xx/\yy in {
         			-0.03 /0.00/0.025/0.31,
         			-0.025/0.35/0.05/0.625,
         			1.075 /0.00/0.95/0.29,
         			1.1   /0.29/0.975/0.55
         		}
         		{
         			\drawTallItem{\x*\w}{\y*\hprime}{\xx*\w}{\yy*\hprime};
         		}
         		
         		\foreach \x/\y/\xx/\yy in {
         			0    /0   /0.025/0.35,
         			0.950/0   /1    /0.29,
         			0.975/0.29/1    /0.55
         		}{
         			\draw[pattern = north west lines] (\x*\w,\y*\hprime) rectangle (\xx*\w,\yy*\hprime);
         		}
         		
         		\foreach \x/\y/\xx/\yy in {
         			0.175 /0.0 /0.325/0.31,
         			0.35  /0.0 /0.45 /0.29,
         			0.45  /0.0 /0.5  /0.53,
         			0.5   /0.0 /0.575/0.29,
         			0.6   /0.0 /0.7  /0.35,
         			0.7   /0.0 /0.75 /0.52,
         			0.75  /0.0 /0.85 /0.3,
         			0.875 /0.0 /0.925/0.29,
         			0.075 /0.0 /0.15 /0.28
         		}
         		{
         			\drawTallItem{\x*\w}{\y*\hprime}{\xx*\w}{\yy*\hprime};
         		}
         		
         		\foreach \x/\y/\xx/\yy in {
         			0.05 /0.67/0.125/1,
         			0.15 /0.70/0.175/1,
         			0.225/0.72/0.3  /1,
         			0.3  /0.71/0.4  /1,
         			0.55 /0.65/0.6  /1
         		}
         		{
         			\drawTallItem{\x*\w}{\y*\hprime-\h/4}{\xx*\w}{\yy*\hprime-\h/4};
         		}
         		
         		\foreach \x/\y/\xx/\yy in {
         			0.0   /0.71 /0.025/1,
         			0.025 /0.69 /0.075/1,
         			0.075 /0.71 /0.125/1,
         			0.125 /0.69 /0.2  /1,
         			0.2   /0.37 /0.225/1,
         			0.225 /0.66 /0.25 /1,
         			0.25  /0.66 /0.275/1,
         			0.275 /0.67 /0.325/1,
         			0.325 /0.66 /0.35 /1,
         			0.35  /0.68 /0.375/1,
         			0.375 /0.67 /0.425/1,
         			0.475 /0.69 /0.575/1,
         			0.425 /0.38 /0.45 /1,
         			0.6   /0.40 /0.625/1,
         			0.825 /0.37 /0.9  /1,
         			0.9   /0.66 /0.975/1
         		}
         		{
         			\drawTallItem{\x*\w}{\y*\hprime+\h/4}{\xx*\w}{\yy*\hprime+\h/4};
         		}
         		
         		\foreach \x/\y/\xx/\yy in {
         			0.45 /0.53/0.475/1,
         			0.575/0.69/0.6  /1,
         			0.625/0.35/0.7  /1,
         			0.7  /0.52/0.75 /1,
         			0.75 /0.30/0.825/1,
         			0.975/0.62/1    /1
         		}
         		{
         			\drawVerticalItem{\x*\w}{\y*\hprime+\h/4}{\xx*\w}{\yy*\hprime+\h/4};
         		}
         		
         		\foreach \x/\y/\xx/\yy in {
         			0.0  /1/0.025/1.085,
         			0.025/1/0.05 /1.065,
         			0.075/1/0.125/1.06 ,
         			0.075/1.06/0.125/1.10,
         			0.225/1/0.25 /1.04 ,
         			0.225/1.04/0.25 /1.07,
         			0.25 /1/0.275/1.04 ,
         			0.25 /1.04/0.275/1.07,
         			0.275/1/0.3  /1.05 ,
         			0.275/1.05/0.3  /1.08,
         			0.3  /1/0.325/1.02 ,
         			0.3  /1.02/0.325/1.07,
         			0.35 /1/0.375/1.03 ,
         			0.35 /1.03/0.375/1.10,
         			0.375/1/0.4  /1.02 ,
         			0.375/1.02/0.4  /1.09,
         			0.475/1/0.5  /1.16 ,
         			0.55 /1/0.575/1.05
         		}
         		{
         			\drawVerticalItem{\x*\w}{\y*\hprime-\h/4}{\xx*\w}{\yy*\hprime-\h/4};
         		}
         		
         		\foreach \x/\y/\xx/\yy in {
         			0.125/0.59/0.15 /1,
         			0.175/0.62/0.2  /1,
         			0.4  /0.62/0.425/1,
         			0.5  /0.60/0.55 /1,
         			0.9  /0.63/0.925/1,
         			0.95 /0.68/0.975/1
         		}
         		{
         			\drawVerticalItem{\x*\w}{\y*\hprime-\h/4}{\xx*\w}{\yy*\hprime-\h/4};
         		}
         		
         		\foreach \x/\y/\xx/\yy in {
         			0.2  /1/0.225/1.06,
         			0.425/1/0.45 /1.09,
         			0.6  /1/0.625/1.05 ,
         			0.825/1/0.85 /1.07,
         			0.875/1/0.9  /1.08
         		}
         		{
         			\drawVerticalItem{\x*\w}{\y*\hprime-\hprime/2}{\xx*\w}{\yy*\hprime-\hprime/2};
         		}
         		
         		\foreach \x/\y/\xx/\yy in {
         			0.025/0.0 /0.05 /0.35,
         			0.05 /0.0 /0.075/0.32,
         			0.05 /0.32/0.075/0.36,
         			0.15 /0.0 /0.175/0.33,
         			0.15 /0.33/0.175/0.39,
         			0.325/0.0 /0.35 /0.36,
         			0.325/0.36/0.35 /0.37,
         			0.575/0.0 /0.6  /0.34,
         			0.85 /0.0 /0.875/0.37,
         			0.925/0.0 /0.95 /0.66
         		}
         		{
         			\drawVerticalItem{\x*\w}{\y*\hprime}{\xx*\w}{\yy*\hprime};
         		}
         		
         		\foreach \x/\y/\xx/\yy in {
         			0.2  /0.225,
         			0.425/0.45 ,
         			0.6  /0.625,
         			0.825/0.9
         		}{
         			\draw[pattern = north east lines,pattern color = gray] (\x*\w,0.5*\hprime) rectangle (\xx*\w,1*\hprime+\h/4);
         		}
         		
         		\foreach \y/\z in {
         			0.5*\hprime 		/ $\nicefrac{1}{2}\height(B)$,
         			0.25*\h     		/ $\nicefrac{1}{4}H$,
         			\hprime - 0.25*\h	/ $\height(B) - \nicefrac{1}{4}H$,
         			\hprime				/ $\height(B)$,
         			\hprime + 0.25*\h	/ $\height(B) + \nicefrac{1}{4}H$
         		}{
         			\draw[dotted] (-0.1*\w,\y)node[left]{\z} -- (1.1*\w,\y);
         		}
         		
         		
         		\end{tikzpicture}
         		\caption{Fusing the pseudo items (light gray hatched area)}
         		\label{fig:sub:FusingPseudoItems}
         	\end{subfigure}
         	\hfill
         	\begin{subfigure}[t]{0.3\textwidth}
         		\centering
         		\begin{tikzpicture}
         		
         		\pgfmathsetmacro{\w}{2.7}
         		\pgfmathsetmacro{\h}{4.8}
         		\pgfmathsetmacro{\hprime}{4.5}
         		\draw (0*\w,0) rectangle (1*\w,\hprime +\h/4);

         		\foreach \x/\y/\xx/\yy in {
         			-0.03 /0.00/0.025/0.31,
         			-0.025/0.35/0.05/0.625,
         			1.075 /0.00/0.95/0.29,
         			1.1   /0.29/0.975/0.55
         		}
         		{
         			\drawTallItem{\x*\w}{\y*\hprime}{\xx*\w}{\yy*\hprime};
         		}
         		
         		\foreach \x/\y/\xx/\yy in {
         			0    /0   /0.025/0.35,
         			0.950/0   /1    /0.29,
         			0.975/0.29/1    /0.55
         		}{
         			\draw[pattern = north west lines] (\x*\w,\y*\hprime) rectangle (\xx*\w,\yy*\hprime);
         		}
         		
         		\foreach \x/\y/\xx/\yy in {
         			0.175 /0.0 /0.325/0.31,
         			0.35  /0.0 /0.45 /0.29,
         			0.45  /0.0 /0.5  /0.53,
         			0.5   /0.0 /0.575/0.29,
         			0.6   /0.0 /0.7  /0.35,
         			0.7   /0.0 /0.75 /0.52,
         			0.75  /0.0 /0.85 /0.3,
         			0.875 /0.0 /0.925/0.29,
         			0.075 /0.0 /0.15 /0.28
         		}
         		{
         			\drawTallItem{\x*\w}{\y*\hprime}{\xx*\w}{\yy*\hprime};
         		}
         		
         		\foreach \x/\y/\xx/\yy in {
         			0.05 /0.67/0.125/1,
         			0.15 /0.70/0.175/1,
         			0.225/0.72/0.3  /1,
         			0.3  /0.71/0.4  /1,
         			0.55 /0.65/0.6  /1
         		}
         		{
         			\drawTallItem{\x*\w}{\y*\hprime-\h/4}{\xx*\w}{\yy*\hprime-\h/4};
         		}
         		
         		\foreach \x/\y/\xx/\yy in {
         			0.0   /0.71 /0.025/1,
         			0.025 /0.69 /0.075/1,
         			0.075 /0.71 /0.125/1,
         			0.125 /0.69 /0.2  /1,
         			0.2   /0.37 /0.225/1,
         			0.225 /0.66 /0.25 /1,
         			0.25  /0.66 /0.275/1,
         			0.275 /0.67 /0.325/1,
         			0.325 /0.66 /0.35 /1,
         			0.35  /0.68 /0.375/1,
         			0.375 /0.67 /0.425/1,
         			0.475 /0.69 /0.575/1,
         			0.425 /0.38 /0.45 /1,
         			0.6   /0.40 /0.625/1,
         			0.825 /0.37 /0.9  /1,
         			0.9   /0.66 /0.975/1
         		}
         		{
         			\drawTallItem{\x*\w}{\y*\hprime+\h/4}{\xx*\w}{\yy*\hprime+\h/4};
         		}
         		
         		\foreach \x/\y/\xx/\yy in {
         			0.45 /0.53/0.475/1,
         			0.575/0.69/0.6  /1,
         			0.625/0.35/0.7  /1,
         			0.7  /0.52/0.75 /1,
         			0.75 /0.30/0.825/1,
         			0.975/0.62/1    /1
         		}
         		{
         			\drawVerticalItem{\x*\w}{\y*\hprime+\h/4}{\xx*\w}{\yy*\hprime+\h/4};
         		}
         		
         		\foreach \x/\y/\xx/\yy in {
         			0.0  /1/0.025/1.085,
         			0.025/1/0.05 /1.065,
         			0.075/1/0.125/1.10,
         			0.225/1/0.25 /1.07,
         			0.25 /1/0.275/1.07,
         			0.275/1/0.3  /1.08,
         			0.3  /1/0.325/1.07,
         			0.35 /1/0.375/1.10,
         			0.375/1/0.4  /1.09,
         			0.475/1/0.5  /1.16,
         			0.55 /1/0.575/1.05
         		}
         		{
         			\drawVerticalItem{\x*\w}{\y*\hprime-\h/4}{\xx*\w}{\yy*\hprime-\h/4};
         		}
         		
         		\foreach \x/\y/\xx/\yy in {
         			0.125/0.59/0.15 /1,
         			0.175/0.62/0.2  /1,
         			0.4  /0.62/0.425/1,
         			0.5  /0.60/0.55 /1,
         			0.9  /0.63/0.925/1,
         			0.95 /0.68/0.975/1
         		}
         		{
         			\drawVerticalItem{\x*\w}{\y*\hprime-\h/4}{\xx*\w}{\yy*\hprime-\h/4};
         		}
         		
         		\foreach \x/\y/\xx/\yy in {
         			0.2  /1/0.225/1.06,
         			0.425/1/0.45 /1.09,
         			0.6  /1/0.625/1.05 ,
         			0.825/1/0.85 /1.07,
         			0.875/1/0.9  /1.08
         		}
         		{
         			\drawVerticalItem{\x*\w}{\y*\hprime-\hprime/2}{\xx*\w}{\yy*\hprime-\hprime/2};
         		}
         		
         		\foreach \x/\y/\xx/\yy in {
         			0.025/0.0 /0.05 /0.35,
         			0.05 /0.0 /0.075/0.36,
         			0.15 /0.0 /0.175/0.39,
         			0.325/0.0 /0.35 /0.37,
         			0.575/0.0 /0.6  /0.34,
         			0.85 /0.0 /0.875/0.37,
         			0.925/0.0 /0.95 /0.66
         		}
         		{
         			\drawVerticalItem{\x*\w}{\y*\hprime}{\xx*\w}{\yy*\hprime};
         		}
         		
         		\foreach \x/\y/\xx/\yy in {
         			0.2  /0.225,
         			0.425/0.45 ,
         			0.6  /0.625,
         			0.825/0.9
         		}{
         			\draw[pattern = north east lines, pattern color = gray] (\x*\w,0.5*\hprime) rectangle (\xx*\w,1*\hprime+\h/4);
         		}
         		
         		\foreach \y/\z in {
         			0.5*\hprime 		/ $\nicefrac{1}{2}\height(B)$,
         			0.25*\h     		/ $\nicefrac{1}{4}H$,
         			\hprime - 0.25*\h	/ $\height(B) - \nicefrac{1}{4}H$,
         			\hprime				/ $\height(B)$,
         			\hprime + 0.25*\h	/ $\height(B) + \nicefrac{1}{4}H$
         		}{
         			\draw[dotted] (-0.1*\w,\y) -- (1.1*\w,\y) ;
         		}
         		
         		\draw[very thick, <-] (0.2075*\w,\hprime + \h/4) -- (0.2075*\w,\hprime + 1.1* \h/4) node[above]{$i_\ell$};
         		\draw[very thick, <-] (0.865*\w,\hprime + \h/4) -- (0.865*\w,\hprime + 1.1* \h/4) node[above]{$i_r$};
         		
         		\draw[red] (0.2*\w,0) --(0.2 *\w,\hprime +\h/4);
         		\draw[red] (0.9*\w,0) --(0.9 *\w,\hprime +\h/4);
         		
         		\draw[red, very thick, pattern = dots, pattern color = gray] (0*\w,\hprime-\h/4) rectangle (0.2 *\w,\hprime +\h/4);
         		\draw[red] (0.00 *\w,\hprime) -- (-0.1 *\w,\hprime +\h/8) node[left,black]{$B_{\ell,1}$};
         		
         		\draw[red, very thick, pattern = dots, pattern color = gray] (0.9*\w,\hprime-\h/4) rectangle (1 *\w,\hprime +\h/4);
         		\draw[red] (1.00 *\w,\hprime) --(1.15 *\w,\hprime +\h/8) node[above,black]{$B_{r,1}$};
         		

         		\end{tikzpicture}
         		\caption{The area $B_{\ell,1}$ and $B_{r,1}$}
         		\label{fig:sub:AreaB1}
         	\end{subfigure}
         	\hfill
         	\begin{subfigure}[t]{0.3\textwidth}
         		\centering
         		\begin{tikzpicture}
         		\pgfmathsetmacro{\w}{2.7}
         		\pgfmathsetmacro{\h}{4.8}
         		\pgfmathsetmacro{\hprime}{4.5}
         		\draw (0*\w,0) rectangle (1*\w,\hprime +\h/4);
         		
         		\draw (0*\w,0) rectangle (1*\w,\hprime +\h/4);

         		\foreach \x/\y/\xx/\yy in {
         			-0.03 /0.00/0.025/0.31,
         			-0.025/0.35/0.05/0.625,
         			1.075 /0.00/0.95/0.29,
         			1.1   /0.29/0.975/0.55
         		}
         		{
         			\drawTallItem{\x*\w}{\y*\hprime}{\xx*\w}{\yy*\hprime};
         		}
         		
         		\foreach \x/\y/\xx/\yy in {
         			0    /0   /0.025/0.35,
         			0.950/0   /1    /0.29,
         			0.975/0.29/1    /0.55
         		}{
         			\draw[pattern = north west lines] (\x*\w,\y*\hprime) rectangle (\xx*\w,\yy*\hprime);
         		}
         		
         		\foreach \x/\y/\xx/\yy in {
         			0.175 /0.0 /0.325/0.31,
         			0.35  /0.0 /0.45 /0.29,
         			0.45  /0.0 /0.5  /0.53,
         			0.5   /0.0 /0.575/0.29,
         			0.6   /0.0 /0.7  /0.35,
         			0.7   /0.0 /0.75 /0.52,
         			0.75  /0.0 /0.85 /0.3,
         			0.875 /0.0 /0.925/0.29,
         			0.075 /0.0 /0.15 /0.28
         		}
         		{
         			\drawTallItem{\x*\w}{\y*\hprime}{\xx*\w}{\yy*\hprime};
         		}
         		
         		\foreach \x/\y/\xx/\yy in {
         			0.05 /0.67/0.125/1,
         			0.15 /0.70/0.175/1,
         			0.225/0.72/0.3  /1,
         			0.3  /0.71/0.4  /1,
         			0.55 /0.65/0.6  /1
         		}
         		{
         			\drawTallItem{\x*\w}{\y*\hprime-\h/4}{\xx*\w}{\yy*\hprime-\h/4};
         		}
         		
         		\foreach \x/\y/\xx/\yy in {
         			0.0   /0.71 /0.025/1,
         			0.025 /0.71 /0.075/1,
         			0.075 /0.69 /0.125/1,
         			0.125 /0.69 /0.2  /1,
         			0.2   /0.37 /0.225/1,
         			0.225 /0.66 /0.25 /1,
         			0.25  /0.66 /0.275/1,
         			0.275 /0.67 /0.325/1,
         			0.325 /0.66 /0.35 /1,
         			0.35  /0.68 /0.375/1,
         			0.375 /0.67 /0.425/1,
         			0.475 /0.69 /0.575/1,
         			0.425 /0.38 /0.45 /1,
         			0.6   /0.40 /0.625/1,
         			0.825 /0.37 /0.9  /1,
         			0.9   /0.66 /0.975/1
         		}
         		{
         			\drawTallItem{\x*\w}{\y*\hprime+\h/4}{\xx*\w}{\yy*\hprime+\h/4};
         		}
         		
         		\foreach \x/\y/\xx/\yy in {
         			0.45 /0.53/0.475/1,
         			0.575/0.69/0.6  /1,
         			0.625/0.35/0.7  /1,
         			0.7  /0.52/0.75 /1,
         			0.75 /0.30/0.825/1,
         			0.975/0.62/1    /1
         		}
         		{
         			\drawVerticalItem{\x*\w}{\y*\hprime+\h/4}{\xx*\w}{\yy*\hprime+\h/4};
         		}
         		
         		\foreach \x/\y/\xx/\yy in {
         			0.000/1/0.05/1.10,
         			0.05 /1/0.075/1.085,
         			0.075/1/0.1 /1.065,
         			0.225/1/0.25 /1.07,
         			0.25 /1/0.275/1.07,
         			0.275/1/0.3  /1.08,
         			0.3  /1/0.325/1.07,
         			0.35 /1/0.375/1.10,
         			0.375/1/0.4  /1.09,
         			0.475/1/0.5  /1.16,
         			0.55 /1/0.575/1.05
         		}
         		{
         			\drawVerticalItem{\x*\w}{\y*\hprime-\h/4}{\xx*\w}{\yy*\hprime-\h/4};
         		}
         		
         		\foreach \x/\y/\xx/\yy in {
         			0.125/0.59/0.15 /1,
         			0.175/0.62/0.2  /1,
         			0.4  /0.62/0.425/1,
         			0.5  /0.60/0.55 /1,
         			0.9  /0.63/0.925/1,
         			0.95 /0.68/0.975/1
         		}
         		{
         			\drawVerticalItem{\x*\w}{\y*\hprime-\h/4}{\xx*\w}{\yy*\hprime-\h/4};
         		}
         		
         		\foreach \x/\y/\xx/\yy in {
         			0.2  /1/0.225/1.06,
         			0.425/1/0.45 /1.09,
         			0.6  /1/0.625/1.05 ,
         			0.825/1/0.85 /1.07,
         			0.875/1/0.9  /1.08
         		}
         		{
         			\drawVerticalItem{\x*\w}{\y*\hprime-\hprime/2}{\xx*\w}{\yy*\hprime-\hprime/2};
         		}
         		
         		\foreach \x/\y/\xx/\yy in {
         			0.025/0.0 /0.05 /0.35,
         			0.05 /0.0 /0.075/0.36,
         			0.15 /0.0 /0.175/0.39,
         			0.325/0.0 /0.35 /0.37,
         			0.575/0.0 /0.6  /0.34,
         			0.85 /0.0 /0.875/0.37,
         			0.925/0.0 /0.95 /0.66
         		}
         		{
         			\drawVerticalItem{\x*\w}{\y*\hprime}{\xx*\w}{\yy*\hprime};
         		}
         		
         		\foreach \x/\y/\xx/\yy in {
         			0.2  /0.225,
         			0.425/0.45 ,
         			0.6  /0.625,
         			0.825/0.9
         		}{
         			\draw[pattern = north east lines, pattern color = gray] (\x*\w,0.5*\hprime) rectangle (\xx*\w,1*\hprime+\h/4);
         		}
         		
         		\foreach \y/\z in {
         			0.5*\hprime 		/ $\nicefrac{1}{2}\height(B)$,
         			0.25*\h     		/ $\nicefrac{1}{4}H$,
         			\hprime - 0.25*\h	/ $\height(B) - \nicefrac{1}{4}H$,
         			\hprime				/ $\height(B)$,
         			\hprime + 0.25*\h	/ $\height(B) + \nicefrac{1}{4}H$
         		}{
         			\draw[dotted] (-0.1*\w,\y) -- (1.1*\w,\y) ;
         		}
         		
         		\draw[very thick] (0*\w,\hprime-\h/4) rectangle (0.2 *\w,\hprime +\h/4);
         		\draw[fill = white, opacity = 0.5] (0*\w,\hprime-\h/4) rectangle (0.2 *\w,\hprime +\h/4);
         		
         		\draw[very thick] (0.9*\w,\hprime-\h/4) rectangle (1 *\w,\hprime +\h/4);
         		\draw[fill = white, opacity = 0.5] (0.9*\w,\hprime-\h/4) rectangle (1 *\w,\hprime +\h/4);
         		
         		\draw[red, very thick,pattern = dots,pattern color = gray] (0*\w,\hprime-\h/4) rectangle node[midway,opacity =1](A){} (0.2 *\w,0);
         		\draw[thick] (A) --(-0.2*\w,3*\hprime/6) node[above]{$B_{\ell,2}$};
         		\draw[red, very thick,pattern = dots,pattern color = gray] (0.9*\w,\hprime-\h/4) rectangle (1 *\w,0);
         		\draw[thick] (1*\w,\h/4) --(1.2 *\w,3*\hprime/8) node[above]{$B_{r,2}$};
         		
         		\end{tikzpicture}
         		\caption{The area $B_{\ell,2}$ and $B_{r,2}$}
         		\label{fig:sub:AreaB2}
         	\end{subfigure}			
         	\caption{We fuse pseudo items as discussed in Step 1.5. Then we begin defining areas as in \cref{lem:ReorderTallVertBox34}. In the middle, we define the areas $B_{\ell,1}$ and $B_{r,1}$, to the right the areas $B_{\ell,2},B_{r,2}$. Each area is packed in a structured manner in the image following its introduction, i.e. $B_{\ell,1}$ is structured in the right image.}
         	\label{fig:GenralReordering1}
         \end{figure*}
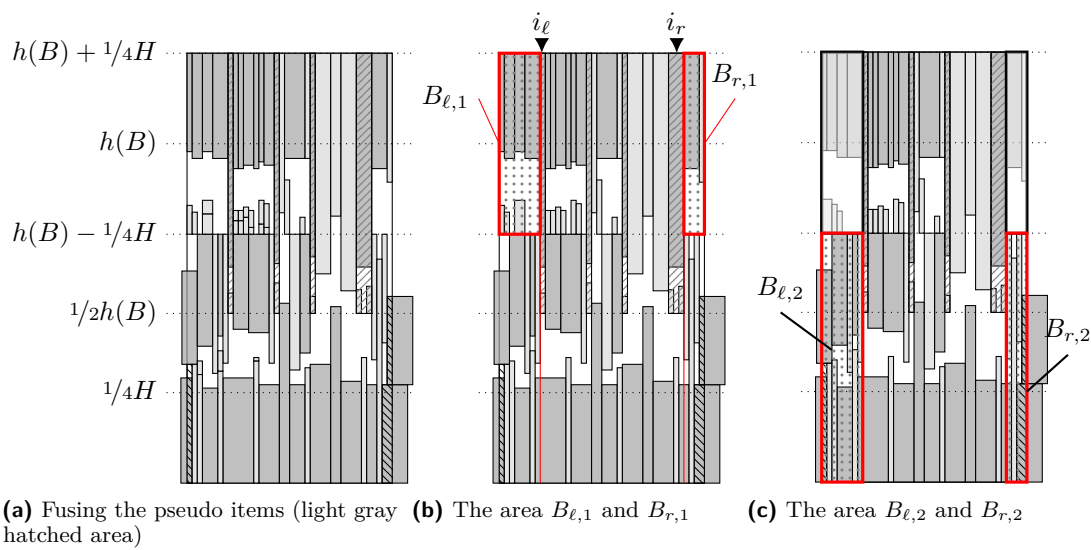
          \textit{Area $B_{\ell,3}$:} Next, we look at the area between the left border of $i_\ell$ and the right border of $i_r$. For an example of this area, see \cref{fig:sub:AreaB3}. Let $r(i)$ be the right border of an item $i$. If $r(\ell_b)$ is to the right of $r(i_\ell)$, we draw a vertical line at $r(\ell_b)$, called $L_1$. If $L_1$ intersects a tall item with upper border at $\height(B)-\nicefrac{1}{4}H$, we call this item $\ell_m$. Left of $L_1$ and right of $i_\ell$, we shift up each item touching $\height(B)-\nicefrac{1}{4}H$ with its top, including $\ell_m$, such that its lower border touches $\nicefrac{1}{2}\height(B)$ and shift down each pseudo item touching $\height(B)-\nicefrac{1}{4}H$ with its lower border such that it touches $\nicefrac{1}{2}\height(B)$ with its upper border. All pseudo items right of $L_1$ above $\ell_m$ are shifted, such that they touch the top of $\ell_m$ with their bottom. No pseudo item is intersected by the line $L_1$.
         \begin{claim}
             After this shift no item overlaps another item.
         \end{claim}
         \begin{proof}
             Consider an item $i$ that was shifted up such that its lower border is at $\nicefrac{1}{2}\height(B)$. Note that the distance between the upper border of $\ell_b$ and $\nicefrac{1}{2}\height(B)$ is less than $\nicefrac{1}{4}H$ because the upper border of $\ell_b$ is above $\nicefrac{1}{4}H$ by definition. Hence, there has to be some free space left between the upper border of $i$ and the lower border of each item above since we added $\nicefrac{1}{4}H$ to the packing height.\newline\indent
             Now consider an item $i'$ that was shifted down such that its top border is at $\nicefrac{1}{2}\height(B)$. Above this item, there has to be a tall item $i''$ that has its lower border at $\nicefrac{1}{2}\height(B)$, which has a height greater than $\nicefrac{1}{4}H$. The item $i'''$ above $i''$, an item with its top border at $\height(B)+\nicefrac{1}{4}H$, has a height larger than $\nicefrac{1}{4}H$ as well, since all items with their top borders at $\height(B)+\nicefrac{1}{4}H$ have this property. Therefore, the vertical distance between $i''$ and $i'''$ is smaller than $\nicefrac{1}{4}H$. Since we have added $\nicefrac{1}{4}H$ to the packing height, the vertical distance between the bottom of $i'$ and the top of $\ell_b$ has to be larger than zero. This, in turn, means that no item parts can overlap each other vertically.
         \end{proof}
                  	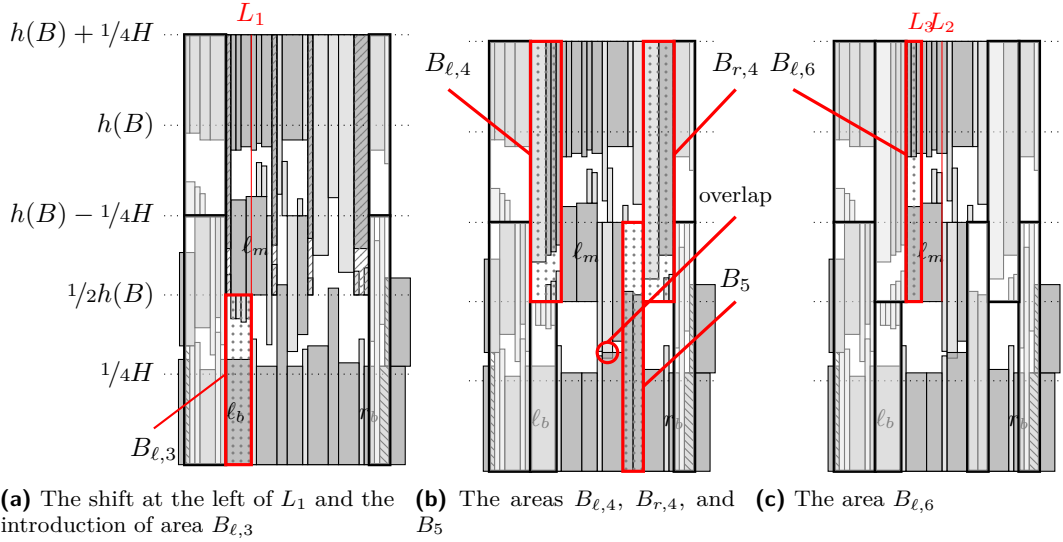
\begin{figure*}[t]
         	\centering
         	\begin{subfigure}[t]{0.37\textwidth}
         		\centering
         		\begin{tikzpicture}
         		\pgfmathsetmacro{\w}{2.7}
         		\pgfmathsetmacro{\h}{4.8}
         		\pgfmathsetmacro{\hprime}{4.5}
         		
         		\draw (0*\w,0) rectangle (1*\w,\hprime +\h/4);
         		
         		\foreach \x/\y/\xx/\yy in {
         			-0.03 /0.00/0.025/0.31,
         			-0.025/0.35/0.05/0.625,
         			1.075 /0.00/0.95/0.29,
         			1.1   /0.29/0.975/0.55
         		}
         		{
         			\drawTallItem{\x*\w}{\y*\hprime}{\xx*\w}{\yy*\hprime};
         		}
         		
         		\foreach \x/\y/\xx/\yy in {
         			0    /0   /0.025/0.35,
         			0.950/0   /1    /0.29,
         			0.975/0.29/1    /0.55
         		}{
         			\draw[pattern = north west lines] (\x*\w,\y*\hprime) rectangle (\xx*\w,\yy*\hprime);
         		}
         		
         		\foreach \x/\y/\xx/\yy in {
         			0.05 /0.67/0.125/1,
         			0.15 /0.70/0.175/1,
         			0.55 /0.65/0.6  /1
         		}
         		{
         			\drawTallItem{\x*\w}{\y*\hprime-\h/4}{\xx*\w}{\yy*\hprime-\h/4};
         		}
         		
         		\foreach \x/\y/\xx/\yy/\z in {
         			0.225/1/0.3  /1.28/,
         			0.3  /1/0.4  /1.29/$\ell_m$
         		}
         		{
         			\drawTallItem[\small \z]{\x*\w}{\y*\hprime-0.5*\hprime}{\xx*\w}{\yy*\hprime- 0.5*\hprime};
         		}
         		\foreach \x/\y/\xx/\yy in {
         			0.0   /0.71 /0.025/1,
         			0.025 /0.71 /0.075/1,
         			0.075 /0.69 /0.125/1,
         			0.125 /0.69 /0.2  /1,
         			0.2   /0.37 /0.225/1,
         			0.225 /0.66 /0.25 /1,
         			0.25  /0.66 /0.275/1,
         			0.275 /0.67 /0.325/1,
         			0.325 /0.66 /0.35 /1,
         			0.35  /0.68 /0.375/1,
         			0.375 /0.67 /0.425/1,
         			0.475 /0.69 /0.575/1,
         			0.425 /0.38 /0.45 /1,
         			0.6   /0.40 /0.625/1,
         			0.825 /0.37 /0.9  /1,
         			0.9   /0.66 /0.975/1
         		}
         		{
         			\drawTallItem{\x*\w}{\y*\hprime+\h/4}{\xx*\w}{\yy*\hprime+\h/4};
         		}
         		
         		\foreach \x/\y/\xx/\yy in {
         			0.45 /0.53/0.475/1,
         			0.575/0.69/0.6  /1,
         			0.625/0.35/0.7  /1,
         			0.7  /0.52/0.75 /1,
         			0.75 /0.30/0.825/1,
         			0.975/0.62/1    /1
         		}
         		{
         			\drawVerticalItem{\x*\w}{\y*\hprime+\h/4}{\xx*\w}{\yy*\hprime+\h/4};
         		}
         		
         		\foreach \x/\y/\xx/\yy in {
         			0.000/1/0.05/1.10,
         			0.05 /1/0.075/1.085,
         			0.075/1/0.1 /1.065,
         			0.475/1/0.5  /1.16,
         			0.55 /1/0.575/1.05
         		}
         		{
         			\drawVerticalItem{\x*\w}{\y*\hprime-\h/4}{\xx*\w}{\yy*\hprime-\h/4};
         		}
         		
         		\foreach \x/\y/\xx/\yy in {
         			0.225/0.93/0.25 /1,
         			0.25 /0.93/0.275/1,
         			0.275/0.92/0.3  /1,
         			0.3  /0.93/0.325/1,
         			0.35 /1.29/0.375/1.39,
         			0.375/1.29/0.4  /1.38
         		}
         		{
         			\drawVerticalItem{\x*\w}{\y*\hprime-0.5*\hprime}{\xx*\w}{\yy*\hprime-0.5*\hprime};
         		}

         		\foreach \x/\y/\xx/\yy in {
         			0.125/0.59/0.15 /1,
         			0.175/0.62/0.2  /1,
         			0.4  /0.62/0.425/1,
         			0.5  /0.60/0.55 /1,
         			0.9  /0.63/0.925/1,
         			0.95 /0.68/0.975/1
         		}
         		{
         			\drawVerticalItem{\x*\w}{\y*\hprime-\h/4}{\xx*\w}{\yy*\hprime-\h/4};
         		}
         		
         		\foreach \x/\y/\xx/\yy in {
         			0.2  /1/0.225/1.06,
         			0.425/1/0.45 /1.09,
         			0.6  /1/0.625/1.05,
         			0.825/1/0.85 /1.07,
         			0.875/1/0.9  /1.08
         		}
         		{
         			\drawVerticalItem{\x*\w}{\y*\hprime-\hprime/2}{\xx*\w}{\yy*\hprime-\hprime/2};
         		}
         		
         		\foreach \x/\y/\xx/\yy in {
         			0.025/0.0 /0.05 /0.35,
         			0.05 /0.0 /0.075/0.36,
         			0.15 /0.0 /0.175/0.39,
         			0.325/0.0 /0.35 /0.37,
         			0.575/0.0 /0.6  /0.34,
         			0.85 /0.0 /0.875/0.37,
         			0.925/0.0 /0.95 /0.66
         		}
         		{
         			\drawVerticalItem{\x*\w}{\y*\hprime}{\xx*\w}{\yy*\hprime};
         		}
         		
         		\foreach \x/\y/\xx/\yy in {
         			0.2  /0.225,
         			0.425/0.45 ,
         			0.6  /0.625,
         			0.825/0.9
         		}{
         			\draw[pattern = north east lines, pattern color = gray] (\x*\w,0.5*\hprime) rectangle (\xx*\w,1*\hprime+\h/4);
         		}
         		
         		\foreach \x/\y/\xx/\yy/\z in {
         			0.075 /0.0 /0.15 /0.28/,
         			0.175 /0.0 /0.325/0.31/$\ell_b$,
         			0.35  /0.0 /0.45 /0.29/,
         			0.45  /0.0 /0.5  /0.53/,
         			0.5   /0.0 /0.575/0.29/,
         			0.6   /0.0 /0.7  /0.35/,
         			0.7   /0.0 /0.75 /0.52/,
         			0.75  /0.0 /0.85 /0.30/,
         			0.875 /0.0 /0.925/0.29/$r_b$
         		}
         		{
         			\drawTallItem[\small \z]{\x*\w}{\y*\hprime}{\xx*\w}{\yy*\hprime};
         		}

         		\foreach \y/\z in {
         			0.5*\hprime 		/ $\nicefrac{1}{2}\height(B)$,
         			0.25*\h     		/ $\nicefrac{1}{4}H$,
         			\hprime - 0.25*\h	/ $\height(B) - \nicefrac{1}{4}H$,
         			\hprime				/ $\height(B)$,
         			\hprime + 0.25*\h	/ $\height(B) + \nicefrac{1}{4}H$
         		}{
         			\draw[dotted] (-0.1*\w,\y) node[left]{\z}-- (1.1*\w,\y) ;
         		}
         		
         		\foreach \x/\y/\xx/\yy in {
         			0.0*\w/\hprime-0.25*\h/0.2*\w/\hprime + 0.25*\h,
         			0.9*\w/\hprime-0.25*\h/1.0*\w/\hprime +0.25*\h,
         			0.0*\w/\hprime-0.25*\h/0.2*\w/0,
         			0.9*\w/\hprime-0.25*\h/1.0*\w/0
         		}{
         			\draw[very thick] (\x,\y) rectangle (\xx,\yy);
         			\draw[fill = white, opacity = 0.5] (\x,\y) rectangle (\xx,\yy);
         			
         		}
         		
         		\draw[red] (0.325*\w,0.00*\hprime) rectangle (0.325*\w,\hprime +\h/4) node[above]{$L_1$};
         		\draw[very thick,red,pattern = dots,pattern color = gray] (0.2*\w,\hprime/2) rectangle (0.325 *\w,0);
         		\draw[thick,red] (0.2*\w,\h/4) --(-0.15 *\w,0.1*\h) node[below,black]{$B_{\ell,3}$};
         		

         		\end{tikzpicture}
         		\caption{The shift at the left of $L_1$ and the introduction of area $B_{\ell,3}$}
         		\label{fig:sub:AreaB3}
         	\end{subfigure}
         	\hfill
         	\begin{subfigure}[t]{0.3\textwidth}
         		\centering
         		\begin{tikzpicture}
         		\pgfmathsetmacro{\w}{2.7}
         		\pgfmathsetmacro{\h}{4.8}
         		\pgfmathsetmacro{\hprime}{4.5}
         		
         		\draw (0*\w,0) rectangle (1*\w,\hprime +\h/4);
         		
         		\foreach \x/\y/\xx/\yy in {
         			-0.03 /0.00/0.025/0.31,
         			-0.025/0.35/0.05/0.625,
         			1.075 /0.00/0.95/0.29,
         			1.1   /0.29/0.975/0.55
         		}
         		{
         			\drawTallItem{\x*\w}{\y*\hprime}{\xx*\w}{\yy*\hprime};
         		}
         		
         		\foreach \x/\y/\xx/\yy in {
         			0    /0   /0.025/0.35,
         			0.950/0   /1    /0.29,
         			0.975/0.29/1    /0.55
         		}{
         			\draw[pattern = north west lines] (\x*\w,\y*\hprime) rectangle (\xx*\w,\yy*\hprime);
         		}
         		
         		\foreach \x/\y/\xx/\yy in {
         			0.05 /0.67/0.125/1,
         			0.15 /0.70/0.175/1,
         			0.6 /0.65/0.65  /1
         		}
         		{
         			\drawTallItem{\x*\w}{\y*\hprime-\h/4}{\xx*\w}{\yy*\hprime-\h/4};
         		}
         		
         		\foreach \x/\y/\xx/\yy/\z in {
         			0.35 /1/0.425  /1.28/,
         			0.425/1/0.525  /1.29/$\ell_m$
         		}
         		{
         			\drawTallItem[\small \z]{\x*\w}{\y*\hprime-0.5*\hprime}{\xx*\w}{\yy*\hprime-0.5* \hprime};
         		}
         		\foreach \x/\y/\xx/\yy in {
         			0.0   /0.71 /0.025/1,
         			0.025 /0.71 /0.075/1,
         			0.075 /0.69 /0.125/1,
         			0.125 /0.69 /0.2  /1,
         			0.275 /0.37 /0.3  /1,
         			0.3   /0.38 /0.325/1,
         			0.325 /0.40 /0.35 /1,
         			0.35  /0.66 /0.375/1,
         			0.375 /0.66 /0.4  /1,
         			0.4   /0.67 /0.45 /1,
         			0.45  /0.66 /0.475/1,
         			0.475 /0.68 /0.5  /1,
         			0.5   /0.67 /0.55 /1,
         			0.575 /0.69 /0.675/1,
         			0.825 /0.37 /0.9  /1,
         			0.9   /0.66 /0.975/1
         		}
         		{
         			\drawTallItem{\x*\w}{\y*\hprime+\h/4}{\xx*\w}{\yy*\hprime+\h/4};
         		}
         		
         		\foreach \x/\y/\xx/\yy in {
         			0.2  /0.35/0.275/1,
         			0.55 /0.53/0.575/1,
         			0.675/0.69/0.7  /1,
         			0.7  /0.52/0.75 /1,
         			0.75 /0.30/0.825/1,
         			0.975/0.62/1    /1
         		}
         		{
         			\drawVerticalItem{\x*\w}{\y*\hprime+\h/4}{\xx*\w}{\yy*\hprime+\h/4};
         		}
         		
         		\foreach \x/\y/\xx/\yy in {
         			0.000/1/0.05/1.10,
         			0.05 /1/0.075/1.085,
         			0.075/1/0.1 /1.065,
         			0.575/1/0.6  /1.16,
         			0.65 /1/0.675/1.05
         		}
         		{
         			\drawVerticalItem{\x*\w}{\y*\hprime-\h/4}{\xx*\w}{\yy*\hprime-\h/4};
         		}
         		
         		\foreach \x/\y/\xx/\yy in {
         			0.225/0.93/0.25 /1,
         			0.25 /0.93/0.275/1,
         			0.275/0.92/0.3  /1,
         			0.3  /0.93/0.325/1,
         			0.475 /1.29/0.5/1.39,
         			0.5/1.29/0.525  /1.38
         		}
         		{
         			\drawVerticalItem{\x*\w}{\y*\hprime-0.5*\hprime}{\xx*\w}{\yy*\hprime-0.5*\hprime};
         		}

         		\foreach \x/\y/\xx/\yy in {
         			0.125/0.59/0.15 /1,
         			0.175/0.62/0.2  /1,
         			0.525/0.62/0.55 /1,
         			0.55 /0.60/0.6 /1,
         			0.9  /0.63/0.925/1,
         			0.95 /0.68/0.975/1
         		}
         		{
         			\drawVerticalItem{\x*\w}{\y*\hprime-\h/4}{\xx*\w}{\yy*\hprime-\h/4};
         		}
         		
         		\foreach \x/\y/\xx/\yy in {
         			0.275/1/0.300/1.05,
         			0.300/1/0.325/1.06,
         			0.325/1/0.350/1.09,
         			0.85 /1/0.875/1.07,
         			0.875/1/0.9  /1.08
         		}
         		{
         			\drawVerticalItem{\x*\w}{\y*\hprime-\hprime/2}{\xx*\w}{\yy*\hprime-\hprime/2};
         		}
         		
         		\foreach \x/\y/\xx/\yy in {
         			0.025/0.0 /0.05 /0.35,
         			0.05 /0.0 /0.075/0.36,
         			0.15 /0.0 /0.175/0.39,
         			0.325/0.0 /0.35 /0.37,
         			0.525/0.0 /0.55 /0.34,
         			0.85 /0.0 /0.875/0.37,
         			0.925/0.0 /0.95 /0.66
         		}
         		{
         			\drawVerticalItem{\x*\w}{\y*\hprime}{\xx*\w}{\yy*\hprime};
         		}

         		\foreach \x/\y/\xx/\yy/\z in {
         			0.075 /0.0 /0.15 /0.28/,
         			0.175 /0.0 /0.325/0.31/$\ell_b$,
         			0.35  /0.0 /0.45 /0.29/,
         			0.45  /0.0 /0.525/0.29/,
         			0.55  /0.0 /0.65  /0.35/,
         			0.65  /0.0 /0.7  /0.53/,
         			0.7   /0.0 /0.75 /0.52/,
         			0.75  /0.0 /0.85 /0.30/,
         			0.875 /0.0 /0.925/0.29/$r_b$
         		}
         		{
         			\drawTallItem[\small \z]{\x*\w}{\y*\hprime}{\xx*\w}{\yy*\hprime};
         		}

         		\foreach \y/\z in {
         			0.5*\hprime 		/ $\nicefrac{1}{2}\height(B)$,
         			0.25*\h     		/ $\nicefrac{1}{4}H$,
         			\hprime - 0.25*\h	/ $\height(B) - \nicefrac{1}{4}H$,
         			\hprime				/ $\height(B)$,
         			\hprime + 0.25*\h	/ $\height(B) + \nicefrac{1}{4}H$
         		}{
         			\draw[dotted] (-0.1*\w,\y) -- (1.1*\w,\y) ;
         		}
         		
         		\foreach \x/\y/\xx/\yy in {
         			0.0*\w/\hprime-0.25*\h/0.2*\w/\hprime + 0.25*\h, 
         			0.9*\w/\hprime-0.25*\h/1.0*\w/\hprime +0.25*\h, 
         			0.0*\w/\hprime-0.25*\h/0.2*\w/0, 
         			0.9*\w/\hprime-0.25*\h/1.0*\w/0, 
         			0.2*\w/0/0.325*\w/0.5*\hprime 
         		}{
         			\draw[very thick] (\x,\y) rectangle (\xx,\yy);
         			\draw[fill = white, opacity = 0.5] (\x,\y) rectangle (\xx,\yy);
         			
         		}

         		\draw[very thick, red, pattern = dots, pattern color = gray] (0.2*\w,\hprime/2) rectangle (0.35 *\w,\hprime+\h/4);
         		\draw[very thick,red] (0.2*\w,7*\h/8) -- (-0.2*\w,9*\hprime/8) node[above, black]{$B_{\ell,4}$} ;
         		
         		\draw[very thick, red, pattern = dots, pattern color = gray] (0.75*\w,\hprime/2) rectangle (0.9 *\w,\hprime+\h/4);
         		\draw[very thick, red] (0.9*\w,7*\h/8) -- (1.2*\w,9*\hprime/8) node[above,black]{$B_{r,4}$} ;
         		
         		\draw[very thick, red, pattern = dots, pattern color = gray] (0.65*\w,0.0*\hprime) rectangle node[midway](C){} (0.75 *\w,\hprime -\h/4);
         		\draw[very thick,red] (0.75*\w,0.25*\hprime) -- (1.2*\w, 0.5*\hprime) node[above,black]{$B_{5}$};
         		
         		\draw[very thick, red] (0.575*\w,0.35*\hprime) circle (0.05*\w);
         		\draw[very thick,red] (0.575*\w,0.35*\hprime+0.05*\w) -- (1.2*\w, 0.75*\hprime) node[above,black]{\footnotesize overlap};
         		
         		\end{tikzpicture}
         		\caption{The areas $B_{\ell,4}$, $B_{r,4}$, and $B_{5}$}
         		\label{fig:sub:AreaB4andB5}
         	\end{subfigure}
         	\hfill
         	\begin{subfigure}[t]{0.29\textwidth}
         		\centering
         		\begin{tikzpicture}
         		\pgfmathsetmacro{\w}{2.7}
         		\pgfmathsetmacro{\h}{4.8}
         		\pgfmathsetmacro{\hprime}{4.5}
         		\draw (0*\w,0) rectangle (1*\w,\hprime +\h/4);
         		
         		\foreach \x/\y/\xx/\yy in {
         			-0.03 /0.00/0.025/0.31,
         			-0.025/0.35/0.05/0.625,
         			1.075 /0.00/0.95/0.29,
         			1.1   /0.29/0.975/0.55
         		}
         		{
         			\drawTallItem{\x*\w}{\y*\hprime}{\xx*\w}{\yy*\hprime};
         		}
         		
         		\foreach \x/\y/\xx/\yy in {
         			0    /0   /0.025/0.35,
         			0.950/0   /1    /0.29,
         			0.975/0.29/1    /0.55
         		}{
         			\draw[pattern = north west lines] (\x*\w,\y*\hprime) rectangle (\xx*\w,\yy*\hprime);
         		}
         		
         		\foreach \x/\y/\xx/\yy in {
         			0.05 /0.67/0.125/1,
         			0.15 /0.70/0.175/1,
         			0.6 /0.65/0.65  /1
         		}
         		{
         			\drawTallItem{\x*\w}{\y*\hprime-\h/4}{\xx*\w}{\yy*\hprime-\h/4};
         		}
         		
         		\foreach \x/\y/\xx/\yy/\z in {
         			0.35 /1/0.425  /1.28/,
         			0.425/1/0.525  /1.29/$\ell_m$
         		}
         		{
         			\drawTallItem[\small \z]{\x*\w}{\y*\hprime-0.5*\hprime}{\xx*\w}{\yy*\hprime-0.5* \hprime};
         		}
         		\foreach \x/\y/\xx/\yy in {
         			0.0   /0.71 /0.025/1,
         			0.025 /0.71 /0.075/1,
         			0.075 /0.69 /0.125/1,
         			0.125 /0.69 /0.2  /1,
         			0.275 /0.37 /0.3  /1,
         			0.3   /0.38 /0.325/1,
         			0.325 /0.40 /0.35 /1,
         			0.35  /0.66 /0.375/1,
         			0.375 /0.66 /0.4  /1,
         			0.4   /0.67 /0.45 /1,
         			0.45  /0.66 /0.475/1,
         			0.475 /0.68 /0.5  /1,
         			0.5   /0.67 /0.55 /1,
         			0.575 /0.69 /0.675/1,
         			0.825 /0.37 /0.9  /1,
         			0.9   /0.66 /0.975/1
         		}
         		{
         			\drawTallItem{\x*\w}{\y*\hprime+\h/4}{\xx*\w}{\yy*\hprime+\h/4};
         		}
         		
         		\foreach \x/\y/\xx/\yy in {
         			0.2  /0.35/0.275/1,
         			0.55 /0.53/0.575/1,
         			0.675/0.69/0.7  /1,
         			0.7  /0.52/0.75 /1,
         			0.75 /0.30/0.825/1,
         			0.975/0.62/1    /1
         		}
         		{
         			\drawVerticalItem{\x*\w}{\y*\hprime+\h/4}{\xx*\w}{\yy*\hprime+\h/4};
         		}
         		
         		\foreach \x/\y/\xx/\yy in {
         			0.000/1/0.05/1.10,
         			0.05 /1/0.075/1.085,
         			0.075/1/0.1 /1.065,
         			0.575/1/0.6  /1.16,
         			0.65 /1/0.675/1.05
         		}
         		{
         			\drawVerticalItem{\x*\w}{\y*\hprime-\h/4}{\xx*\w}{\yy*\hprime-\h/4};
         		}
         		
         		\foreach \x/\y/\xx/\yy in {
         			0.225/0.93/0.25 /1,
         			0.25 /0.93/0.275/1,
         			0.275/0.92/0.3  /1,
         			0.3  /0.93/0.325/1,
         			0.475 /1.29/0.5/1.39,
         			0.5/1.29/0.525  /1.38
         		}
         		{
         			\drawVerticalItem{\x*\w}{\y*\hprime-0.5*\hprime}{\xx*\w}{\yy*\hprime-0.5*\hprime};
         		}

         		\foreach \x/\y/\xx/\yy in {
         			0.125/0.59/0.15 /1,
         			0.175/0.62/0.2  /1,
         			0.525/0.62/0.55 /1,
         			0.55 /0.60/0.6 /1,
         			0.9  /0.63/0.925/1,
         			0.95 /0.68/0.975/1
         		}
         		{
         			\drawVerticalItem{\x*\w}{\y*\hprime-\h/4}{\xx*\w}{\yy*\hprime-\h/4};
         		}
         		
         		\foreach \x/\y/\xx/\yy in {
         			0.275/1/0.300/1.05,
         			0.300/1/0.325/1.06,
         			0.325/1/0.350/1.09,
         			0.85 /1/0.875/1.07,
         			0.875/1/0.9  /1.08
         		}
         		{
         			\drawVerticalItem{\x*\w}{\y*\hprime-\hprime/2}{\xx*\w}{\yy*\hprime-\hprime/2};
         		}
         		
         		\foreach \x/\y/\xx/\yy in {
         			0.025/0.0 /0.05 /0.35,
         			0.05 /0.0 /0.075/0.36,
         			0.15 /0.0 /0.175/0.39,
         			0.325/0.0 /0.35 /0.37,
         			0.525/0.0 /0.55  /0.34,
         			0.85 /0.0 /0.875/0.37,
         			0.925/0.0 /0.95 /0.66
         		}
         		{
         			\drawVerticalItem{\x*\w}{\y*\hprime}{\xx*\w}{\yy*\hprime};
         		}

         		\foreach \x/\y/\xx/\yy/\z in {
         			0.075 /0.0 /0.15 /0.28/,
         			0.175 /0.0 /0.325/0.31/$\ell_b$,
         			0.35  /0.0 /0.45 /0.29/,
         			0.45  /0.0 /0.525/0.29/,
         			0.55  /0.0 /0.65 /0.35/,
         			0.65  /0.0 /0.7  /0.53/,
         			0.7   /0.0 /0.75 /0.52/,
         			0.75  /0.0 /0.85 /0.30/,
         			0.875 /0.0 /0.925/0.29/$r_b$
         		}
         		{
         			\drawTallItem[\small \z]{\x*\w}{\y*\hprime}{\xx*\w}{\yy*\hprime};
         		}

         		\foreach \y/\z in {
         			0.5*\hprime 		/ $\nicefrac{1}{2}\height(B)$,
         			0.25*\h     		/ $\nicefrac{1}{4}H$,
         			\hprime - 0.25*\h	/ $\height(B) - \nicefrac{1}{4}H$,
         			\hprime				/ $\height(B)$,
         			\hprime + 0.25*\h	/ $\height(B) + \nicefrac{1}{4}H$
         		}{
         			\draw[dotted] (-0.1*\w,\y) -- (1.1*\w,\y) ;
         		}
         		
         		\foreach \x/\y/\xx/\yy in {
         			0.0*\w/\hprime-0.25*\h/0.2*\w/\hprime + 0.25*\h,  
         			0.9*\w/\hprime-0.25*\h/1.0*\w/\hprime +0.25*\h,   
         			0.0*\w/\hprime-0.25*\h/0.2*\w/0,                  
         			0.9*\w/\hprime-0.25*\h/1.0*\w/0,                  
         			0.2*\w/0/0.325*\w/0.5*\hprime,                    
         			0.2*\w /0.5*\hprime/0.35*\w/\hprime + 0.25*\h,    
         			0.75*\w/0.5*\hprime/0.9*\w/\hprime + 0.25*\h,     
         			0.65*\w/0.0*\hprime/0.75 *\w/\hprime -0.25*\h      
         		}{
         			\draw[very thick] (\x,\y) rectangle (\xx,\yy);
         			\draw[fill = white, opacity = 0.5] (\x,\y) rectangle (\xx,\yy);
         			
         		}

         		\draw[red] (0.425*\w,\hprime/2) -- (0.425*\w,\hprime +\h/4) node[above]{\small $L_3$};
         		\draw[red] (0.525*\w,\hprime/2) -- (0.525*\w,\hprime +\h/4) node[above]{\small $L_2$};
         		
         		\draw[very thick, red, pattern = dots,  pattern color = gray] (0.35*\w,\hprime/2) rectangle (0.425*\w,\hprime +\h/4);
         		
         		\draw[very thick,red] (0.35*\w,7*\h/8) -- (-0.2*\w,9*\hprime/8) node[above, black]{$B_{\ell,6}$} ;
         		
         		\end{tikzpicture}
         		\caption{The area $B_{\ell,6}$}
         		\label{fig:sub:AreaB6}
         	\end{subfigure}				
         	\caption{Illustration of the areas $B_{\ell,3}$ through $B_{\ell,6}$. When reordering the items to create areas $B_{\ell,4}$, $B_{r,4}$, and $B_{5}$, some items may overlap. We resolve this further on, when considering areas $B_{\ell,9}$ and $B_{r,9}$.}
         	\label{fig:GenralReordering2}
         \end{figure*}
         Let $\instance_{\ell,\nicefrac{1}{2}\height(B)}$ be the set of shifted items now touching $\nicefrac{1}{2}\height(B)$ with their lower border. All the items in $\instance_{\ell,\nicefrac{1}{2}\height(B)}$ have a height of at most $\nicefrac{1}{2}\height(B)$ because their upper border was at $\height(B)-\nicefrac{1}{4}H$ before. The area left of $L_1$ and right of the left border of $\ell$ below $\nicefrac{1}{2}\height(B)$ is called $B_{\ell,3}$. This area contains pseudo items touching $\nicefrac{1}{2}\height(B)$ and a part of $\ell_b$ at the bottom. We sort the pseudo items above $\ell_b$ touching $\nicefrac{1}{2}\height(B)$ in descending order of their heights.\newline\indent
         On the other hand, if $r(\ell_b)$ is left of $r(i_\ell)$, we introduce the line $L_1$ but do not shift any items. On the right of $i_r$, we introduce the same line, just named $R_1$, and the area $B_{r,3}$ analogously.\newline\indent
          \textit{Simple cases:} It is possible that $\ell_b$ and $r_b$ are the same item. Similarly, one of the lines $L_1$ or $R_1$ might intersect with $i_\ell$ or $i_r$ respectively. Finally, $L_1$ and $R_1$ might be at the exact same spot. For any of these cases to occur, there cannot be an item with height larger than $\nicefrac{1}{2}\height(B)$ touching the bottom of the box between $L_1$ and $R_1$. If $\ell_b=r_b,$ then this item cannot be taller than $\nicefrac{1}{2}\height(B)$ because $i_\ell$ and $i_r$ are the left and rightmost items with height larger than $\nicefrac{1}{2}\height(B)$. Since they are placed above $\ell_b=r_b$, this item cannot be larger than $\nicefrac{1}{2}\height(B)$ as the total would be larger than the box. Similarly, if $L_1$ intersects $i_r$, the items $\ell_b$ and $i_r$ are placed above each other leading to the same contradiction. Obviously, the reverse holds as well. Finally, if $L_1$ and $R_1$ are in the same spot the items $\ell_b$ and $r_b$ must be placed next to each other. Again, neither of these items can have a height greater than $\nicefrac{1}{2}\height(B)$, and they extend to the bottom of the items $i_\ell$ and $i_r$ respectively. When there is no item with height greater than $\nicefrac{1}{2}\height(B)$ touching the bottom of the box, we shift all items between $L_1$ and $R_1$ touching $\height(B)-\nicefrac{1}{4}H$ with their top such that they touch $\nicefrac{1}{2}\height(B)$ with their lower borders. The pseudo items touching $\height(B)-\nicefrac{1}{4}H$ with their lower borders in this area are placed such that their upper border touches $\nicefrac{1}{2}\height(B)$, just as we did with the items above $\ell_b$. Now, there is no item intersecting the horizontal line $\nicefrac{1}{2}\height(B)$. Hence, we can sort the items above $\nicefrac{1}{2}\height(B)$ between $i_\ell$ and $i_r$ by their heights in descending order from left to right. We can do the same for items that are below $\nicefrac{1}{2}\height(B)$. After this step, no further reordering is required.\newline\indent
         \textit{Areas $B_{\ell,4}$ and $B_5$:} Next, we consider the case that there is an item with height taller than $\nicefrac{1}{2}\height(B)$ touching the bottom of the box between $i_\ell$ and $i_r$. This case requires further reordering. The aim is to reorder the (pseudo) items with height $\nicefrac{1}{2}\height(B)+\nicefrac{1}{4}H$ touching the top of the box in a way that they build two blocks next to $i_\ell$ and $i_r$ respectively. These blocks will be the areas $B_{\ell,4}$ and $B_{r,4}$, see \cref{fig:sub:AreaB4andB5}. To facilitate this reordering, we have to define a border between $i_\ell$ and $i_r$ such that all of the items with height $\nicefrac{1}{2}\height(B)+\nicefrac{1}{4}H$ are shifted to $i_\ell$ on the left and $i_r$ on the right. As we know one such item must exist, let $i$ be an item of height larger than $\nicefrac{1}{2}\height(B)$ touching the bottom of the packing. Clearly, no items of height $\nicefrac{1}{2}\height(B)+\nicefrac{1}{4}H$ can be above $i$, even in the box we extended by $\nicefrac{1}{4}H$. This item defines the border between $i_\ell$ and $i_r$.\newline\indent
         Consider items with height $\nicefrac{1}{2}\height(B)+\nicefrac{1}{4}H$ touching the top of the box between $i_\ell$ and $i$. We shift those items left of $i$ to the left until they touch $i_\ell$. We shift the items between $i$ and $i_r$ analogously, but to the right. All other items with parts above $\nicefrac{1}{2}\height(B)$ are shifted to the left and right accordingly to avoid overlap. This must be possible because these items could not be placed below the items of height $\nicefrac{1}{2}\height(B)+\nicefrac{1}{4}H$ originally, and the total available width left and right of $i$ has not changed. We sort the items with height $\nicefrac{1}{2}\height(B)+\nicefrac{1}{4}H$ such that the pseudo items that include tall items of equal height are placed adjacently. The area including these items to the left of $i$, including $i_\ell$, is called $B_{\ell,4}$.\newline\indent
         While doing the shift discussed above, we shift all items between $i_\ell$ and $i_r$ with height $\height(B)-\nicefrac{1}{4}H$ touching the bottom of the box such that they are next to $i$ and shift the other items to the left or right accordingly. This is feasible for the same reasons as above. This comprises the area $B_5$, which can also be seen in \cref{fig:sub:AreaB4andB5}.\newline\indent
         However, as an interplay between the areas $B_{\ell,4}$ and $B_5$, it may occur that items touching $\height(B)-\nicefrac{1}{4}H$ are overlapping with items that touch the bottom of the box. This overlap gets handled later on, in the areas $B_{\ell,9}$ and $B_{r,9}$ respectively.\newline\indent

         \textit{Area $B_{\ell,6}$:} Note that items in the set $I_{\ell,\nicefrac{1}{2}\height(B)}$ are now placed next to each other. Before the last step it was possible that items with $\nicefrac{1}{2}\height(B)+\nicefrac{1}{4}H$ were positioned between them. Additionally, there is no item part touching the top of the box above an item that touches $\height(B)-\nicefrac{1}{4}H$ with its bottom border. Furthermore, the total width of items with their bottom border between $\nicefrac{1}{4}H$ and $\nicefrac{1}{2}\height(B)$, that are placed entirely between $L_1$ and the right border of $i$ did not change.\newline\indent
         If $\ell_m$ exists, we draw the vertical line $L_2$ at the right border of $\ell_m$ and a vertical line $L_3$ at the left border of $\ell_m$. Define $\ell_{t,r}$ and $\ell_{t,\ell}$ to be the tall items touching the top of the border intersected by these lines if there are any. We look at the left of $L_3$ and right of $B_{\ell,4}$, which is bounded at the top by the top of the box, while the bottom border is at $\nicefrac{1}{2}\height(B)$. We call this area $B_{\ell,6}$. In this area each item touches the bottom or the top and there is at most one item $\ell_{t,\ell}$ intersecting the border. We use the reordering from \cref{lem:ReorderTallVertBox12} to reorder the items in $B_{\ell,6}$. See this area in \cref{fig:sub:AreaB6}.\newline\indent
         \textit{Area $B_{\ell,7}$:} The area above the item $\ell_m$ is called $B_{\ell,7}$, find an illustration of it in \cref{fig:sub:AreaB7}. In this area, all items are touching either the top of the box or $\ell_m$. All items that touch only $\ell_m$ with their bottom are pseudo items. We order the items touching the top of the box in ascending order regarding their heights and move those pseudo items below with them. Next, we look at the overlapping items $\ell_{t,r}$ and $\ell_{t,\ell}$. We move items with their heights next to the items themselves, i.e.\ items with height $\height(\ell_{t,r})$ next to $\ell_{t,r}$ and the same for $\ell_{t,\ell}$. This results in three areas to place pseudo items in. The first two areas are below the respective overlapping items $\ell_{t,r}$, $\ell_{t,\ell}$ and their accompanying items with identical height, while the last is between these two areas. For all those areas, we sort the pseudo items in them by order of descending height.\newline\indent
         The areas~$B_{\ell,6}$ and $B_{\ell,7}$ only exist if the item $\ell_m$ exists. If this is not the case, we introduce the vertical line~$L_2$ at the left border of the area $B_{\ell,4}$. The same holds for the areas~$B_{r,6}$, $B_{r,7}$ and the vertical line~$R_2$, which will be placed at the right border of $B_{r,4}$.\newline\indent
         \textit{Area $B_{8}$:} Next, we inspect the area above $\height(B)-\nicefrac{1}{4}H$ and between the lines~$L_2$ and $R_2$. This area is called $B_8$, see \cref{fig:sub:AreaB8}. There are at most two immovable items overlapping this area. One item~$\ell_{t,r}$ on the left touching the top of the box, and a similar item~$r_{t,\ell}$ on the right, also touching the top of the box. Since $B_8$ does not contain any item of height $\nicefrac{1}{2}\height(B)+\nicefrac{1}{4}H$ or items from the sets~$I_{\ell,\nicefrac{1}{2}\height(B)}$ or $I_{r,\nicefrac{1}{2}\height(B)}$, each item either touches the top or bottom of this area. Furthermore, all items touching the bottom are pseudo items. Therefore, we can sort this area in the same manner we did $B_{\ell,7}$.\newline\indent
         \textit{Area $B_{\ell,9}$:} Finally, we have to look at the items on the bottom between $L_1$ and $R_1$, as well as the items touching $\height(B)-\nicefrac{1}{4}H$ with their top between $L_2$ and $R_2$. We consider the items touching the bottom between $L_1$ and the left border of $B_5$ and the items touching $\height(B)-\nicefrac{1}{4}H$ with their top between $L_2$ and the left border of $B_5$. The area containing these items is called $B_{\ell,9}$. In this area, we sort all items touching $\height(B)-\nicefrac{1}{4}H$ in ascending order of their heights and the items at the bottom in descending order of their heights, from outside inward. We do the same for the right side in the area $B_{r,9}$. This procedure is shown in \cref{fig:sub:AreaB9}.
         \begin{claim}
             After this step, there is no item that overlaps another item in the area $B_{\ell,9}$.
         \end{claim}
         \begin{proof}
             First of all, let us note that all items that have their bottom border on the bottom of the box have a height of less than $\nicefrac{1}{2}\height(B)$. Otherwise, they would be contained in the area $B_5$. Thus, they cannot overlap with items that have their upper border on the top of the box, as those have a maximum height of $\nicefrac{1}{2}\height(B)+\nicefrac{1}{4}H$.\newline\indent
             Let us now assume the claim is false, i.e.\ there is an item $b$ from the bottom of the box intersecting an item $t$ from the top of the box. Let this intersection be at a point $(x,y)$ in the area $B_{\ell,9}$. Due to the manner in which we sorted the items touching either box border, all items touching $\height(B)-\nicefrac{1}{4}H$ left of $(x,y)$ and inside $B_{\ell,9}$ must also overlap the horizontal line at height $y$. Analogously, all items to the right of $(x,y)$ and the left of $L_1$ must also overlap this horizontal line. Note that the total width of items with lower border below $y$ and above $\nicefrac{1}{4}H$ between $L_1$ and the left border of $i$ has not changed after the shifting of items with height $\nicefrac{1}{2}\height(B)+\nicefrac{1}{4}H$ on the top of the box. This is because items left of $\ell_m$ have their lower border at $\nicefrac{1}{2}\height(B)$. Additionally, the total width of items touching the bottom of the box with upper border above $y$ in this area has not changed either. Therefore, the total width of items overlapping the horizontal line at $y$ in this area is larger than the width of this area. Thus, the items must have had some overlap before the first horizontal shift. This is a contradiction because we assumed a feasible optimal packing. Therefore, there cannot be an item overlapping another item inside $B_{\ell,9}$.
         \end{proof}
     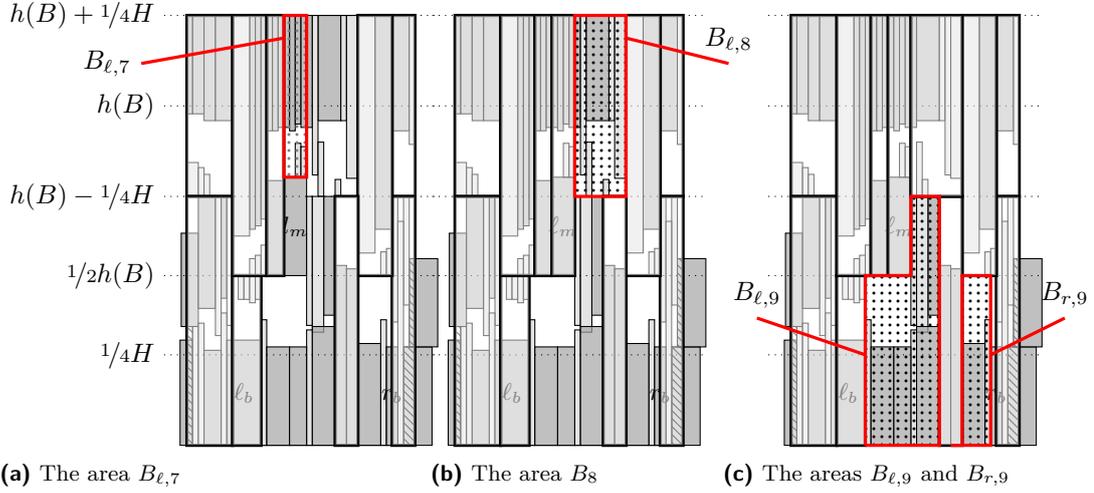
\begin{figure*}[ht]
     	\centering
     	\begin{subfigure}[t]{0.37\textwidth}
     		\centering
     		\begin{tikzpicture}
     		\pgfmathsetmacro{\w}{3}
     		\pgfmathsetmacro{\h}{4.8}
     		\pgfmathsetmacro{\hprime}{4.5}
     		
     		\draw (0*\w,0) rectangle (1*\w,\hprime +\h/4);
     		
     		\foreach \x/\y/\xx/\yy in {
     			-0.03 /0.00/0.025/0.31,
     			-0.025/0.35/0.05/0.625,
     			1.075 /0.00/0.95/0.29,
     			1.1   /0.29/0.975/0.55
     		}
     		{
     			\drawTallItem{\x*\w}{\y*\hprime}{\xx*\w}{\yy*\hprime};
     		}
     		
     		\foreach \x/\y/\xx/\yy in {
     			0    /0   /0.025/0.35,
     			0.950/0   /1    /0.29,
     			0.975/0.29/1    /0.55
     		}{
     			\draw[pattern = north west lines] (\x*\w,\y*\hprime) rectangle (\xx*\w,\yy*\hprime);
     		}
     		
     		\foreach \x/\y/\xx/\yy in {
     			0.05 /0.67/0.125/1,
     			0.15 /0.70/0.175/1,
     			0.6 /0.65/0.65  /1
     		}
     		{
     			\drawTallItem{\x*\w}{\y*\hprime-\h/4}{\xx*\w}{\yy*\hprime-\h/4};
     		}
     		
     		\foreach \x/\y/\xx/\yy/\z in {
     			0.35 /1/0.425  /1.28/,
     			0.425/1/0.525  /1.29/$l_m$
     		}
     		{
     			\drawTallItem[\small \z]{\x*\w}{\y*\hprime-0.5*\hprime}{\xx*\w}{\yy*\hprime-0.5* \hprime};
     		}
     		\foreach \x/\y/\xx/\yy in {
     			0.0   /0.71 /0.025/1,
     			0.025 /0.71 /0.075/1,
     			0.075 /0.69 /0.125/1,
     			0.125 /0.69 /0.2  /1,
     			0.275 /0.37 /0.3  /1,
     			0.3   /0.38 /0.325/1,
     			0.325 /0.40 /0.35 /1,
     			0.35  /0.66 /0.375/1,
     			0.375 /0.66 /0.4  /1,
     			0.4   /0.67 /0.45 /1,
     			0.45  /0.66 /0.475/1,
     			0.475 /0.68 /0.5  /1,
     			0.5   /0.67 /0.55 /1,
     			0.575 /0.69 /0.675/1,
     			0.825 /0.37 /0.9  /1,
     			0.9   /0.66 /0.975/1
     		}
     		{
     			\drawTallItem{\x*\w}{\y*\hprime+\h/4}{\xx*\w}{\yy*\hprime+\h/4};
     		}
     		
     		\foreach \x/\y/\xx/\yy in {
     			0.2  /0.35/0.275/1,
     			0.55 /0.53/0.575/1,
     			0.675/0.69/0.7  /1,
     			0.7  /0.52/0.75 /1,
     			0.75 /0.30/0.825/1,
     			0.975/0.62/1    /1
     		}
     		{
     			\drawVerticalItem{\x*\w}{\y*\hprime+\h/4}{\xx*\w}{\yy*\hprime+\h/4};
     		}
     		
     		\foreach \x/\y/\xx/\yy in {
     			0.000/1/0.05/1.10,
     			0.05 /1/0.075/1.085,
     			0.075/1/0.1 /1.065,
     			0.575/1/0.6  /1.16,
     			0.65 /1/0.675/1.05
     		}
     		{
     			\drawVerticalItem{\x*\w}{\y*\hprime-\h/4}{\xx*\w}{\yy*\hprime-\h/4};
     		}
     		
     		\foreach \x/\y/\xx/\yy in {
     			0.225/0.93/0.25 /1,
     			0.25 /0.93/0.275/1,
     			0.275/0.92/0.3  /1,
     			0.3  /0.93/0.325/1,
     			0.475 /1.29/0.5/1.39,
     			0.5/1.29/0.525  /1.38
     		}
     		{
     			\drawVerticalItem{\x*\w}{\y*\hprime-0.5*\hprime}{\xx*\w}{\yy*\hprime-0.5*\hprime};
     		}

     		\foreach \x/\y/\xx/\yy in {
     			0.125/0.59/0.15 /1,
     			0.175/0.62/0.2  /1,
     			0.525/0.62/0.55 /1,
     			0.55 /0.60/0.6 /1,
     			0.9  /0.63/0.925/1,
     			0.95 /0.68/0.975/1
     		}
     		{
     			\drawVerticalItem{\x*\w}{\y*\hprime-\h/4}{\xx*\w}{\yy*\hprime-\h/4};
     		}
     		
     		\foreach \x/\y/\xx/\yy in {
     			0.275/1/0.300/1.05,
     			0.300/1/0.325/1.06,
     			0.325/1/0.350/1.09,
     			0.85 /1/0.875/1.07,
     			0.875/1/0.9  /1.08
     		}
     		{
     			\drawVerticalItem{\x*\w}{\y*\hprime-\hprime/2}{\xx*\w}{\yy*\hprime-\hprime/2};
     		}
     		
     		\foreach \x/\y/\xx/\yy in {
     			0.025/0.0 /0.05 /0.35,
     			0.05 /0.0 /0.075/0.36,
     			0.15 /0.0 /0.175/0.39,
     			0.325/0.0 /0.35 /0.37,
     			0.525/0.0 /0.55  /0.34,
     			0.85 /0.0 /0.875/0.37,
     			0.925/0.0 /0.95 /0.66
     		}
     		{
     			\drawVerticalItem{\x*\w}{\y*\hprime}{\xx*\w}{\yy*\hprime};
     		}

     		\foreach \x/\y/\xx/\yy/\z in {
     			0.075 /0.0 /0.15 /0.28/,
     			0.175 /0.0 /0.325/0.31/$\ell_b$,
     			0.35  /0.0 /0.45 /0.29/,
     			0.45  /0.0 /0.525/0.29/,
     			0.55  /0.0 /0.65  /0.35/,
     			0.65  /0.0 /0.7  /0.53/,
     			0.7   /0.0 /0.75 /0.52/,
     			0.75  /0.0 /0.85 /0.30/,
     			0.875 /0.0 /0.925/0.29/$r_b$
     		}
     		{
     			\drawTallItem[\small \z]{\x*\w}{\y*\hprime}{\xx*\w}{\yy*\hprime};
     		}

     		\foreach \y/\z in {
     			0.5*\hprime 		/ $\nicefrac{1}{2}\height(B)$,
     			0.25*\h     		/ $\nicefrac{1}{4}H$,
     			\hprime - 0.25*\h	/ $\height(B) - \nicefrac{1}{4}H$,
     			\hprime				/ $\height(B)$,
     			\hprime + 0.25*\h	/ $\height(B) + \nicefrac{1}{4}H$
     		}{
     			\draw[dotted] (-0.1*\w,\y)node[left]{\z} -- (1.1*\w,\y) ;
     		}

     		\foreach \x/\y/\xx/\yy in {
     			0.0*\w/\hprime-0.25*\h/0.2*\w/\hprime + 0.25*\h,  
     			0.9*\w/\hprime-0.25*\h/1.0*\w/\hprime +0.25*\h,   
     			0.0*\w/\hprime-0.25*\h/0.2*\w/0,                  
     			0.9*\w/\hprime-0.25*\h/1.0*\w/0,                  
     			0.2*\w/0/0.325*\w/0.5*\hprime,                    
     			0.2*\w /0.5*\hprime/0.35*\w/\hprime + 0.25*\h,    
     			0.75*\w/0.5*\hprime/0.9*\w/\hprime + 0.25*\h,     
     			0.65*\w/0.0*\hprime/0.75 *\w/\hprime -0.25*\h,    
     			0.35*\w/0.5*\hprime/0.425*\w/\hprime + 0.25*\h 
     		}{
     			\draw[very thick] (\x,\y) rectangle (\xx,\yy);
     			\draw[fill = white, opacity = 0.5] (\x,\y) rectangle (\xx,\yy);
     			
     		}

     		\draw[very thick,red] (0.425*\w,9*\h/8) -- (-0.2*\w,9*\hprime/8) node[left, black]{$B_{\ell,7}$};
     		
     		\draw[very thick, red, pattern = dots, pattern color = gray] (0.425*\w,1.29*\hprime-0.5*\hprime) rectangle (0.525*\w,\hprime +\h/4);
     		
     		\end{tikzpicture}
     		\caption{The area $B_{\ell,7}$}
     		\label{fig:sub:AreaB7}
     	\end{subfigure}
     	\hfill
     	\begin{subfigure}[t]{0.24\textwidth}
     		\centering
     		\begin{tikzpicture}
     		\pgfmathsetmacro{\w}{3}
     		\pgfmathsetmacro{\h}{4.8}
     		\pgfmathsetmacro{\hprime}{4.5}
     		\draw (0*\w,0) rectangle (1*\w,\hprime +\h/4);
     		
     		\foreach \x/\y/\xx/\yy in {
     			-0.03 /0.00/0.025/0.31,
     			-0.025/0.35/0.05/0.625,
     			1.075 /0.00/0.95/0.29,
     			1.1   /0.29/0.975/0.55
     		}
     		{
     			\drawTallItem{\x*\w}{\y*\hprime}{\xx*\w}{\yy*\hprime};
     		}
     		
     		\foreach \x/\y/\xx/\yy in {
     			0    /0   /0.025/0.35,
     			0.950/0   /1    /0.29,
     			0.975/0.29/1    /0.55
     		}{
     			\draw[pattern = north west lines] (\x*\w,\y*\hprime) rectangle (\xx*\w,\yy*\hprime);
     		}
     		
     		\foreach \x/\y/\xx/\yy in {
     			0.05 /0.67/0.125/1,
     			0.15 /0.70/0.175/1,
     			0.6 /0.65/0.65  /1
     		}
     		{
     			\drawTallItem{\x*\w}{\y*\hprime-\h/4}{\xx*\w}{\yy*\hprime-\h/4};
     		}
     		
     		\foreach \x/\y/\xx/\yy/\z in {
     			0.35 /1/0.425  /1.28/,
     			0.425/1/0.525  /1.29/$\ell_m$
     		}
     		{
     			\drawTallItem[\small \z]{\x*\w}{\y*\hprime-0.5*\hprime}{\xx*\w}{\yy*\hprime - 0.5* \hprime};
     		}
     		\foreach \x/\y/\xx/\yy in {
     			0.0   /0.71 /0.025/1,
     			0.025 /0.71 /0.075/1,
     			0.075 /0.69 /0.125/1,
     			0.125 /0.69 /0.2  /1,
     			0.275 /0.37 /0.3  /1,
     			0.3   /0.38 /0.325/1,
     			0.325 /0.40 /0.35 /1,
     			0.35  /0.66 /0.375/1,
     			0.375 /0.66 /0.4  /1,
     			0.4   /0.67 /0.45 /1,
     			0.45  /0.66 /0.475/1,
     			0.475 /0.68 /0.5  /1,
     			0.5   /0.67 /0.55 /1,
     			0.575 /0.69 /0.675/1,
     			0.825 /0.37 /0.9  /1,
     			0.9   /0.66 /0.975/1
     		}
     		{
     			\drawTallItem{\x*\w}{\y*\hprime+\h/4}{\xx*\w}{\yy*\hprime+\h/4};
     		}
     		
     		\foreach \x/\y/\xx/\yy in {
     			0.2  /0.35/0.275/1,
     			0.55 /0.53/0.575/1,
     			0.675/0.69/0.7  /1,
     			0.7  /0.52/0.75 /1,
     			0.75 /0.30/0.825/1,
     			0.975/0.62/1    /1
     		}
     		{
     			\drawVerticalItem{\x*\w}{\y*\hprime+\h/4}{\xx*\w}{\yy*\hprime+\h/4};
     		}
     		
     		\foreach \x/\y/\xx/\yy in {
     			0.000/1/0.05/1.10,
     			0.05 /1/0.075/1.085,
     			0.075/1/0.1 /1.065,
     			0.575/1/0.6  /1.16,
     			0.65 /1/0.675/1.05
     		}
     		{
     			\drawVerticalItem{\x*\w}{\y*\hprime-\h/4}{\xx*\w}{\yy*\hprime-\h/4};
     		}
     		
     		\foreach \x/\y/\xx/\yy in {
     			0.225/0.93/0.25 /1,
     			0.25 /0.93/0.275/1,
     			0.275/0.92/0.3  /1,
     			0.3  /0.93/0.325/1,
     			0.475 /1.29/0.5/1.39,
     			0.5/1.29/0.525  /1.38
     		}
     		{
     			\drawVerticalItem{\x*\w}{\y*\hprime-0.5*\hprime}{\xx*\w}{\yy*\hprime-0.5*\hprime};
     		}

     		\foreach \x/\y/\xx/\yy in {
     			0.125/0.59/0.15 /1,
     			0.175/0.62/0.2  /1,
     			0.525/0.62/0.55 /1,
     			0.55 /0.60/0.6 /1,
     			0.9  /0.63/0.925/1,
     			0.95 /0.68/0.975/1
     		}
     		{
     			\drawVerticalItem{\x*\w}{\y*\hprime-\h/4}{\xx*\w}{\yy*\hprime-\h/4};
     		}
     		
     		\foreach \x/\y/\xx/\yy in {
     			0.275/1/0.300/1.05,
     			0.300/1/0.325/1.06,
     			0.325/1/0.350/1.09,
     			0.85 /1/0.875/1.07,
     			0.875/1/0.9  /1.08
     		}
     		{
     			\drawVerticalItem{\x*\w}{\y*\hprime-\hprime/2}{\xx*\w}{\yy*\hprime-\hprime/2};
     		}
     		
     		\foreach \x/\y/\xx/\yy in {
     			0.025/0.0 /0.05 /0.35,
     			0.05 /0.0 /0.075/0.36,
     			0.15 /0.0 /0.175/0.39,
     			0.325/0.0 /0.35 /0.37,
     			0.525/0.0 /0.55  /0.34,
     			0.85 /0.0 /0.875/0.37,
     			0.925/0.0 /0.95 /0.66
     		}
     		{
     			\drawVerticalItem{\x*\w}{\y*\hprime}{\xx*\w}{\yy*\hprime};
     		}

     		\foreach \x/\y/\xx/\yy/\z in {
     			0.075 /0.0 /0.15 /0.28/,
     			0.175 /0.0 /0.325/0.31/$\ell_b$,
     			0.35  /0.0 /0.45 /0.29/,
     			0.45  /0.0 /0.525/0.29/,
     			0.55  /0.0 /0.65  /0.35/,
     			0.65  /0.0 /0.7  /0.53/,
     			0.7   /0.0 /0.75 /0.52/,
     			0.75  /0.0 /0.85 /0.30/,
     			0.875 /0.0 /0.925/0.29/$r_b$
     		}
     		{
     			\drawTallItem[\small \z]{\x*\w}{\y*\hprime}{\xx*\w}{\yy*\hprime};
     		}

     		\foreach \y/\z in {
     			0.5*\hprime 		/ $\nicefrac{1}{2}\height(B)$,
     			0.25*\h     		/ $\nicefrac{1}{4}H$,
     			\hprime - 0.25*\h	/ $\height(B) - \nicefrac{1}{4}H$,
     			\hprime				/ $\height(B)$,
     			\hprime + 0.25*\h	/ $\height(B) + \nicefrac{1}{4}H$
     		}{
     			\draw[dotted] (-0.1*\w,\y) -- (1.1*\w,\y) ;
     		}

     		\foreach \x/\y/\xx/\yy in {
     			0.0*\w/\hprime-0.25*\h/0.2*\w/\hprime + 0.25*\h,  
     			0.9*\w/\hprime-0.25*\h/1.0*\w/\hprime +0.25*\h,   
     			0.0*\w/\hprime-0.25*\h/0.2*\w/0,                  
     			0.9*\w/\hprime-0.25*\h/1.0*\w/0,                  
     			0.2*\w/0/0.325*\w/0.5*\hprime,                    
     			0.2*\w /0.5*\hprime/0.35*\w/\hprime + 0.25*\h,    
     			0.75*\w/0.5*\hprime/0.9*\w/\hprime + 0.25*\h,     
     			0.65*\w/0.0*\hprime/0.75 *\w/\hprime -0.25*\h,    
     			0.35*\w/0.5*\hprime/0.425*\w/\hprime + 0.25*\h,    
     			0.425*\w/0.5*\hprime/0.525*\w/\hprime +0.25*\h
     		}{
     			\draw[very thick] (\x,\y) rectangle (\xx,\yy);
     			\draw[fill = white, opacity = 0.5] (\x,\y) rectangle (\xx,\yy);
     			
     		}

     		
     		\draw[very thick, red, pattern = dots] (0.525*\w,\hprime -\h/4) rectangle (0.75*\w,\hprime +\h/4);
     		
     		\draw[very thick,red] (0.75*\w,9*\h/8) -- (1.2*\w,9*\hprime/8) node[above, black]{$B_{\ell,8}$};
     		\end{tikzpicture}
     		\caption{The area $B_{8}$}
     		\label{fig:sub:AreaB8}
     	\end{subfigure}
     	\hfill
     	\begin{subfigure}[t]{0.32\textwidth}
     		\centering
     		\begin{tikzpicture}
     		\pgfmathsetmacro{\w}{3}
     		\pgfmathsetmacro{\h}{4.8}
     		\pgfmathsetmacro{\hprime}{4.5}
     		
     		\draw (0*\w,0) rectangle (1*\w,\hprime +\h/4);
     		
     		\draw (0*\w,0) rectangle (1*\w,\hprime +\h/4);
     		
     		\foreach \x/\y/\xx/\yy in {
     			-0.03 /0.00/0.025/0.31,
     			-0.025/0.35/0.05/0.625,
     			1.075 /0.00/0.95/0.29,
     			1.1   /0.29/0.975/0.55
     		}
     		{
     			\drawTallItem{\x*\w}{\y*\hprime}{\xx*\w}{\yy*\hprime};
     		}
     		
     		\foreach \x/\y/\xx/\yy in {
     			0    /0   /0.025/0.35,
     			0.950/0   /1    /0.29,
     			0.975/0.29/1    /0.55
     		}{
     			\draw[pattern = north west lines] (\x*\w,\y*\hprime) rectangle (\xx*\w,\yy*\hprime);
     		}
     		
     		\foreach \x/\y/\xx/\yy in {
     			0.05 /0.67/0.125/1,
     			0.15 /0.70/0.175/1,
     			0.6 /0.65/0.65  /1
     		}
     		{
     			\drawTallItem{\x*\w}{\y*\hprime-\h/4}{\xx*\w}{\yy*\hprime-\h/4};
     		}
     		
     		\foreach \x/\y/\xx/\yy/\z in {
     			0.35 /1/0.425  /1.28/,
     			0.425/1/0.525  /1.29/$\ell_m$
     		}
     		{
     			\drawTallItem[\small \z]{\x * \w}{\y * \hprime - 0.5 * \hprime}{\xx * \w}{\yy * \hprime - 0.5 * \hprime};
     		}
     		\foreach \x/\y/\xx/\yy in {
     			0.0   /0.71 /0.025/1,
     			0.025 /0.71 /0.075/1,
     			0.075 /0.69 /0.125/1,
     			0.125 /0.69 /0.2  /1,
     			0.275 /0.37 /0.3  /1,
     			0.3   /0.38 /0.325/1,
     			0.325 /0.40 /0.35 /1,
     			0.35  /0.66 /0.375/1,
     			0.375 /0.66 /0.4  /1,
     			0.4   /0.67 /0.45 /1,
     			0.45  /0.66 /0.475/1,
     			0.475 /0.68 /0.5  /1,
     			0.5   /0.67 /0.55 /1,
     			0.55  /0.69 /0.65 /1,
     			0.825 /0.37 /0.9  /1,
     			0.9   /0.66 /0.975/1
     		}
     		{
     			\drawTallItem{\x*\w}{\y*\hprime+\h/4}{\xx*\w}{\yy*\hprime+\h/4};
     		}
     		
     		\foreach \x/\y/\xx/\yy in {
     			0.2  /0.35/0.275/1,
     			0.65 /0.69/0.675/1,
     			0.675/0.53/0.7/1,
     			0.7  /0.52/0.75 /1,
     			0.75 /0.30/0.825/1,
     			0.975/0.62/1    /1
     		}
     		{
     			\drawVerticalItem{\x*\w}{\y*\hprime+\h/4}{\xx*\w}{\yy*\hprime+\h/4};
     		}
     		
     		\foreach \x/\y/\xx/\yy in {
     			0.000/1/0.05/1.10,
     			0.05 /1/0.075/1.085,
     			0.075/1/0.1 /1.065,
     			0.55 /1/0.575  /1.16,
     			0.575/1/0.6    /1.05
     		}
     		{
     			\drawVerticalItem{\x*\w}{\y*\hprime-\h/4}{\xx*\w}{\yy*\hprime-\h/4};
     		}
     		
     		\foreach \x/\y/\xx/\yy in {
     			0.225/0.93/0.25 /1,
     			0.25 /0.93/0.275/1,
     			0.275/0.92/0.3  /1,
     			0.3  /0.93/0.325/1,
     			0.475 /1.29/0.5/1.39,
     			0.5/1.29/0.525  /1.38
     		}
     		{
     			\drawVerticalItem{\x*\w}{\y*\hprime-0.5*\hprime}{\xx*\w}{\yy*\hprime-0.5*\hprime};
     		}

     		\foreach \x/\y/\xx/\yy in {
     			0.125/0.59/0.15 /1,
     			0.175/0.62/0.2  /1,
     			0.525/0.62/0.55 /1,
     			0.55 /0.60/0.6 /1,
     			0.9  /0.63/0.925/1,
     			0.95 /0.68/0.975/1
     		}
     		{
     			\drawVerticalItem{\x*\w}{\y*\hprime-\h/4}{\xx*\w}{\yy*\hprime-\h/4};
     		}
     		
     		\foreach \x/\y/\xx/\yy in {
     			0.275/1/0.300/1.05,
     			0.300/1/0.325/1.06,
     			0.325/1/0.350/1.09,
     			0.85 /1/0.875/1.07,
     			0.875/1/0.9  /1.08
     		}
     		{
     			\drawVerticalItem{\x*\w}{\y*\hprime-\hprime/2}{\xx*\w}{\yy*\hprime-\hprime/2};
     		}
     		
     		\foreach \x/\y/\xx/\yy in {
     			0.025/0.0 /0.05 /0.35,
     			0.05 /0.0 /0.075/0.36,
     			0.15 /0.0 /0.175/0.39,
     			0.325/0.0 /0.35 /0.37,
     			0.525/0.0 /0.55  /0.34,
     			0.85 /0.0 /0.875/0.37,
     			0.925/0.0 /0.95 /0.66
     		}
     		{
     			\drawVerticalItem{\x*\w}{\y*\hprime}{\xx*\w}{\yy*\hprime};
     		}

     		\foreach \x/\y/\xx/\yy/\z in {
     			0.075 /0.0 /0.15 /0.28/,
     			0.175 /0.0 /0.325/0.31/$\ell_b$,
     			0.35  /0.0 /0.45 /0.29/,
     			0.45  /0.0 /0.525/0.29/,
     			0.55  /0.0 /0.65  /0.35/,
     			0.65  /0.0 /0.7  /0.53/,
     			0.7   /0.0 /0.75 /0.52/,
     			0.75  /0.0 /0.85 /0.30/,
     			0.875 /0.0 /0.925/0.29/$r_b$
     		}
     		{
     			\drawTallItem[\small \z]{\x*\w}{\y*\hprime}{\xx*\w}{\yy*\hprime};
     		}

     		\foreach \y/\z in {
     			0.5*\hprime 		/ $\nicefrac{1}{2}\height(B)$,
     			0.25*\h     		/ $\nicefrac{1}{4}H$,
     			\hprime - 0.25*\h	/ $\height(B) - \nicefrac{1}{4}H$,
     			\hprime				/ $\height(B)$,
     			\hprime + 0.25*\h	/ $\height(B) + \nicefrac{1}{4}H$
     		}{
     			\draw[dotted] (-0.1*\w,\y) -- (1.1*\w,\y) ;
     		}

     		\foreach \x/\y/\xx/\yy in {
     			0.0*\w/\hprime-0.25*\h/0.2*\w/\hprime + 0.25*\h,  
     			0.9*\w/\hprime-0.25*\h/1.0*\w/\hprime +0.25*\h,   
     			0.0*\w/\hprime-0.25*\h/0.2*\w/0,                  
     			0.9*\w/\hprime-0.25*\h/1.0*\w/0,                  
     			0.2*\w/0/0.325*\w/0.5*\hprime,                    
     			0.2*\w /0.5*\hprime/0.35*\w/\hprime + 0.25*\h,    
     			0.75*\w/0.5*\hprime/0.9*\w/\hprime + 0.25*\h,     
     			0.65*\w/0.0*\hprime/0.75 *\w/\hprime -0.25*\h,    
     			0.35*\w/0.5*\hprime/0.425*\w/\hprime + 0.25*\h,   
     			0.425*\w/0.5*\hprime/0.525*\w/\hprime +0.25*\h,   
     			0.525*\w/\hprime -0.25*\h/0.75*\w/\hprime +0.25*\h     
     		}{
     			\draw[very thick] (\x,\y) rectangle (\xx,\yy);
     			\draw[fill = white, opacity = 0.5] (\x,\y) rectangle (\xx,\yy);
     			
     		}
     		
     		\draw[very thick,red,pattern=dots] (0.325*\w,0*\hprime) -- (0.325*\w,\hprime/2) -- (0.525*\w,\hprime/2) -- (0.525*\w,\hprime -\h/4) -- (0.65*\w,\hprime -\h/4) -- (0.65*\w,0*\hprime) -- (0.325*\w,0*\hprime);
     		
     		\draw[very thick,red,pattern=dots] (0.75*\w,0*\hprime) -- (0.75*\w,0.58*\hprime -\h/4) -- (0.75*\w,0.58*\hprime -\h/4)--(0.75*\w,\hprime/2) -- (0.875*\w,\hprime /2) -- (0.875*\w,0) --(0.7*\w,0);
     		
     		\draw[very thick,red] (0.325*\w,\h/4) -- (-0.15*\w,3*\hprime/8) node[above, black]{$B_{\ell,9}$};
     		\draw[very thick,red] (0.875*\w,\h/4) -- (1.2*\w,3*\hprime/8) node[above, black]{$B_{r,9}$};
     		
     		\end{tikzpicture}
     		\caption{The areas $B_{\ell,9}$ and $B_{r,9}$}
     		\label{fig:sub:AreaB9}
     	\end{subfigure}
     	\caption{The definition of the final areas, and the reordering of items inside them.}
     	\label{fig:GenralReordering4}
     \end{figure*}
         \textit{Items with height $\height(B)$:} Finally, we consider the case that there are (pseudo) items with height $\height(B)$ inside the box. This can be handled in a very simple manner. We select one of these items, and shift it all the way to right of the box. Next, we shift all other items with height $\height(B)$ right next to this item. Since these items completely overlapped the box vertically due to their height, shifting them to one border of the packing results in the remainder still being feasible, after shifting them over correspondingly. Afterwards, we can continue with the procedure described above, effectively treating the box as being smaller by the width of the items with height $\height(B)$.\newline\indent
         \textit{Analyzing the number of constructed boxes.} In the worst case, we have (pseudo) items with height $\height(B)$ and both $\ell$ and $r$ exist. Furthermore, the left border of $\ell_b$ should be right of the left border of $\ell$, as well as the left border of $r_b$ being left of $r$. For repacking purposes, we did not require the assumption we got from our rounding, that tall items are placed on an arithmetic grid. To analyze the the number of boxes, however, it is convenient to use this assumption. To begin, we analyze the number of boxes for tall items we generate.
         \begin{claim}
             The number of boxes for tall items is bounded by $2N^2+(14N)/4+8$, where $H/N$ is the distance between the grid lines.
         \end{claim}
         \begin{proof}
             We show that the above is true by iteratively inspecting the generated areas and counting the number of generated boxes inside each of them. In the following, whenever we name a area, we only refer to the left area, but the arguments hold for both the left and right areas, as they are analogous.\newline\indent
             \textit{Area $B_{\ell,1}$: $N/4$ boxes.} In these areas, there are tall items with heights between $\nicefrac{1}{4}H$ and $1/2H$ at the top of the box. For each of these sizes, we generate at most one box in each area. Therefore, both contain at most $N/4$ boxes for tall items.\newline\indent
              \textit{Area $B_{\ell,2}$: $N^2/2+N/4+3$ boxes.} These areas contain at most one box for each item height larger than $\nicefrac{1}{2}\height(B)$ and smaller than $1/3\height(B)$. These are at most $N/4$ sizes. For the remaining sizes, we know by \cref{lem:ReorderTallVertBox12}, that we create at most $4\countTall\countPseudoTall+3$ boxes in total, since there are at most three immovable items overlapping this area. We know that $\countTall\leq N/4$ since the tall items have heights between $\nicefrac{1}{4}H$ and $\nicefrac{1}{2}\height(B)$, $\countPseudo\leq N/2$ since they have heights smaller than $\nicefrac{1}{2}\height(B)$. Finally, because both the tall and pseudo items in this area have a height of at most $\nicefrac{1}{2}\height(B),$ we have $\countPseudoTall\leq N/2$ as well. This leaves us with at most $4N/4 \cdot N/2+N/4+3=N^2/2+N/4+3$ boxes for the tall items in these areas.\newline\indent
               \textit{Area $B_{\ell,3}$: $0$ boxes.} The only tall items inside these areas are $\ell_b$ and $r_b$ respectively. As they overlap the areas $B_{\ell,2}$, they were counted already. \newline\indent
                \textit{Area $B_{\ell,4}$: $N/4$ boxes.} These areas contain only tall items with heights between $\nicefrac{1}{2}\height(B)$ and $\nicefrac{1}{2}\height(B)+\nicefrac{1}{4}H$. We simply create one box for each of these sizes, leaving us with $\nicefrac{1}{4}N$ boxes. \newline\indent
                 \textit{Area $B_{\ell,5}$: $N/4$ boxes.} Similarly to the area $B_{\ell,4}$, these areas only contain tall items with heights between $\nicefrac{1}{2}\height(B)$ and $\nicefrac{1}{2}\height(B)+\nicefrac{1}{4}H$, again leaving us with $\nicefrac{1}{4}N$ boxes.\newline\indent
                  \textit{Area $B_{\ell,6}$: $N^2/2+1$ boxes.} In these areas, there are only tall and pseudo items with height less than $\nicefrac{1}{2}\height(B).$ As we solve this area with \cref{lem:ReorderTallVertBox12}, we create at most $N^2/2$ boxes for tall items in each of the areas. Further, there might be an overlapping tall item, requiring another single box, resulting in at most $N^2/2+1$ total boxes.\newline\indent
                \textit{Area $B_{\ell,7}$: $N/4$ boxes.} These areas contain $\ell_m$ and $r_m$ respectively. Above these items we create at most $N/4$ boxes for tall items each, since the tall items have heights between $\nicefrac{1}{4}H$ and $\nicefrac{1}{2}\height(B).$ The box for the item overlapping $L_3$ is already counted.\newline\indent
                 \textit{Area $B_{8}$: $N/4$ boxes.} In this area the tall items have a height of at most $\nicefrac{1}{2}\height(B)$ and we create one box per item size. Thus, this leaves us with a total of $N/4$ boxes for $B_8$. Any items overlapping the lines $L_2$ or $R_2$ are already counted in the area $B_{\ell,7}.$\newline\indent
                \textit{Area $B_{\ell,9}$: $2N/4$ boxes.} Next, we consider these areas. Here, all items have a height of at most $\nicefrac{1}{2}\height(B)$ and we create a box for each size at both the bottom and top of the areas. This leaves us with $2N/4$ boxes for each of the areas.\newline\indent
                 \textit{Area $B_{10}$: $N/4$ boxes.} This area contains pseudo items with height $\height(B)$. Here, we create at most one box for each item with height larger than $\height(B)-\nicefrac{1}{4}H$ resulting in at most $N/4$ boxes.\newline\indent
                 After doubling the count of boxes for all duplicate areas, this leaves us with at most
                 \begin{align*}
                 	&2(N/4+N^2/2+N/4+3+N/4+N^2/2+1+N/4+2N/4)\\
                 	&+2N/4\\
                 	&=2N^2+14N/4+8
                 \end{align*}  boxes, proving the claim.
         \end{proof}
         Next we consider the boxes for vertical items.
         \begin{claim}
             The number of boxes for vertical items is at most $4N^2+41N/4+5$.
         \end{claim}
         \begin{proof}
             Similarly to the last claim, we prove this by considering the areas iteratively. In the following, whenever we name a area, we only refer to the left area, but the arguments hold for both the left and right areas, as they are analogous.\newline\indent
             \textit{Area $B_{\ell,1}$: $N/2$ boxes.} In these areas, items touching the bottom have a height of at most $\nicefrac{1}{4}H$. The same holds true for items touching the top, as they are vertical. 
             This yields at most $2N/4=N/2$ boxes in each of these areas.\newline\indent
             \textit{Area $B_{\ell,2}$: $N^2+2$ boxes.} In these areas, the pseudo items touching the bottom have sizes between $\nicefrac{1}{4}H$ and $\nicefrac{1}{2}\height(B)$ and the items touching the top have sizes up to $\nicefrac{1}{2}\height(B)$. Using \cref{lem:ReorderTallVertBox12}, we generate at most $4\countPseudo\countPseudoTall\leq 4N/2\cdot N/2 =N^2$ boxes plus the two boxes for extending the immovable items in each area. Thus, we create at most $N^2+2$ boxes in each area.\newline\indent
             \textit{Area $B_{\ell,3}$: $N/4$ boxes.} These areas contain pseudo items with heights up to $\nicefrac{1}{4}H$. For each of these sizes, we generate at most one box, resulting in $N/4$ boxes for each of these areas.\newline\indent
             \textit{Area $B_{\ell,4}$: $N/2$ boxes.} Here, below the items with height larger than $\nicefrac{1}{2}\height(B),$ we packed items with height at most $\nicefrac{1}{4}H$. We have two blocks of these items, one each at $\ell$ and $r$. Thus, we create $N/4$ boxes in these areas, resulting in a total of $N/2$ boxes for these areas.\newline\indent
             \textit{Area $B_{5}$: $N/4$ boxes.} This area does not contain pseudo items above its tall items with height between $\nicefrac{1}{2}\height(B)$ and $\nicefrac{1}{2}\height(B)-\nicefrac{1}{4}H$, as they were shifted up so their lower border is at $\height(B)-\nicefrac{1}{4}H$. Therefore, this area only contains pseudo items with height between $\nicefrac{1}{2}\height(B)$ and $\nicefrac{1}{2}\height(B)$, and we create at most one box for each size, yielding $N/4$ boxes.\newline\indent
             \textit{Area $B_{\ell,6}$: $N^2$ boxes.} These areas contain pseudo items with heights between $\nicefrac{1}{4}H$ and $\nicefrac{1}{2}\height(B)$ at the top and items with heights up to $\nicefrac{1}{2}\height(B)$ at the bottom. Therefore, by \cref{lem:ReorderTallVertBox12}, we generate at most $N^2$ boxes, analogously to the area $B_{\ell,2}$. We do not create the additional pseudo item here, as $\ell_{t,\ell}$ already touches the top of the area. Thus, the number of boxes is at most $N^2$.\newline\indent
             \textit{Area $B_{\ell,7}$: $N$ boxes.} In these areas, we have at most three areas for pseudo items touching $\ell_m$ with their lower border. These items have a height of at most $\nicefrac{1}{4}H$. Therefore, we create at most $3\cdot \nicefrac{1}{4}N$ boxes for these items. Finally, the items touching the $\height(B)+\nicefrac{1}{4}H$ with their top have a height between $\nicefrac{1}{4}\height(B)$ and $\nicefrac{1}{2}\height(B)$. Again, we generate one box for each height, resulting in at most $\nicefrac{1}{4}N$ boxes. This leaves us with at most $N$ boxes in each area.\newline\indent
             \textit{Area $B_{8}$: $N/2$ boxes.} In this area, we have (pseudo) items with height up to $\nicefrac{1}{4}H$ touching the bottom and pseudo items with sizes between $\nicefrac{1}{4}H$ and $\nicefrac{1}{2}\height(B)$ touching the top. For each of these sizes, we generate at most one box, yielding $N/2$ boxes.\newline\indent
             \textit{Area $B_{\ell,9}$: $N$ boxes.} In these areas, the pseudo items have height between $\nicefrac{1}{4}H$ and $\nicefrac{3}{4}\cdot \height(B)$. We create at most one box for each of these sizes, yielding $N/2$ boxes. At the top of the area, we have items with height at most $\nicefrac{1}{2}\height(B)$, again generating one box per size. This results in another $N/2$ boxes, totalling at $N$ boxes for the area.\newline\indent
             \textit{Area $B_{10}$: $N/4+1$ boxes.} Finally, we consider the items with height larger than $3\height(B)/4$ in the area $B_{10}$. There can be pseudo items with height at most $\nicefrac{1}{4}H$ above these items. Since we create at most one box per size, we create at most $N/4$ boxes. Further, there may be exactly one box for pseudo items with height exactly $\height(B)$.\newline\indent
             Summing the generated boxes up and considering duplicated areas, we create at most
             \begin{align*}
             	& 2(\nicefrac{N}{2}+N^2+2+\nicefrac{N}{4}+N/2+N^2+N+\nicefrac{N}{2}+N+\nicefrac{N}{4})\\
             	&+N/4+1\\
             	&=4N^2+41N/4+5
             \end{align*}  boxes for vertical items.\newline\indent         
         \end{proof}
         This concludes the proof of \cref{lem:ReorderTallVertBox34}. For an illustration of the resulting box, see \cref{fig:GenralReordering5}.
         \end{proof}
     	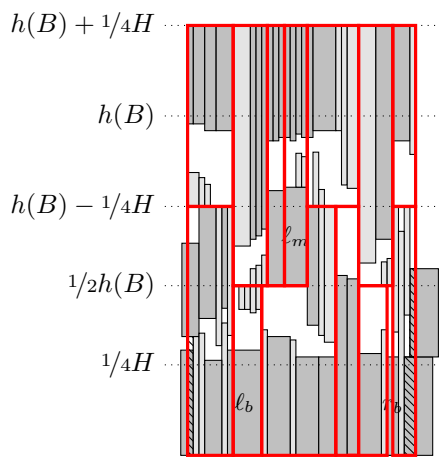
\begin{figure*}[ht]
     	\centering
     	\begin{tikzpicture}
     	\pgfmathsetmacro{\w}{3}
     	\pgfmathsetmacro{\h}{4.8}
     	\pgfmathsetmacro{\hprime}{4.5}
     	
     	\draw (0*\w,0) rectangle (1*\w,\hprime +\h/4);
     	
     	\foreach \x/\y/\xx/\yy in {
     		-0.03 /0.00/0.025/0.31,
     		-0.025/0.35/0.05/0.625,
     		1.075 /0.00/0.95/0.29,
     		1.1   /0.29/0.975/0.55
     	}
     	{
     		\drawTallItem{\x*\w}{\y*\hprime}{\xx*\w}{\yy*\hprime};
     	}
     	
     	\foreach \x/\y/\xx/\yy in {
     		0    /0   /0.025/0.35,
     		0.950/0   /1    /0.29,
     		0.975/0.29/1    /0.55
     	}{
     		\draw[pattern = north west lines] (\x*\w,\y*\hprime) rectangle (\xx*\w,\yy*\hprime);
     	}
     	
     	\foreach \x/\y/\xx/\yy in {
     		0.05 /0.67/0.125/1,
     		0.15 /0.70/0.175/1,
     		0.525/0.65/0.575/1
     	}
     	{
     		\drawTallItem{\x*\w}{\y*\hprime-\h/4}{\xx*\w}{\yy*\hprime-\h/4};
     	}
     	
     	\foreach \x/\y/\xx/\yy/\z in {
     		0.35 /1/0.425  /1.28/,
     		0.425/1/0.525  /1.29/$\ell_m$
     	}
     	{
     		\drawTallItem[\small \z]{\x*\w}{\y*\hprime-0.5*\hprime} {\xx*\w}{\yy*\hprime-0.5*\hprime};
     	}
     	\foreach \x/\y/\xx/\yy in {
     		0.0   /0.71 /0.025/1,
     		0.025 /0.71 /0.075/1,
     		0.075 /0.69 /0.125/1,
     		0.125 /0.69 /0.2  /1,
     		0.275 /0.37 /0.3  /1,
     		0.3   /0.38 /0.325/1,
     		0.325 /0.40 /0.35 /1,
     		0.35  /0.66 /0.375/1,
     		0.375 /0.66 /0.4  /1,
     		0.4   /0.67 /0.45 /1,
     		0.45  /0.66 /0.475/1,
     		0.475 /0.68 /0.5  /1,
     		0.5   /0.67 /0.55 /1,
     		0.55  /0.69 /0.65 /1,
     		0.825 /0.37 /0.9  /1,
     		0.9   /0.66 /0.975/1
     	}
     	{
     		\drawTallItem{\x*\w}{\y*\hprime+\h/4}{\xx*\w}{\yy*\hprime+\h/4};
     	}
     	
     	\foreach \x/\y/\xx/\yy in {
     		0.2  /0.35/0.275/1,
     		0.65 /0.69/0.675/1,
     		0.675/0.53/0.7/1,
     		0.7  /0.52/0.75 /1,
     		0.75 /0.30/0.825/1,
     		0.975/0.62/1    /1
     	}
     	{
     		\drawVerticalItem{\x*\w}{\y*\hprime+\h/4}{\xx*\w}{\yy*\hprime+\h/4};
     	}
     	
     	\foreach \x/\y/\xx/\yy in {
     		0.000/1/0.05/1.10,
     		0.05 /1/0.075/1.085,
     		0.075/1/0.1 /1.065,
     		0.55 /1/0.575  /1.16,
     		0.575/1/0.6    /1.05
     	}
     	{
     		\drawVerticalItem{\x*\w}{\y*\hprime-\h/4}{\xx*\w}{\yy*\hprime-\h/4};
     	}
     	
     	\foreach \x/\y/\xx/\yy in {
     		0.225/0.93/0.25 /1,
     		0.25 /0.93/0.275/1,
     		0.275/0.92/0.3  /1,
     		0.3  /0.93/0.325/1,
     		0.475 /1.29/0.5/1.39,
     		0.5/1.29/0.525  /1.38
     	}
     	{
     		\drawVerticalItem{\x*\w}{\y*\hprime-0.5*\hprime}{\xx*\w}{\yy*\hprime-0.5*\hprime};
     	}

     	\foreach \x/\y/\xx/\yy in {
     		0.125/0.59/0.15 /1,
     		0.175/0.62/0.2  /1,
     		0.575/0.62/0.6  /1,
     		0.6  /0.60/0.65 /1,
     		0.9  /0.63/0.925/1,
     		0.95 /0.68/0.975/1
     	}
     	{
     		\drawVerticalItem{\x*\w}{\y*\hprime-\h/4}{\xx*\w}{\yy*\hprime-\h/4};
     	}
     	
     	\foreach \x/\y/\xx/\yy in {
     		0.275/1/0.300/1.05,
     		0.300/1/0.325/1.06,
     		0.325/1/0.350/1.09,
     		0.85 /1/0.875/1.07,
     		0.875/1/0.9  /1.08
     	}
     	{
     		\drawVerticalItem{\x*\w}{\y*\hprime-\hprime/2}{\xx*\w}{\yy*\hprime-\hprime/2};
     	}
     	
     	\foreach \x/\y/\xx/\yy in {
     		0.025/0.0 /0.05 /0.35,
     		0.05 /0.0 /0.075/0.36,
     		0.15 /0.0 /0.175/0.39,
     		0.325/0.0 /0.35 /0.37,
     		0.45 /0.0 /0.475/0.34,
     		0.85 /0.0 /0.875/0.37,
     		0.925/0.0 /0.95 /0.66
     	}
     	{
     		\drawVerticalItem{\x*\w}{\y*\hprime}{\xx*\w}{\yy*\hprime};
     	}

     	\foreach \x/\y/\xx/\yy/\z in {
     		0.075 /0.0 /0.15 /0.28/,
     		0.175 /0.0 /0.325/0.31/$\ell_b$,
     		0.35  /0.0 /0.45 /0.35/,
     		0.475 /0.0 /0.575/0.29/,
     		0.575 /0.0 /0.65 /0.29/,
     		0.65  /0.0 /0.7  /0.53/,
     		0.7   /0.0 /0.75 /0.52/,
     		0.75  /0.0 /0.85 /0.30/,
     		0.875 /0.0 /0.925/0.29/$r_b$
     	}
     	{
     		\drawTallItem[\small \z]{\x*\w}{\y*\hprime}{\xx*\w}{\yy*\hprime};
     	}

     	\foreach \y/\z in {
     		0.5*\hprime 		/ $\nicefrac{1}{2}\height(B)$,
     		0.25*\h     		/ $\nicefrac{1}{4}H$,
     		\hprime - 0.25*\h	/ $\height(B) - \nicefrac{1}{4}H$,
     		\hprime				/ $\height(B)$,
     		\hprime + 0.25*\h	/ $\height(B) + \nicefrac{1}{4}H$
     	}{
     		\draw[dotted] (-0.1*\w,\y)node[left]{\z} -- (1.1*\w,\y) ;
     	}
     	
     	\foreach \x/\y/\xx/\yy in {
     		0.0*\w/\hprime-0.25*\h/0.2*\w/\hprime + 0.25*\h,  
     		0.9*\w/\hprime-0.25*\h/1.0*\w/\hprime +0.25*\h,   
     		0.0*\w/\hprime-0.25*\h/0.2*\w/0,                  
     		0.9*\w/\hprime-0.25*\h/1.0*\w/0,                  
     		0.2*\w/0/0.325*\w/0.5*\hprime,                    
     		0.2*\w /0.5*\hprime/0.35*\w/\hprime + 0.25*\h,    
     		0.75*\w/0.5*\hprime/0.9*\w/\hprime + 0.25*\h,     
     		0.65*\w/0.0*\hprime/0.75 *\w/\hprime -0.25*\h,    
     		0.35*\w/0.5*\hprime/0.425*\w/\hprime + 0.25*\h,   
     		0.425*\w/0.5*\hprime/0.525*\w/\hprime +0.25*\h,   
     		0.525*\w/\hprime -0.25*\h/0.75*\w/\hprime +0.25*\h     
     	}{
     		\draw[very thick, red] (\x,\y) rectangle (\xx,\yy);
     	}
     	
     	\draw[very thick,red] (0.325*\w,0*\hprime) -- (0.325*\w,\hprime/2) -- (0.525*\w,\hprime/2) -- (0.525*\w,\hprime -\h/4) -- (0.65*\w,\hprime -\h/4) -- (0.65*\w,0*\hprime) -- (0.325*\w,0*\hprime);
     	
     	\draw[very thick,red] (0.75*\w,0*\hprime) -- (0.75*\w,0.58*\hprime -\h/4) -- (0.75*\w,0.58*\hprime -\h/4)--(0.75*\w,\hprime/2) -- (0.875*\w,\hprime /2) -- (0.875*\w,0) --(0.7*\w,0);
     	
     	\end{tikzpicture}
     	\caption{The packing that results from the procedure given in \cref{lem:ReorderTallVertBox34}. All sub-areas are highlighted individually and sorted as discussed above.}
     	\label{fig:GenralReordering5}
     \end{figure*}
\paragraph*{Step 5} 
Having shown that we can separate the optimal packing into $\Oh_{\eps}(1)$ many boxes that each contain only a single item or item type, we now show that we can fill these boxes with their respective items.

The boxes for horizontal and vertical items are both filled using configuration~IPs. Vertical items might have been treated as pseudo items and have been separated. During their placement, we ensure that they are placed integrally again. The configuration~IPs for both types function in a similar manner and ensure that all items are placed integrally while not exceeding the border of the strips.
First, consider the placement of vertical items.
\begin{lem}
	Let $\heightVertical$ be the set of different heights of vertical items and $\mu\stripWidth$ the maximal width of a vertical item. Furthermore, let $\boxesPseudo$ be the set of boxes containing all separated vertical items and only them.
	There exists a non separated placement of vertical items into the boxes~$\boxesPseudo$ and at most $7(|\heightVertical|+|\boxesPseudo|)$ additional boxes $\boxesPseudo'$, each of height at most $\nicefrac{1}{4}\stripHeight$ and width $\mu\stripWidth,$, such that the boxes $\boxesPseudo\cup \boxesPseudo'$ are partitioned into at most $\bigO((|\heightVertical|+|\boxesPseudo|)/\delta)$ sub-boxes $\boxesVert,$ containing only vertical items of the same height and at most $\bigO(|\heightVertical|+|\boxesPseudo|)$ empty boxes $\boxesSmallVert$ with total area $a(\boxesSmallVert)\geq a(\boxesPseudo)-a(\itemsVert)$.
	\label{lem:PlacementVertSep}
\end{lem}
\begin{proof}
	This procedure is an evolution of the ideas in \cite{StripPacking54}.
	We prove this lemma via a configuration LP. We define a configuration as a multiset of items that can be placed atop each other without exceeding the boundaries of a given box $B$. We then place these configurations of items inside their assigned boxes. Therefore, let $C=\{a_\height:\height|\height\in \heightVertical\}$. The height of a configurations is the sum of all item-heights inside the configuration, thus $\height(C):=\sum_{\height\in \heightVertical}\height\cdot a_\height$. Let $\mathcal{C}_B$ be the set of configurations with height of at most $\height(B)$, i.e.\ configurations that fit inside this box. We use $X_{C,B}$ to indicate how many configurations $C$ are assigned to a box $B$ and define $w(C)$ as the width of a configuration. Finally, we define $w_\height$ to be the total width of all vertical items with height $\height$ for all $\height\in \heightVertical.$\newline\indent
	The configuration LP is defined as follows:
	\begin{align*}
	\sum_{C\in \mathcal{C}_B} X_{C,B} w(C) &= w(B)   &\forall B\in \boxesPseudo \\
	\sum_{B\in \boxesVert}\sum_{C\in \mathcal{C}_B} X_{C,B}\cdot a_{\height,C} &= w_\height &\forall \height\in \heightVertical \\
	X_{C,B}&\geq 0 &\forall B\in \boxesVert, C\in \mathcal{C}_B
	\end{align*}
	The first equation ensures that all items inside a configuration fit inside their assigned boxes regarding their width. We later fill these configurations greedily and handle the resulting overlap through the additional boxes. The second equation ensures that we place all items of a given height across all our configurations. Doing this for every height ensures that all items are placed after solving the configuration LP. Finally, the last equation states that all configurations can only contain positive entries to ensure correctness, as there cannot be negative items that free up space.\newline\indent
	As the configuration LP is a linear program in itself, it has a basic solution with at most $|\heightVertical|+|\boxesPseudo|$ non-zero components. This is because it has that many conditions.\newline\indent
	Given such a basic solution, we place the selected configurations inside their corresponding boxes. Afterwards, we place the items into the configurations, such that the last item overlaps the configuration border. Each configuration has a height of at most $H$ since the boxes~$\boxesPseudo$ have at most that height.\newline\indent
 	We partition the set of overlapping items in each configuration into seven boxes with height $\nicefrac{1}{4}H$ and width $\mu\stripWidth$ in the following way: First, we stack the items in four boxes one by one atop each other such that the last item overlaps the box on top. Since the total height of the items is at most $H$, there are at most three overlapping items. Each item is placed into its own box. These boxes are called $\boxesPseudo'$. In total, we generate at most $7(|\heightVertical|+|\boxesPseudo|)$ boxes of width $\mu\stripWidth$. The items can be placed wholly inside these boxes as they each have a width of at most $\mu\stripWidth$, since they are vertical items.\newline\indent
	It is worth noting that the configuration width as defined by the basic solution of the configuration LP might not be integral. Since we removed all overlapping items, however, we only need an integral width, and can therefore reduce the configuration width to the next smaller integer. This may result in empty configurations inside the strip. These empty configuration have  at least a width of the sum of all non integral fractions we removed from the configurations inside their box. Since their boxes had an integral width and all other configurations have an integral width as well, these empty configurations have integral widths themselves.\newline\indent
	Since configurations have a height of at most $\stripHeight$ and each item has a height of at least $\delta\opt$, each configuration contains at most $\stripHeight/(\delta\opt)\in \bigO(1/\delta)$ items. Therefore, the set of boxes $\boxesPseudo\cup\boxesPseudo'$ is divided into at most $2(|\heightVertical|+|\boxesPseudo|)H/(\delta\opt)\in\bigO((|\heightVertical|+|\boxesPseudo|)/\delta)$ sub-boxes containing only vertical items of the same height.\newline\indent
	Consider a configuration $C\in \mathcal{C}_B$ which has a non-zero entry $X_{C,B}$ in the considered solution. Above this configuration, there is a free area of height $\height(B)-\height(C)$ and width $X_{C,B}$ inside the box $B$, see \cref{fig:PlaceVertItems}. Furthermore, in each box there might be a new empty configuration which generates an empty box as well. Let $\boxesSmallVert$ be the set of these boxes. There are at most $\bigO(|\heightVertical|+|\boxesPseudo|)$, one for each configuration and one extra for each box. Since we ensured that the configurations use exactly the width of the vertical items, the total area of these empty boxes must be at least $a(\boxesSmallVert)\geq a(\boxesPseudo)-a(\itemsVert)$.                 
\end{proof}
	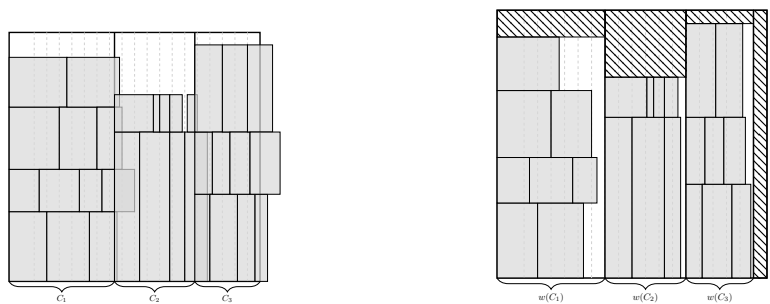
\begin{figure*}[t]
		\centering
		\begin{subfigure}{0.45\textwidth}
			\centering
			\resizebox{0.58\textwidth}{!}{
				\begin{tikzpicture}
				\pgfmathsetmacro{\w}{10}
				\pgfmathsetmacro{\h}{10}
				\pgfmathsetmacro{\cTwo}{0.42}
				\pgfmathsetmacro{\cThree}{0.74}
				
				\draw [ultra thick, color=black] (0,0) --(\w,0) -- (\w,\h) -- (0,\h) -- cycle;
				\draw [dashed, color= tbBgOdd] (0.1*\w,0) -- (0.1*\w,\h);
				\draw [dashed, color= tbBgOdd] (0.2*\w,0) -- (0.2*\w,\h);
				\draw [dashed, color= tbBgOdd] (0.3*\w,0) -- (0.3*\w,\h);
				\draw [dashed, color= tbBgOdd] (0.4*\w,0) -- (0.4*\w,\h);
				\draw [dashed, color= tbBgOdd] (0.5*\w,0) -- (0.5*\w,\h);
				\draw [dashed, color= tbBgOdd] (0.6*\w,0) -- (0.6*\w,\h);
				\draw [dashed, color= tbBgOdd] (0.7*\w,0) -- (0.7*\w,\h);
				\draw [dashed, color= tbBgOdd] (0.8*\w,0) -- (0.8*\w,\h);
				\draw [dashed, color= tbBgOdd] (0.9*\w,0) -- (0.9*\w,\h);
				
				\draw [dashed, color= tbBgOdd] (0.15*\w,0) -- (0.15*\w,\h);
				\draw [dashed, color= tbBgOdd] (0.25*\w,0) -- (0.25*\w,\h);
				\draw [dashed, color= tbBgOdd] (0.35*\w,0) -- (0.35*\w,\h);
				\draw [dashed, color= tbBgOdd] (0.45*\w,0) -- (0.45*\w,\h);
				\draw [dashed, color= tbBgOdd] (0.55*\w,0) -- (0.55*\w,\h);
				\draw [dashed, color= tbBgOdd] (0.65*\w,0) -- (0.65*\w,\h);
				\draw [dashed, color= tbBgOdd] (0.75*\w,0) -- (0.75*\w,\h);
				\draw [dashed, color= tbBgOdd] (0.85*\w,0) -- (0.85*\w,\h);
				\draw [dashed, color= tbBgOdd] (0.95*\w,0) -- (0.95*\w,\h);
				
				\draw [ultra thick, color=black] (\cTwo*\w,0) --(\cTwo*\w,\h); 
				\draw [ultra thick, color=black] (\cThree*\w,0) --(\cThree*\w,\h);
				\foreach \x/\y/\xx/\yy/\z in {
					0.0 /0.0 /0.15 /0.28/,
					0.0 /0.28 /0.12/0.45/,
					0.0 /0.45 /0.2 /0.7/,
					0.0 /0.7 /0.23 /0.9/,
					0.15 /0.0 /0.32 /0.28/,
					0.12 /0.28 /0.28/0.45/,
					0.2 /0.45 /0.35 /0.7/,
					0.23 /0.7 /0.44 /0.9/,
					0.32 /0.0 /0.43 /0.28/,
					0.28 /0.28 /0.37/0.45/,
					0.35 /0.45 /0.45 /0.7/,
					0.37 /0.28 /0.5/0.45/,
					\cTwo  / 0.0 /0.52 /0.6/,
					\cTwo / 0.6 /0.575/0.75/,
					0.52  / 0.0 /0.64 /0.6/,
					0.575 / 0.6 /0.6/0.75/,
					0.64  / 0.0 /0.7 /0.6/,
					0.6 / 0.6 /0.64/0.75/,		
					0.7  / 0.0 /0.79 /0.6/,
					0.64 / 0.6 /0.69/0.75/,	
					0.71 / 0.6 /0.75/0.75/,
					\cThree  /0.0 /0.8  /0.35/,
					\cThree   /0.35 /0.81 /0.6/,
					\cThree  /0.6 /0.85 /0.95/,
					0.8  /0.0 /0.91  /0.35/,
					0.81   /0.35 /0.88 /0.6/,
					0.85  /0.6 /0.95 /0.95/,
					0.91  /0.0 /0.98  /0.35/,
					0.88   /0.35 /0.96 /0.6/,
					0.95  /0.6 /1.05 /0.95/,
					0.98  /0.0 /1.03  /0.35/,
					0.96   /0.35 /1.08 /0.6/			
				}
				{
					\drawVerticalItem[\small \z]{\x*\w}{\y*\h}{\xx*\w}{\yy*\h};
				}
				\draw[thick,decorate,decoration={brace,amplitude=12pt,mirror}] (0,0) -- (\cTwo*\w,0)node[midway, below=12 pt] {$C_1$};
				\draw[thick,decorate,decoration={brace,amplitude=12pt,mirror}] (\cTwo*\w,0) -- (\cThree*\w,0)node[midway, below=12 pt] {$C_2$};
				\draw[thick,decorate,decoration={brace,amplitude=12pt,mirror}] (\cThree*\w,0) -- (\w,0)node[midway, below=12 pt] {$C_3$};
				\end{tikzpicture}
			}
		\end{subfigure}
		\begin{subfigure}{0.45\textwidth}
			\centering
		\resizebox{0.58\textwidth}{!}{
			\begin{tikzpicture}
			\pgfmathsetmacro{\w}{10}
			\pgfmathsetmacro{\h}{10}
			\pgfmathsetmacro{\cTwo}{0.42}
			\pgfmathsetmacro{\cTwoPrime}{0.4}
			\pgfmathsetmacro{\cThree}{0.74}
			\pgfmathsetmacro{\cThreePrime}{0.7}
			
			\draw [ultra thick, color=black] (0,0) --(\w,0) -- (\w,\h) -- (0,\h) -- cycle;
			\draw [dashed, color= tbBgOdd] (0.1*\w,0) -- (0.1*\w,\h);
			\draw [dashed, color= tbBgOdd] (0.2*\w,0) -- (0.2*\w,\h);
			\draw [dashed, color= tbBgOdd] (0.3*\w,0) -- (0.3*\w,\h);
			\draw [dashed, color= tbBgOdd] (0.4*\w,0) -- (0.4*\w,\h);
			\draw [dashed, color= tbBgOdd] (0.5*\w,0) -- (0.5*\w,\h);
			\draw [dashed, color= tbBgOdd] (0.6*\w,0) -- (0.6*\w,\h);
			\draw [dashed, color= tbBgOdd] (0.7*\w,0) -- (0.7*\w,\h);
			\draw [dashed, color= tbBgOdd] (0.8*\w,0) -- (0.8*\w,\h);
			\draw [dashed, color= tbBgOdd] (0.9*\w,0) -- (0.9*\w,\h);
			
			\draw [dashed, color= tbBgOdd] (0.15*\w,0) -- (0.15*\w,\h);
			\draw [dashed, color= tbBgOdd] (0.25*\w,0) -- (0.25*\w,\h);
			\draw [dashed, color= tbBgOdd] (0.35*\w,0) -- (0.35*\w,\h);
			\draw [dashed, color= tbBgOdd] (0.45*\w,0) -- (0.45*\w,\h);
			\draw [dashed, color= tbBgOdd] (0.55*\w,0) -- (0.55*\w,\h);
			\draw [dashed, color= tbBgOdd] (0.65*\w,0) -- (0.65*\w,\h);
			\draw [dashed, color= tbBgOdd] (0.75*\w,0) -- (0.75*\w,\h);
			\draw [dashed, color= tbBgOdd] (0.85*\w,0) -- (0.85*\w,\h);
			\draw [dashed, color= tbBgOdd] (0.95*\w,0) -- (0.95*\w,\h);
			
			\draw [ultra thick, color=black] (\cTwoPrime*\w,0) --(\cTwoPrime*\w,\h); 
			\draw [ultra thick, color=black] (\cThreePrime*\w,0) --(\cThreePrime*\w,\h);
			\draw [ultra thick, color=black] (0.95*\w,0) --(0.95*\w,\h);
			\foreach \x/\y/\xx/\yy/\z in {
				0.0 /0.0 /0.15 /0.28/,
				0.0 /0.28 /0.12/0.45/,
				0.0 /0.45 /0.2 /0.7/,
				0.0 /0.7 /0.23 /0.9/,
				0.15 /0.0 /0.32 /0.28/,
				0.12 /0.28 /0.28/0.45/,
				0.2 /0.45 /0.35 /0.7/,
				0.28 /0.28 /0.37/0.45/,
				\cTwoPrime  / 0.0 /0.50 /0.6/,
				\cTwoPrime / 0.6 /0.555/0.75/,
				0.5  / 0.0 /0.62 /0.6/,
				0.555 / 0.6 /0.58/0.75/,
				0.62  / 0.0 /0.68 /0.6/,
				0.58 / 0.6 /0.62/0.75/,		
				0.62 / 0.6 /0.67/0.75/,	
				\cThreePrime  /0.0 /0.76  /0.35/,
				\cThreePrime   /0.35 /0.77 /0.6/,
				\cThreePrime  /0.6 /0.81 /0.95/,
				0.76  /0.0 /0.87  /0.35/,
				0.77   /0.35 /0.84 /0.6/,
				0.81  /0.6 /0.91 /0.95/,
				0.87  /0.0 /0.94  /0.35/,
				0.84   /0.35 /0.92 /0.6/
			}
			{
				\drawVerticalItem[\small \z]{\x*\w}{\y*\h}{\xx*\w}{\yy*\h};
			}
			\draw[thick,decorate,decoration={brace,amplitude=12pt,mirror}] (0,0) -- (\cTwoPrime*\w,0)node[midway, below=12 pt] {$w(C_1)$};
			\draw[thick,decorate,decoration={brace,amplitude=12pt,mirror}] (\cTwoPrime*\w,0) -- (\cThreePrime*\w,0)node[midway, below=12 pt] {$w(C_2)$};
			\draw[thick,decorate,decoration={brace,amplitude=12pt,mirror}] (\cThreePrime*\w,0) -- (0.95*\w,0)node[midway, below=12 pt] {$w(C_3)$};
			
			\draw[pattern=north west lines] (0,0.9*\h) rectangle (\cTwoPrime*\w,\h);
			\draw[pattern=north west lines] (\cTwoPrime*\w,0.75*\h) rectangle (\cThreePrime*\w,\h);
			\draw[pattern=north west lines] (\cThreePrime*\w,0.95*\h) rectangle (0.95*\w,\h);
			\draw[pattern=north west lines] (0.95*\w,0) rectangle (\w,\h);
			\end{tikzpicture}
		}
	\end{subfigure}
\caption{An illustration of a set of configurations assigned to a box~$B$. The width of a configuration is shown by the thick black lines. The configurations are filled greedily until each height in it exceeds the width of the configuration. After removing these excess items, we can reduce the width of configurations to the next integral value. The area not used by items is shown as hatched and can later be used to pack small items.}
\label{fig:PlaceVertItems}
	\end{figure*}
We place the horizontal items using a similar configuration~IP.
\begin{lem}
	There is an algorithm with running~time \\$(\log(1/\delta)/\eps)^{\bigO(1/\eps\delta^3)}$ that places the horizontal items into the boxes $\boxesHor$ and an extra box $B_H$ of height at most $\eps^9\opt$ and width $\stripWidth$.
	Furthermore, the algorithm creates at most $\bigO(1/\eps\delta^2)$ empty boxes $\boxesHorSmall$ with total area $a(\boxesHorSmall)=a(\boxesHor)-a(\Tilde{H})$, where $\Tilde{H}$ is the set of horizontal items that overlap the top or bottom box borders. 
	\label{lem:PlacementHorItems}
\end{lem}
\begin{proof}
	This algorithm is an evolution of similar ideas in \cite{StripPacking54}.
	We begin by rounding the input, in this case the horizontal items. To achieve this, we stack horizontal items on top of one another ordered by descending width, i.e.\ widest at the bottom. This stack has a height of at most $\opt/\delta$, because each horizontal item has a width of at least $\delta\stripWidth$ and the total packing area has size $\opt\cdot\stripWidth$. We group the items in the stack to at most $1/\eps\delta^2$ groups, each of height $\eps\delta^2\opt/\delta=\eps\delta\opt$ and round the items in the groups to the widest width occurring in this group. This step reduces the number of different sizes to at most $1/\delta\eps^2$, i.e.\ one per generated group. The rounded horizontal items can be placed fractionally into the non-rounded items of the group containing the next larger items. The group containing the widest rounded items has to be placed atop the final packing. As this is just one group, this extra box has a height of at most $\eps\delta\opt$. We define an extra box of width $\stripWidth$ and that height to pack these items into. To keep notations simple we assume that $\boxesHor$ includes this box as well. \newline\indent
	Similar to our procedure to place the separated vertical items in \cref{lem:PlacementVertSep}, we place the rounded horizontal items into their boxes through the use of a configuration LP. Here, a configuration is a set of items that fit next to each other inside the boxes. Thus, a configuration $C$ is a multiset of the form $\{a_w:w|w\in \widthsHor\}$ where  $\widthsHor$ is the set of all widths of horizontal items. The width of such a configuration is consequently defined as $w(C):=\sum_{w\in \widthsHor}a_ww$, i.e.\ simply lining all items inside a configuration up beside one another. Additionally, $\mathcal{C}_w$ denotes the set of configurations with width at most $w$. We use $X_{C,B}$ to indicate the number of times configuration $C$ is assigned to box $B$. The height of a configuration is denoted by $h(C)$. Finally, $\height(w)$ is defined as the total height of all items with rounded width~$w$.\newline\indent
	The set of configurations $\mathcal{C}_\stripWidth$ is bounded by $\bigO((\log(1/\delta)/\eps)^{1/\delta})$ because the items have a width of at least $\delta\stripWidth$. Therefore, there can be at most $1/\delta$ items inside each configuration. The following configuration LP is solvable since the rounded horizontal items fit fractionally into the boxes $\boxesHor$.
	\begin{align*}
	\sum_{C\in \mathcal{C}_w(B)} X_{C,B} h(C) &= h(B) &\forall B\in \boxesHor\\
	\sum_{B\in \boxesHor} \sum_{C\in \mathcal{C}_{w(B)}} X_{C,B} a_{w,C}&=h(w) &\forall w\in \widthsHor\\
	X_{C,B} &\geq 0 &\forall B\in \boxesHor, C \in \mathcal{C}_{w(B)}
	\end{align*}
	The first constraint ensures that all boxes are filled by configurations to exactly their height. We fill these configurations greedily in the following. The second constraint ensures that all items of every height are placed inside some configuration, i.e.\ that we successfully place all items. Finally, the third constraint ensures the feasibility of the packing by disallowing negative configurations to compensate for space, as such configurations are impossible.\newline\indent
	We can solve this linear program by guessing the at most $|\widthsHor|+|\boxesHor|=\bigO(1/(\eps\delta^2))$ non-zero entries of the basic solution and solve the resulting equality system using the Gauß-Jordan-Elimination. We use the first found solution where no variables are negative. Such a solution can be found in at most $\bigO(|C_\stripWidth|^{|\widthsHor|+|\boxesHor|}\cdot(|\widthsHor|+|\boxesHor|)^3\leq (\log(1/\delta)/\eps)^{\bigO(1/\eps\delta^3)}$ operations since the configuration LP has to be solvable for the correct partition.\newline\indent
	We place each configuration into the corresponding box and place the original horizontal items inside these configurations greedily. We allow the last item to overlap the configuration border. One by one, we place the original items inside an area reserved by the configurations for their rounded equivalents until an item overlaps this area on top. We repeat this process for the next area. Since the total width of these parts is exactly as large as the total width of items with that rounded width, there are enough parts to place all of them. Recall that we ensured this with the second constraint.\newline\indent
	Next, we handle the items that overlap the box borders. We place them on top of the box. Each of these removed items is horizontal, and as such has a height of at most $\mu\opt$. We add at most $\eps^{10}\opt$ to the packing height by shifting the overlapping items to the top of the packing. This is due to a basic solution having at most $\bigO(1/(\eps\delta^2))$ configurations. All items inside a given configuration can be placed next to each other per definition, meaning we can add at most one overlapping item for every configuration. Finally, we selected $\mu$ such that $\mu\leq \delta^2\eps^{11}/k$ holds for a suitably large constant $k\in \mathbb{N}$. Taking these items together with the extra box we require due to rounding the items, the total height is bounded by $\eps^{10}\opt+\delta\eps\opt\leq\eps^9\opt.$\newline\indent
	Similar to the steps taken in \cref{lem:PlacementVertSep} we can reduce the height of each configuration to the next smaller integer since all horizontal items have integral height. This introduces at most one new configuration per box, i.e.\ the empty one. Inside each box~$B$, to the right of every used configuration $C$, there might be some unused space of width~$w(B)-w(C)$ and height~$x_{C,B}$. This area defines one of the empty boxes~$\boxesHorSmall$. We have at most $|\widthsHor|+|\boxesHor|$ configurations and at most $|\boxesHor|$ boxes for horizontal items. We introduce at most one box for every configuration in any box. Therefore, we introduce at most $\bigO(1/\eps\delta^2)$ empty boxes $\boxesHorSmall$. Their total area has to be at least as large as the empty space left inside the boxes, i.e. $a(\boxesHorSmall)=a(\boxesHor)-a(\Tilde{H})$, since the configurations contain exactly the total area of the rounded horizontal items.\newline\indent
	Finally, let us consider the sub-boxes for horizontal items we generate in this step for every box $B$ in $\boxesHor$. Each configuration contains at most $\bigO(1/\delta)$ positions for items, as each horizontal item has a width of at least $\delta\stripWidth$. We generate one sub-box that has rounded width of the item assigned to this position by the configuration. The height of this sub-box is the sum of all heights of the items with that width positioned inside this box. We can combine the height of these sub-boxes because there might be several configurations assigned to $B$ that all contain an item of the same rounded width at the same position.  Finally, we create an additional box for each shifted item. Thus, we introduce at most $\bigO(1/(\eps\delta^3))$ boxes for horizontal items. Each of these boxes only contain items of the same rounded width.
\end{proof}
 
Next, we show that any optimal packing can be rearranged and partitioned into $\Oh_{\eps}(1)$ many boxes that are structured in some way. 
We prove this using the partition given in \cref{lem:boxPartition} into sets of boxes $\boxesLarge,\boxesHor$, and $\boxesTallVert$. 
We then utilize the reordering techniques discussed in step~4 to further partition all boxes $B\in \boxesTallVert$ into boxes for tall items $\boxesTall$ and boxes for vertical items $\boxesVert$. 
We show that we provide an adequate amount of space to pack all vertical items that were separated during the repacking of $\boxesTallVert$. 
\begin{lem}
	By extending the packing height to $(\nicefrac{5}{4}+5\eps)\opt$ each rounded optimal packing can be rearranged and partitioned into $\bigO(1/(\delta^3\eps^5))$ boxes with the following properties:
	\begin{itemize}
		\item there are $|\itemsLarge|+|\itemsMedVert|=\bigO(1/(\delta^2\eps))$ boxes $\boxesLarge$ each containing exactly one item from the set $\itemsLarge\cup \itemsMedVert$ and all items from this set are contained in these boxes,
		\item there are at most $\bigO(1/\delta^2\eps)$ boxes $\boxesHor$ containing all horizontal items with $\boxesHor \cap \boxesLarge =\emptyset$. The horizontal items can overlap the top and bottom border of the boxes, but never the left or right border,
		\item there are at most $\bigO(1/\delta^2\eps^5))$ boxes $\boxesTall$ containing tall items, such that each tall item $t$ is placed inside a box with rounded height $\height(t)$,
		\item there are at most $\bigO(1/\delta^3\eps^5))$ boxes $\boxesVert$ containing vertical items, such that each vertical item $v$ is placed inside a box with rounded height $\height(v)$,
		\item there are at most $\bigO(1/\delta^2\eps^5))$ boxes $\boxesSmall$ for small items, such that the total area of these boxes combined with the total free area inside the horizontal boxes and vertical boxes is at least as large as the total area of small items,
		\item the upper and lower border of each box is placed at a multiple of $\eps\delta\opt.$
	\end{itemize}
	\label{lem:structure}
\end{lem}
\begin{proof}
	This procedure is an adaptation of similar procedures in \cite{StripPacking54}.
	Before we get to the proof itself, we give a short overview of the process to come. \newline\indent
	Recall that the process begins with the partition given in \cref{lem:boxPartition} into the sets of boxes $\boxesLarge,\boxesHor$ and $\boxesTallVert$. We define the height of the strip $\stripHeight:=(1+2\eps)\opt$. Therefore, each tall item has a height larger than $H/4$. This allows us to use our reordering techniques explained in the previous section to separate $\boxesTallVert$ into $\boxesTall$ and $\boxesVert$. We have seen in \cref{lem:PlacementVertSep} that we require some extra boxes to place the vertical items. Thus, the main thing to be proven here is that there is adequate space left in the packing to place these boxes. \newline\indent
	We consider three options to place them. The first choice depends on the widest tall items that intersects the line at $\nicefrac{1}{2}\stripHeight$. We fix their position and attempt to place these extra boxes atop these tall items if their width is large enough. If this is not possible, we know that the tall items intersecting $\nicefrac{1}{2}\stripHeight$ are not very wide. This allows us to attempt to place the extra boxes inside the boxes $B$ of height $\height(B)>\nicefrac{3}{4}\stripHeight$. Again, the feasibility of this placement depends on the total width of these boxes. Finally, if all else fails, we place them on top of boxes with height between $\nicefrac{1}{2}\stripHeight$ and $\nicefrac{3}{4}\stripHeight$.\newline\indent
	The second part of this proof is to show that tall and medium boxes have to start and end at grid lines. Since we already know from our rounding, \cref{lem:Rounding}, that tall items start and end at multiples of $\eps^2\opt$, we can choose these lines as the grid. Setting the start- and endpoints of the tall and medium boxes to these lines leads to a negligible loss in approximation ratio as we show in the following.\newline\indent
	After placing the boxes on the grid as such, we have fulfilled the requirements to utilize the reordering techniques from \cref{lem:ReorderTallVertBox14}, \cref{lem:ReorderTallVertBox12} and \cref{lem:ReorderTallVertBox34} to reorder items inside the boxes $\boxesTallVert$. After this reordering, we analyze the number of containers constructed for vertical items and find a place for the additional containers we need to generate through \cref{lem:PlacementVertSep}. Finally, we show how to place the boxes for horizontal and small items.\newline\indent
	\emph{Step 1: Fixing the position of the widest tall items intersecting $\stripHeight/2$.} We begin by inspecting the $1/(\delta^2\eps)$ widest tall items crossing the horizontal line at $\nicefrac{1}{2}\stripHeight$. We refer to this set of items as $\tallItemsHalf$. Each item in this set defines a new immovable item. Any box that contains this item gets split into three parts. The part to the left, the part to the right and the part containing the item. The parts to the left and right are reordered as any other box. The part that contains it is reordered differently. Since we have a tall item intersecting the horizontal line at $\nicefrac{1}{2}\stripHeight$, the parts above and below this item can have a height of at most $\nicefrac{1}{2}\stripHeight$. We define both of these parts as new boxes. Thus, we can reorder both of these boxes using \cref{lem:ReorderTallVertBox14}. These reorderings both generate at most $\bigO(N)$ sub-boxes for tall and vertical items. This less than the number of boxes generated for a box of height greater than $\nicefrac{3}{4}\stripHeight$. Therefore, we can count this as a single box of height greater than $\nicefrac{3}{4}\stripHeight$ when calculating the number of added boxes and assume it to be at most $2/(\delta^2\eps)$. After this step, the total number of boxes containing both tall and vertical items is bounded by $\bigO(1/(\delta^2\eps))$. Furthermore, the number of vertical lines at box borders through the strip is bounded by the same amount, namely $\bigO(1/(\delta^2\eps))$.\newline\indent
	\emph{Ensuring the alignment of tall and medium boxes to the grid-lines.} In both \cref{lem:ReorderTallVertBox12} and \cref{lem:ReorderTallVertBox34} we assume that each box with height larger than $\nicefrac{1}{2}\stripHeight$ starts and ends at grid points. We generate this property here. We define grid lines as integral multiples of $\eps^2\opt$. Let $B$ be a box with $\height(B)>\nicefrac{1}{2}\stripHeight$. Inspect the horizontal line $\ell$ at the smallest multiple of $\eps^2\opt$ in the box $B$. The distance between $\ell$ and the bottom border clearly has to be smaller than $\eps^2\opt$. Now, we remove all vertical items below $\ell$ and each item cut by $\ell$ and place them inside an extra box. We know that each item with height larger than $\eps\opt$ already starts and ends at a multiple of $\eps^2\opt$ due to our rounding, \cref{lem:Rounding}, all items we placed inside this extra box have a height less than $\eps\opt$. We generate another box in the same manner at the top of $B$, and repeat this procedure for every other box. This leaves us with two boxes of height less than $\eps\opt$ for each box of height greater than $\nicefrac{1}{2}\stripHeight$. Since no two boxes of height greater than $\nicefrac{1}{2}\stripHeight$ can overlap, the total width of these newly generated boxes is at most $2\stripWidth$. Therefore, we can place them by stacking two boxes atop another every time. This generates a new box of height at most $2(\eps+\eps^2)\opt$ and width $\stripWidth$. We place this wide box atop the final packing, at height $\nicefrac{5}{4}\stripHeight+\eps\opt$. Clearly, this does not result in any overlapping.\newline\indent
	\emph{Step 3: Reordering tall and vertical items inside the boxes.} After the necessary condition is fulfilled, we can begin reordering the boxes in $\boxesTallVert$. We do this by applying the either \cref{lem:ReorderTallVertBox14}, \cref{lem:ReorderTallVertBox12} or \cref{lem:ReorderTallVertBox34}, depending on which of those is applicable to the box we are currently reordering. To accurately estimate the number of generated sub-boxes we have to assume that each box in $\boxesTallVert$ is one that generates the largest number of sub-boxes. In our case, these are boxes with height greater than $\nicefrac{3}{4}\stripHeight$, i.e.\ boxes solved with \cref{lem:ReorderTallVertBox34}. In every such box~$B$, we draw a vertical line at the left border of each contained sub-box. If a sub-box for vertical items is intersected by such a line we split the sub-box there. Each of these lines intersects at most three boxes for vertical items, because there can be at most four boxes atop another at any point inside $B$. Therefore, we introduce at most three new boxes per vertical line. Since there are at most $\bigO(1/(\delta^2\eps^5))$ sub-boxes for tall and vertical items, we generate at most $\bigO(1/(\delta^2\eps^5))$ vertical lines as well. 
	Since each of these lines splits a constant number of boxes, the number of boxes for vertical items after splitting is still bounded by $\bigO(1/(\delta^2\eps^5))$.\newline\indent
	The area between two consecutive lines defines a strip where the height of all intersected boxes does not change. Consequently, we have at most $\bigO(1/(\delta^2\eps^5))$ of these strips in total. We define $\countStrips$ as this number of strips.\newline\indent
	\emph{Step 4: Placing the extra boxes for vertical items.} By \cref{lem:PlacementVertSep}, we need at most $\bigO(|\heightVertical|+|\boxesPseudo|)$ additional boxes with height $\nicefrac{1}{4}\stripHeight$ and width $\mu\stripWidth$ to place the vertical items continuously inside the boxes, where $\boxesPseudo$ are the boxes for vertical items created to this point. We call this set of additional boxes $\boxesVertPlace$. We can bound the relevant variables as follows: There are at most $|\boxesPseudo|\in \bigO(1/(\delta^2\eps^5))$ boxes for vertical items and at most $|\heightVertical|\leq 1/\delta\eps$ different heights of items. Therefore, we need at most $\countExtraBoxes \in  \bigO(1/(\delta^2\eps^5))$ extra boxes $\boxesVertPlace$.\newline\indent
	We have to place these additional boxes inside the packing area $\stripWidth\times \nicefrac{5}{4}\stripHeight$. We show that this is possible by considering the three possibilities we outlined above. Recall the vertical lines at the borders of boxes in $\boxesTallVert$. These $\countVertLines$ many lines generate at most $\countVertLines+1$ strips. Let $\widthStripsTall$ be the total width of the strips containing items from $\tallItemsHalf$, $\widthStripsBoxes$ the total width of strips containing boxes of height greater than $\nicefrac{3}{4}\stripHeight$ and $\widthStripsRem$ be the total width of the remaining strips. In total, we must have $\widthStripsTall+\widthStripsBoxes+\widthStripsRem=\stripWidth$. Let $|\boxesTallVert| = \countBoxesTallVert$. We can assume $\countBoxesTallVert\leq c_B/(\delta^2\eps), \countStrips\leq c_S/(\delta^2\eps^5), \countExtraBoxes\leq c_F/(\delta^2\eps^5)$ and $\countVertLines \leq c_L/(\delta^2\eps)$ for some constants $c_B,c_S,c_F,c_L\in \mathbb{N}$. At this point, it is necessary to define the function $f$ from \cref{lem:RoundingValues} to find the values $\delta$ and $\mu$ more precisely. We specify $f(\eps)$ by choosing $k\leq (4(c_B+c_F+c_L)c_S).$ Therefore, $\mu\leq \delta^2\eps^{11}/(4(c_B+c_F+c_L)c_S)$ holds.\newline\indent
	Consider the strips that make up the value $\widthStripsRem$, i.e.\ strips without boxes of height greater than $\nicefrac{3}{4}\stripHeight$ or items in $\tallItemsHalf$. These strips can contain boxes with height larger than $\nicefrac{1}{2}\stripHeight$. Therefore, we have a free area with total width of at least $\eps^2\widthStripsRem$ in these strips.\newline\indent
	\begin{claim}
		If $\widthStripsRem\geq \eps^4\stripWidth$, we can place the $\countExtraBoxes$ boxes $\boxesVertPlace$ into the free areas in strips that comprise $\widthStripsRem$.
	\end{claim}
	\begin{proof}
		The considered strips might contain boxes with height larger than $\nicefrac{1}{2}\stripHeight$ but less than $\nicefrac{3}{4}\stripHeight$. Otherwise, they would have been counted in $\widthStripsBoxes$. Therefore, the free area in these strips will be partially used by extra boxes for pseudo items from these boxes. Nevertheless, we know from \cref{lem:ReorderTallVertBox12} that these strips contain free area of width at least $\eps^2\widthStripsRem$ that we can use to place the boxes $\boxesVertPlace$. In each of these at most $\countVertLines+1$ strips the free area is contiguous. However, we have to calculate a small error that might occur: Each of the boxes in $\boxesVertPlace$ has a width of $\mu\stripWidth$ and therefore, there is a residual width of up to $\mu\stripWidth$ where we cannot place a box from $\boxesVertPlace$ in each strip. See \cref{fig:wastedArea} for an intuitive illustration of this issue.\newline\indent
		On the other hand, we can use an area with total width of at least $\eps^2\widthStripsRem-(\countVertLines+1)\mu\stripWidth$ to place the boxes in $\boxesVertPlace$, since there are at most $\countVertLines+1$ strips. Therefore, if $\eps^2\widthStripsRem-(\countVertLines+1)\mu\stripWidth\geq \countExtraBoxes\mu\stripWidth$, we can place all the boxes. Using $\mu:= \delta^2\eps^{11}/(4(c_B+c_F+c_L)c_S)$ it holds that 
		\begin{equation*}
		\countExtraBoxes\mu\stripWidth +(\countVertLines+1)\mu\stripWidth= \mu\stripWidth(c_F/\delta^2\eps^5 +c_L/\delta^2\eps +1)\leq \eps^{11}\stripWidth/\eps^5\leq \eps^6\stripWidth.
		\end{equation*}
		Therefore, if $\widthStripsRem\geq \eps^4\stripWidth$, it holds that $\eps^2\widthStripsRem-(\countVertLines+1)\mu\stripWidth\geq \countExtraBoxes\mu\stripWidth$ and we can place all the boxes $\boxesVertPlace$, proving the claim.
	\end{proof}
\begin{figure}[t]
	\centering
	\resizebox{0.28\textwidth}{!}{
			\begin{tikzpicture}
		\pgfmathsetmacro{\w}{10}
		\pgfmathsetmacro{\h}{10}
		\pgfmathsetmacro{\cTwo}{0.42}
		\pgfmathsetmacro{\cTwoPrime}{0.4}
		\pgfmathsetmacro{\cThree}{0.74}
		\pgfmathsetmacro{\cThreePrime}{0.7}
		
		\draw [ultra thick, color=black] (0,0) --(\w,0) -- (\w,\h) -- (0,\h) -- cycle;
		\draw [dashed, color= tbBgOdd] (0.1*\w,0) -- (0.1*\w,\h);
		\draw [dashed, color= tbBgOdd] (0.2*\w,0) -- (0.2*\w,\h);
		\draw [dashed, color= tbBgOdd] (0.3*\w,0) -- (0.3*\w,\h);
		\draw [dashed, color= tbBgOdd] (0.4*\w,0) -- (0.4*\w,\h);
		\draw [dashed, color= tbBgOdd] (0.5*\w,0) -- (0.5*\w,\h);
		\draw [dashed, color= tbBgOdd] (0.6*\w,0) -- (0.6*\w,\h);
		\draw [dashed, color= tbBgOdd] (0.7*\w,0) -- (0.7*\w,\h);
		\draw [dashed, color= tbBgOdd] (0.8*\w,0) -- (0.8*\w,\h);
		\draw [dashed, color= tbBgOdd] (0.9*\w,0) -- (0.9*\w,\h);
		
		\draw [dashed, color= tbBgOdd] (0.15*\w,0) -- (0.15*\w,\h);
		\draw [dashed, color= tbBgOdd] (0.25*\w,0) -- (0.25*\w,\h);
		\draw [dashed, color= tbBgOdd] (0.35*\w,0) -- (0.35*\w,\h);
		\draw [dashed, color= tbBgOdd] (0.45*\w,0) -- (0.45*\w,\h);
		\draw [dashed, color= tbBgOdd] (0.55*\w,0) -- (0.55*\w,\h);
		\draw [dashed, color= tbBgOdd] (0.65*\w,0) -- (0.65*\w,\h);
		\draw [dashed, color= tbBgOdd] (0.75*\w,0) -- (0.75*\w,\h);
		\draw [dashed, color= tbBgOdd] (0.85*\w,0) -- (0.85*\w,\h);
		\draw [dashed, color= tbBgOdd] (0.95*\w,0) -- (0.95*\w,\h);
		
		\draw [ultra thick, color=black] (\cTwoPrime*\w,0) --(\cTwoPrime*\w,\h); 
		\draw [ultra thick, color=black] (\cThreePrime*\w,0) --(\cThreePrime*\w,\h);
		\draw [ultra thick, color=black] (0.95*\w,0) --(0.95*\w,\h);
		\foreach \x/\y/\xx/\yy/\z in {
			0.0 /0.0 /0.15 /0.28/,
			0.0 /0.28 /0.12/0.45/,
			0.0 /0.45 /0.2 /0.7/,
			0.0 /0.7 /0.23 /0.9/,
			0.15 /0.0 /0.32 /0.28/,
			0.12 /0.28 /0.28/0.45/,
			0.2 /0.45 /0.35 /0.7/,
			0.28 /0.28 /0.37/0.45/,
			\cTwoPrime  / 0.0 /0.50 /0.6/,
			\cTwoPrime / 0.6 /0.555/0.75/,
			0.5  / 0.0 /0.62 /0.6/,
			0.555 / 0.6 /0.58/0.75/,
			0.62  / 0.0 /0.68 /0.6/,
			0.58 / 0.6 /0.62/0.75/,		
			0.62 / 0.6 /0.67/0.75/,	
			\cThreePrime  /0.0 /0.76  /0.35/,
			\cThreePrime   /0.35 /0.77 /0.6/,
			\cThreePrime  /0.6 /0.81 /0.95/,
			0.76  /0.0 /0.87  /0.35/,
			0.77   /0.35 /0.84 /0.6/,
			0.81  /0.6 /0.91 /0.95/,
			0.87  /0.0 /0.94  /0.35/,
			0.84   /0.35 /0.92 /0.6/
		}
		{
			\drawVerticalItem[\small \z]{\x*\w}{\y*\h}{\xx*\w}{\yy*\h};
		}
		\draw[thick,decorate,decoration={brace,amplitude=12pt,mirror}] (0,0) -- (\cTwoPrime*\w,0)node[midway, below=12 pt] {$C_1$};
		\draw[thick,decorate,decoration={brace,amplitude=12pt,mirror}] (\cTwoPrime*\w,0) -- (\cThreePrime*\w,0)node[midway, below=12 pt] {$C_2$};
		\draw[thick,decorate,decoration={brace,amplitude=12pt,mirror}] (\cThreePrime*\w,0) -- (0.95*\w,0)node[midway, below=12 pt] {$C_3$};
		
		\draw[pattern=north west lines] (0,0.9*\h) rectangle (\cTwoPrime*\w,\h);
		\draw[pattern=north west lines] (\cTwoPrime*\w,0.75*\h) rectangle (\cThreePrime*\w,\h);
		\draw[pattern=north west lines] (\cThreePrime*\w,0.95*\h) rectangle (0.95*\w,\h);
		\draw[pattern=north west lines] (0.95*\w,0) rectangle (\w,\h);
		\end{tikzpicture}
	}
\caption{An illustration of the placement of the strips containing the additional strips of vertical items generated by \cref{lem:PlacementVertSep}. As you can see, we may fail to use a width of up to $\mu\stripWidth$ to place these items for each area we find to place these items in.}
\label{fig:wastedArea}
\end{figure}
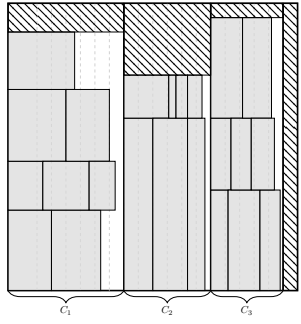
	With this option proven to be feasible, let us now focus our attention on another possibility to place the boxes in $\boxesVertPlace$, depending on the total width of strips containing tall items overlapping the horizontal line at $\nicefrac{1}{2}\stripHeight$. Since we have already shown how to finish the placement of the new boxes if there is ample space in $\widthStripsRem$, we know that this is not the case.
	\begin{claim}
		If $\widthStripsBoxes<\eps^6\stripWidth/(4c_S)$ and $\widthStripsRem< \eps^4\stripWidth$, we can place all the boxes for vertical items inside the boxes of height at least $\nicefrac{3}{4}\stripHeight$.
	\end{claim}
	\begin{proof}
		Due to the definition of $\widthStripsTall$, we know that $\widthStripsBoxes=\stripWidth-(\widthStripsTall+\widthStripsRem)>(1-2\eps^4)\stripWidth\geq \eps\stripWidth$. Additionally, each tall item not in $\tallItemsHalf$ crossing $\nicefrac{1}{2}\stripHeight$ has a width of at most $\widthStripsTall\cdot\delta^2\eps<\eps^7\delta^2\stripWidth/4c_S:=w_{max}$. This is because there would not be enough space left inside $\widthStripsBoxes$ if these items were any wider. After reordering the inside of the boxes, there are $\countStrips$ strips total inside the boxes. We want to determine the total height of free area inside a strip. This area may not be contiguous, as there can be an item in the middle of any of these strips, with free area above and below it. However, we can address this issue trivially by slicing this item downwards, leaving only one contiguous area of free space. In the shifting step, we have added a total area of $\widthStripsBoxes(\nicefrac{1}{4}\stripHeight+\eps\opt)$ to all of these strips. Let $\Breve{\widthStripsBoxes}$ be the total width of the strips containing free area with total height less than $\nicefrac{1}{4}\stripHeight$. Since the boxes in $\boxesVertPlace$ have a height of up to $\nicefrac{1}{4}\stripHeight$, we want to use areas with at least that height to place them. Let $\hat{\widthStripsBoxes}$ be the total width of these strips. Clearly, we have $\Breve{\widthStripsBoxes}+hat{\widthStripsBoxes}=\widthStripsBoxes$. We have to show the following to be able to place the extra boxes as desired:\newline\indent
		\emph{It holds that $\hat{\widthStripsBoxes}\geq \eps \widthStripsBoxes$.}
		\newline\indent
		The total free area can have a height of at most $\nicefrac{3}{4}\eps\opt$ inside each strip, because there either are boxes at the top and bottom of height at least $\nicefrac{1}{4}\stripHeight$, or one box of height at least $\nicefrac{3}{4}\stripHeight$ at the bottom. Therefore, we have\begin{equation*}
		\nicefrac{1}{4}\cdot\Breve{\widthStripsBoxes}+(\nicefrac{3}{4}\stripHeight+\eps\opt)\cdot\hat{\widthStripsBoxes}\geq \widthStripsBoxes(\nicefrac{1}{4}\stripHeight+\eps\opt)
		\end{equation*}
		since the free area in $\Breve{\widthStripsBoxes}$ has a total height of at most $\nicefrac{1}{4}\stripHeight$, the free area in $\hat{\widthStripsBoxes}$ has a height of at most $(\nicefrac{3}{4}\stripHeight+\eps\opt)$ and the total free area is bounded by $\widthStripsBoxes(\nicefrac{1}{4}\stripHeight+\eps\opt)$.\newline\indent
		As a consequence, we can prove that $\hat{\widthStripsBoxes}$ has a sufficient minimum size. The following holds:
		\begin{align*}
		&\widthStripsBoxes(\nicefrac{1}{4}\stripHeight+\eps\opt)\leq \nicefrac{1}{4}\stripHeight\cdot\Breve{\widthStripsBoxes}+(\nicefrac{3}{4}\stripHeight+\eps\opt)\cdot\hat{\widthStripsBoxes} \\
		&=\nicefrac{1}{4}\stripHeight\cdot\widthStripsBoxes+(\nicefrac{1}{2}\stripHeight+\eps\opt)\hat{\widthStripsBoxes}\\
		&=\nicefrac{1}{4}\stripHeight\cdot\widthStripsBoxes+((1+2\eps)\opt/2+\eps\opt)\cdot\hat{\widthStripsBoxes} \\
		&= \nicefrac{1}{4}\stripHeight\cdot\widthStripsBoxes+((1+4\eps)\opt/2)\cdot\hat{\widthStripsBoxes}.
		\end{align*}
		This allows us to deduce that 
		\begin{equation*}
		\eps\widthStripsBoxes\leq ((1+4\eps)/2)\cdot\hat{\widthStripsBoxes}.
		\end{equation*}
		Therefore,
		\begin{align*}
		\hat{\widthStripsBoxes}&\geq2\eps\widthStripsBoxes/(1+4\eps)\geq\eps\widthStripsBoxes, &\forall 0< \eps\leq \nicefrac{1}{4},
		\end{align*}
		which proves the claim.\newline\indent
		As a direct consequence of this claim, strips with total width of at least $\eps\widthStripsBoxes$ contain free area with total height of at least $\nicefrac{1}{4}\stripHeight$. We have already alluded to the fact that, even when the area is split by an item in the middle, we can fuse these areas together trivially.\newline\indent
		In each strip, there is an area with width of at most $\mu\stripWidth$ that we cannot use to place the boxes. Therefore, we can place all boxes for previously separated vertical items if $\eps\widthStripsBoxes-2w_{max}\countStrips-\mu\stripWidth\countStrips\geq \mu \stripWidth\countExtraBoxes$. It holds that 
		\begin{align*}
		&2w_{max} \countStrips+\mu\stripWidth\countStrips+\mu \stripWidth\countExtraBoxes\\
		&\leq (\eps^7\delta^2\stripWidth/(2c_S))(c_S/\delta^2\eps^5)+\mu\stripWidth(c_S/\delta^2\eps^5+c_F/\delta^2\eps^5)\\
		&\leq \eps^2\stripWidth/2+\eps^6\stripWidth\leq \eps^2\stripWidth.
		\end{align*}
		Thus, if $\widthStripsBoxes\geq\eps\stripWidth$, we know that $\eps\widthStripsBoxes-2w_{max}\countStrips-\mu\stripWidth\countStrips\geq\mu\stripWidth\countExtraBoxes$ and we can place all boxes in the designated areas for this case.
	\end{proof}
	In this step, we create at most $2\countStrips\in\bigO(1/(\eps^5\delta^2))$ new boxes for tall items and no new box for vertical items. The boxes for tall items already contain only tall items of the same height. Therefore, we introduce at most $\bigO(1/(\eps^5\delta^2))$ boxes for vertical items $\boxesVert$ such that each box $B\in \boxesVert$ contains only items with height $\height(B)$.\newline\indent
	\emph{The boxes for small items.} The free area inside the boxes from the partition in \cref{lem:boxPartition} for horizontal, tall and vertical items is at least as large as the total area for the small items, because the small items where contained in the optimal packing.\newline\indent
	\emph{Bounding the packing height.} To bound the total height of the packing we summarize which heights we added during the process so far. In the beginning, we had the optimal packing of height $\opt$. Through our rounding with \cref{lem:Rounding}, we generated a packing of height $(1+2\eps)\opt$. When we extended boxes containing tall and vertical items of height at least $\nicefrac{3}{4}\stripHeight$, we added $\nicefrac{1}{4}\stripHeight+\eps\opt\leq \nicefrac{1}{4}(1+2\eps)\opt+\eps\opt$ to the packing height. Earlier in this proof, to ensure that all boxes with height greater than $\nicefrac{1}{2}\stripHeight$ start and end at multiples of $\eps^2\opt$ we shifted some vertical items. This procedure added a further $2(\eps+\eps^2)\opt$ to the packing height. Thus, in total, we have added at most $\nicefrac{1}{4}(1+2\eps)\opt+2(\eps+\eps^2)\opt\leq (\nicefrac{1}{4}+3\eps)\opt$ to the packing height. This assumes sufficiently small, i.e.\ smaller than $\nicefrac{1}{2}$ values for $\eps$. As such, the structured packing has a height of at most $(\nicefrac{5}{4}+5\eps)\opt$.
\end{proof}

We employ a dynamic program to iterate through all possible configurations of these boxes until we find a feasible packing of the items into the boxes.
\paragraph*{Step 6} 
We utilize some techniques developed in~\cite{StripPacking54} to place the small and medium items without exceeding the desired height of the packing.
\begin{lem}
	It is possible to place the small items inside the boxes $\boxesSmallVert$ and $\boxesHorSmall$ and one extra box with width $\stripWidth$ and height at most $2\eps\delta^6\opt.$ 
	\label{lem:PlaceSmallItems}
\end{lem}
\begin{proof}
	This is an adaptation of a similar procedure in \cite{StripPacking54}, considering the slicing of items.
	Recall how we started partitioning the original optimal packing into the boxes in \cref{lem:boxPartition}. Since this optimal packing contained the small items as well as the items we chose as the basis of our partitioning, the total remaining area inside those boxes must be enough to place these small items. In \cref{lem:PlacementVertSep,lem:PlacementHorItems} we generate at most $\bigO(1/(\eps^5\delta^2))$ empty boxes which we can use to place the small items. We have shown that the total area of those empty boxes is at least as large as the empty space left in the original boxes $\boxesHor$ and $\boxesTallVert$ in \cref{lem:boxPartition}. As $\boxesLarge$ only ever contain a single item, there is no empty space in these boxes. Therefore, the total area in those empty boxes is at least as large as the total area of small items.\newline\indent
	Let $\boxesSmall$ be the set of boxes for small items and $|\boxesSmall|=c/\delta^2\eps^5$ for some constant $c\in \mathbb{N}$. We show that we only need one small extra box to place all the small items into $\boxesSmall$ using the Next Fit Decreasing Height~(NFDH) algorithm\cite{CoffmanGJT80}.\newline\indent
	First, we discard any box that is too small in either dimension, i.e.\ has a height of less than $\mu\opt$ or a width of $\mu\stripWidth$. Clearly, the total area of discarded boxes is at most $\mu\stripWidth\opt$. Consider a box $B$ with height and width greater than $\mu\opt$ and $\mu\stripWidth$ respectively. The last item we want to place on each shelf during the NFDH-Algorithm might not fit from its width. As such, we may have an empty space of at most $\mu\stripWidth$ on each shelf, as all small items have at most that width. Furthermore, the final shelf has a distance of at most $\mu\opt$ to the top border of $B$. Again, this area might remain free if no suitable items are left. Finally, the free area between two shelves used in the algorithm has a size of at most $\mu\opt\cdot w(B)$, again because all small items have at most that size. As a result, the total free area inside the box $B$ is at most $\mu\stripWidth\cdot\height(B)+2\mu\opt\cdot\width(B)\leq 3\mu\stripWidth\opt$. We have this free area for every box $B\in\boxesSmall$. Therefore, the total area of items that could not be placed inside boxes is at most $3\mu\stripWidth\opt\cdot c/\delta^2\eps^5.$ Since we have defined $\mu\leq\eps^{11}\delta^2/k$ for some suitable constant $k$ through \cref{lem:RoundingValues} it holds that $3\mu\stripWidth\opt\cdot c/\delta^2\eps^5\leq \eps^6\stripWidth\opt$ when choosing $k\geq c$. \newline\indent
	These remaining items can then be placed into a box of width $\stripWidth$ and height $2\eps^6\opt$ using Steinbergs algorithm \cite{Steinberg97}. The relevant conditions to apply this algorithm are met because each item has a height of at most $\mu\opt$.
\end{proof}
Having placed the small items, there is only one set of items left to place before we have successfully shown that our partition into boxes allows us to generate a structured packing while only adding a height of $(\nicefrac{1}{4})\opt+ \bigO(\eps\opt)$ to the top of the packing. Thus, we now show that the medium items, of which there are few, can be placed in a structured manner.

\begin{lem}\cite{StripPacking54}
	It is possible to place the medium items $M$ into a box of width $\stripWidth$ and height at most $2\eps\opt$.
	\label{lem:placeMediumItems}
\end{lem}
\begin{proof}
	We know, from \cref{lem:RoundingValues}, that the total area of medium items is small. Thus, we use a well defined algorithm to place these items at an approximately greater loss. As we did with the small items, we use the NFDH algorithm for this purpose \cite{CoffmanGJT80}. Before we apply the algorithm, we sort all items by their width, in descending manner. We know that $a(M)$ is bounded by $\eps^{13}\stripWidth\opt$ and the widest item in $M$ has a width of at most $\eps\opt$. Therefore, we can place these items in a box of width $\stripWidth$ and height of at most $2\eps^{13}\opt+\eps\opt\leq 2\eps\opt$.
\end{proof}
\paragraph*{Step 7} 
The algorithm iterates over the previous steps until the minimal value for $\stripHeight'$ has been found and returns the associated packing. 
Because we have shown that any optimal packing of height $\opt$ can be transformed into a structured packing of height $(\nicefrac{5}{4}+\eps)\opt$, we know that $\stripHeight'\leq \opt$ must hold. 
Therefore, our algorithm generates a solution with height at most the desired approximation ratio.
\paragraph*{Running time:} 
Scaling the instance can be done in $\Oh(n\log(n))$. 
The binary search in step $2$ runs logarithmically in its step-length. 
Due to our rounding, this length is $n/\eps$. 
Therefore, the search can be computed in $\mathcal{O}(\log(n/\eps))$. 
For every value of the search, we have to guess the starting points for horizontal items. 
As there are $\Oh_{\eps}(1)$ of these boxes, we can guess these in $\stripWidth^{\Oh_{\eps}(1)}$ time. 
Then, we guess all possible configurations of boxes inside the packing. 
Again, there are $\Oh_{\eps}(1)$ many of those, yielding a running time of $\stripWidth^{\Oh_{\eps}(1)}$ for this step. 
As a result, the running time can be expressed as $\Oh(n\log(n))+\mathcal{O}(\log(n/\eps))\cdot\stripWidth^{\Oh_{\eps}(1)}= \Oh(n\log(n))\cdot\stripWidth^{\Oh_{\eps}(1)}.$
This completes the proof of \cref{thm:BigTheorem}.

\section{Conclusion}
We showed that the \SSP{} problem remains hard to solve with an approximation ratio lower than $(\nicefrac{5}{4})\cdot \opt$, even 
in pseudo-polynomial time. 
In the process of this proof, we developed a transformation algorithm showing that we can treat \ari{DSP} and \ari{PTS} as somewhat dual problems of one another. 
We provided a framework using known results to solve either of these problems optimally when admitting some resource augmentation.
These results are valuable 
in applications where 
penalty costs for exceeding makespan or power demand are large. 

Finally, we have developed an approximation algorithm for DSP that runs in pseudo-polynomial time. 
The adaptations made bridge the known integrality gap between SP and DSP. 
Furthermore, its approximation ratio of $(\nicefrac{5}{4}+\eps)\cdot \opt$ almost matches the proven hardness except for a 
negligibly 
small $\eps$. 
As a contrast, in polynomial time, the gap between proven hardness  $(\nicefrac{3}{2})$~\cite{yaw2014peak} and best-known approximation $(\nicefrac{5}{3}+\eps)$~\cite{DeppertJ0RT21,GalvezGJK2021} remains quite large.

Similar results for the SP algorithm suggest that the presented algorithm can be adapted to admit 90\textdegree{} item rotations or moldable jobs. 
Both of these settings are interesting in the context of smart grids. 
Rotations can model different modes of power consumption, i.e.\ fast versus slow charging. 
Moldable jobs can model the parallel operation of multiple machines for the same task. 
Each variation on the job/item represents a set number of machines assigned to that task.

Further open questions extend to the relation of DSP to other packing problems.
A family of instances exists where the height of optimal solutions for SP and DSP differs by 5/4~\cite{BladekDGS15}. 
Furthermore, due to Steinbergs' algorithm, this difference is bounded by 2. 
Closing this gap might imply better bounds on optimal solutions for either problem. 

On the other hand, the discovered link between DSP and PTS may suggest similar links between DSP and other resource allocation problems. 
Amongst these, the well-known \textsc{Storage Allocation Problem} (SAP)~\cite{MomkeW20} and \textsc{Unsplittable Flow on a Path} (UFP)~\cite{GrandoniMW22} problems are promising candidates. 
So far, a PTAS is known only for the latter~\cite{0001MW22}. 
Techniques used in the presented algorithm may be helpful in adapting existing algorithms for SAP to UFP, or vice versa. 
Furthermore, similar resource augmentation arguments as we made for PTS and DSP may be made for either problem should a link between DSP and UFP or SAP exist. 
Finally, techniques used in the presented approximation scheme may be useful when developing an approximation scheme for SAP. 
    \bibliography{FullPaper}
    \end{document}